%% file: COSTA_formal_arxiv.tex
\documentclass[11pt]{article}
\usepackage[margin=1in]{geometry}
\usepackage[T1]{fontenc}
\usepackage[utf8]{inputenc}
\usepackage{amsmath,amssymb,amsthm,mathtools,bm}
\usepackage{booktabs}
\usepackage{array}
\usepackage{enumitem}
\usepackage{algorithm}
\usepackage{algorithmic}
\usepackage{graphicx}
\usepackage{subcaption}
\usepackage{float}
\usepackage{placeins}
\usepackage{microtype}
\usepackage{xurl}
\usepackage{xcolor}
\usepackage[round,authoryear]{natbib}
\definecolor{mydarkblue}{rgb}{0,0.08,0.45}
\usepackage[colorlinks=true,linkcolor=mydarkblue,citecolor=mydarkblue,urlcolor=mydarkblue]{hyperref}
\usepackage{fontawesome5}

\setlist{nosep,leftmargin=2.0em}


\theoremstyle{plain}
\newtheorem{theorem}{Theorem}[section]
\newtheorem{proposition}[theorem]{Proposition}
\newtheorem{lemma}[theorem]{Lemma}
\newtheorem{assumption}[theorem]{Assumption}
\newtheorem{corollary}[theorem]{Corollary}
\theoremstyle{definition}
\newtheorem{definition}[theorem]{Definition}
\theoremstyle{remark}
\newtheorem{remark}[theorem]{Remark}
\newtheorem{example}[theorem]{Example}

\DeclareMathOperator{\Cov}{Cov}
\DeclareMathOperator{\Var}{Var}
\DeclareMathOperator{\E}{\mathbb{E}}
\DeclareMathOperator{\tr}{tr}
\DeclareMathOperator{\Poisson}{Poisson}
\newcommand{\one}{\mathbf{1}}
\newcommand{\ind}{\mathbb{I}}
\newcommand{\Zcal}{\mathcal{Z}}
\newcommand{\Ocal}{\mathcal{O}}
\newcommand{\Rnorm}{\mathcal{R}}

\newcommand{\comp}{\mathrm{net}}
\newcommand{\tp}{t_p}
\newcommand{\Gp}{\mathcal{G}_p}
\newcommand{\dlat}{d_{\mathrm{lat}}}
\newcommand{\dI}{d_I}
\newcommand{\dT}{d_T}
\newcommand{\lag}{\mathrm{lag}}
\newcommand{\COSTA}{COSTA}
\newcommand{\AI}{A_I}

\title{
{\bfseries COSTA: Covariance-Optimized Design and Causal Inference under Network-Temporal Interference}
}
\author{%
Qianyi Chen\textsuperscript{1} \quad
Bo Li\textsuperscript{1}\thanks{Corresponding author.} \quad
Yongli Qin\textsuperscript{2} \quad
Jinyong Ma\textsuperscript{2}\\[0.35em]
\textsuperscript{1}\,School of Economics and Management, Tsinghua University\\
\textsuperscript{2}\,ByteDance\\[0.25em]
\texttt{cqy22@mails.tsinghua.edu.cn} \quad \texttt{libo@sem.tsinghua.edu.cn}\\
\texttt{yongli.qin@bytedance.com} \quad \texttt{jinyongma@bytedance.com}%
}
\date{}
\begin{document}
\maketitle

\begin{abstract}
Experiments on networks observed over time face network spillovers, temporal carryover, and dependence deliberately introduced by the design. We propose \COSTA{}---\emph{Covariance-Optimized Spatiotemporal Treatment Allocation}---a joint Bernoulli design for unit--time assignments. Under common treatment marginals and a nonnegative linear network--temporal exposure model, Horvitz--Thompson bias for the sustained all-treated versus all-control contrast is exactly the negative expected weight of an assignment cut. A covariance-level variance envelope yields an MSE bound that can be optimized directly over assignment covariance. To scale this design, we introduce a thresholded-Gaussian Kronecker parameterization that mirrors the network and temporal exposure operators while preserving valid Bernoulli marginals.

We next develop inference theory for the joint effects of designed treatment dependence and interference-induced outcome dependence. Canonical correlations between latent blocks generating separated HT contributions supply the coefficients required by graph-$\psi$ central limit and network-HAC theory; a spectral-floor and far-row-mass condition gives a primitive sufficient check. The framework covers sparse, block, Kronecker, locally factored, and other structured covariance sequences satisfying these conditions. Semi-synthetic RetailRocket and MovieLens experiments show substantial default-setting RMSE reductions and well-calibrated model-assisted design-centered intervals across linear, nonlinear, and demand-substitution outcome surfaces.
\end{abstract}

\section{Introduction}
\label{sec:introduction}

Estimating a global average treatment effect (GATE) from one cross-sectional network experiment is intrinsically difficult. The target compares coherent global treatment with coherent global control, whereas independent randomization exposes most units to mixed local treatment environments. Unless interference is tightly restricted or the design creates coherent exposure, an estimator may be precise around a design expectation that remains far from the global contrast \citep{aronowsamii2017,leung2022ani,chen2023ocd}. Repeated experiments provide additional leverage, including variation across treatment proportions \citep{chen2024justrampup}, but temporal replication is not free: assignments may affect later outcomes through persistence, learning, inventory, congestion, or delayed response. Switchback and balanced multi-unit designs address important temporal cases \citep{bojinov2023switchback,missault2025rbsd}; general repeated network experiments must handle spatial spillovers and temporal interference jointly.

We index an experimental cell by $a=(i,b)$, where $i$ is a network or spatial unit and $b$ is a temporal block. The design chooses a common-marginal joint Bernoulli law for $Z=(Z_{ib})$ through its normalized covariance $\Rnorm=\Cov(Z)/\{p(1-p)\}$. Under a linear network--temporal working model, the design-induced bias of a baseline-adjusted Horvitz--Thompson (HT) estimator is an exact exposure-weighted function of $\Rnorm-\one\one^\top$. Positive covariance on exposure-relevant network and lag edges makes observed environments resemble global treatment or control. Positive covariance everywhere, however, approaches a rollout and eliminates experimental contrast. \COSTA{} optimizes this tension using a thresholded-Gaussian covariance model, soft balance and switching constraints, and an optional pairwise shrinkage penalty. The theory targets general network--temporal exposure operators; the semi-synthetic study later instantiates it in item-side multiple-unit switchbacks as one empirically important case.

The inferential problem has two distinct sources of dependence. First, covariance design intentionally correlates treatment indicators, possibly beyond direct exposure edges. Second, network and temporal interference make each outcome contribution a function of a neighborhood of assignments, so contributions may share treatment inputs even under independent Bernoulli randomization. These mechanisms compound under covariance optimization: nonlinear local contributions are evaluated on a correlated assignment field. CLT and HAC theory must therefore control dependence after this field passes through the interference map, rather than treatment covariance or outcome locality alone.

The paper separates three questions:
\[
\mathrm{IC\mbox{-}0}:\quad \frac{\widehat\tau-\mathbb E_D\widehat\tau}{\sqrt V}\Rightarrow N(0,1),\qquad
\mathrm{IC\mbox{-}1}:\quad \frac{|\mathbb E_D\widehat\tau-\tau|}{\sqrt V}\to0,
\]
\[
\mathrm{IC\mbox{-}2}:\quad \widehat V/V\to1.
\]
These requirements are distinct: IC-0 establishes design-centered normality, IC-1 validates the causal GATE as the center, and IC-2 makes studentization feasible from observed data.

Our contributions are fourfold.
\begin{enumerate}[label=(\roman*),leftmargin=1.8em,itemsep=0.20em,topsep=0.25em]
\item Building on the cross-sectional $p=1/2$ covariance-bias calculation of \citet{chen2023ocd}, we derive a weighted common-$p$ bias identity for arbitrary joint Bernoulli network--temporal designs. Its expected-cut representation yields componentwise zero-bias collapse and path and general-graph lower bounds, including the lag-bias cost of required switching.
\item Under the linear working model, a covariance-level variance envelope combines with exact bias to give an optimizable MSE upper bound. A causal-structure-driven thresholded-Gaussian Kronecker parameterization makes the otherwise high-dimensional design scalable by reusing unit factors across blocks and temporal factors across units.
\item We develop CLTs, causal-centering conditions, and variance estimators for the two-layer dependence induced by interference and deliberately correlated treatments. Gaussian block canonical correlation links structured assignment covariance to graph-$\psi$ weak dependence; heterogeneous-mean HAC and design-aware covariance-tail correction provide routes to feasible studentization under structured dependence.
\item We systematically evaluate multiple-unit switchbacks on RetailRocket and MovieLens. Point-estimation results span linear, nonlinear, and demand-substitution outcomes, while a separate inference study reports design-centered and raw-GATE coverage, recentered GATE coverage, interval width, and standard-error calibration.
\end{enumerate}

\subsection{Related work and positioning}
\label{subsec:related_work}

Design-based interference inference uses exposure mappings, inverse weighting, clustering, and local-dependence asymptotics \citep{hudgens2008toward,aronowsamii2017,athey2018exact,liu2014large,liu2016ipw}. Graph-cluster and optimized network designs instead seek coherent exposure or favorable dependence \citep{ugander2013graph,basseairoldi2018model,jagadeesan2020designs}. The cross-sectional $p=1/2$ covariance-MSE formulation of \citet{chen2023ocd} is our direct design precursor; we extend it to common $p$, weighted network--time exposure, cut-based lower bounds, Kronecker scaling, and inference under deliberately dependent assignments. Cluster-scale MSE criteria, approximate-neighborhood asymptotics, and graph-cut surrogates provide complementary bias--variance analyses under spatial or network decay \citep{leung2022aos,leung2022ani,liwager2022random,zhu2025causalgraphcut}.

Repeated-randomization work studies switchbacks, panels, ramp-up sequences, and multiple-randomization designs \citep{bojinov2023switchback,ni2023panel,masoero2024multiple,missault2025rbsd,chen2024justrampup}. These methods show how time creates additional design variation, but network spillovers and temporal carryover must be handled jointly when connected units are repeatedly assigned.

Dependent randomization creates a separate inferential problem even without interference; correlated Gram--Schmidt assignment, for example, requires its own CLT and variance estimator \citep{harshaw2024gramschmidt}. Interference acts downstream by turning each outcome or HT contribution into a function of nearby assignments. General graph-indexed weak-dependence theory then links neighborhood growth and dependence decay to CLTs and HAC consistency \citep{doukhanlouhichi1999,jenishprucha2012,kojevnikov2021network}. We use the graph-$\psi$ results of \citet{kojevnikov2021network} as the probabilistic engine and supply the missing design-side bridge for covariance-engineered Gaussian assignments.

Recent work further clarifies adjacent but distinct problems. Randomized graph clusters and conflict-graph methods improve exposure probabilities or manage cross-cluster spillovers \citep{uganderyin2023rgcr,leung2025crosscluster,kandiros2024conflict}. Finite-sample variance work develops estimable conservative quadratic bounds under interference and complex designs \citep{harshaw2026optimizedvariance}, whereas our IC-2 analysis asks when local or model-assisted estimators are ratio-consistent for the design-centered variance. Network jackknife, exposure-agnostic procedures, null-exposure tests, and alternative estimands provide complementary tools when the sustained global contrast is difficult to recover \citep{parkwager2026jackknife,he2026exposureagnostic,puelz2022graph,choi2024estimands}. COSTA links an exact arbitrary-law bias representation to contribution dependence for deliberately correlated network--time assignments.

A broader inference literature relaxes knowledge of the interference graph or replaces an exact exposure map by bounded-neighborhood, spatial, policy-learning, message-passing, or design-based testing structures \citep{ma2021gnn,shirani2024message,wang2025spatialunknown,lu2026bounded,park2026spillovertesting}. These approaches target uncertainty about the interference mechanism itself. Our main problem is different but complementary: the network--temporal exposure architecture is specified for design, while the treatment law is deliberately correlated and must be carried through to the dependence of the resulting contribution field.

On the design side, staggered rollouts can identify effects under interference without a known network, and recent clustering or robust covariance criteria explicitly balance interference, homophily, and design misspecification \citep{cortez2022staggered,viviano2026causalclustering,thiyageswaran2026robustcovariance}. COSTA instead starts from a target-specific covariance-bias identity, jointly optimizes network and temporal exposure, and asks which structured covariance sequences support valid design-centered and causal inference.

\paragraph{Organization.}
Sections~\ref{sec:setting}--\ref{sec:param_opt} define the experiment, derive the finite-sample design identities, and construct the covariance parameterization. Section~\ref{sec:inference} develops inference for covariance-structured assignments and separates centered normality, causal centering, and feasible variance estimation. Section~\ref{sec:experiments} reports point-estimation and inference evidence. Proofs, implementation details, and secondary extensions are collected after the main article.

\section{Setting and Notation}
\label{sec:setting}

\subsection{Experimental environment}

We consider experiments with $N$ eligible units indexed by $i\in[N]=\{1,\ldots,N\}$ and $B$ temporal blocks indexed by $b\in[B]=\{1,\ldots,B\}$. A unit may be a person, item, seller, route, region, creator, school, or any other entity on which treatment can be assigned. The experimenter observes or constructs a pre-treatment network that summarizes likely cross-unit spillover. Time is represented by design blocks rather than by an asymptotically long time series. The formulation allows any finite block horizon, and the asymptotic theory states explicitly when $B_N$ may vary with $N$.

Let $\Zcal=[N]\times[B]$ be the set of cells whose treatment assignments are randomized. We reserve $Z$ for the random treatment array and write $a=(i,b)$ and $c=(j,r)$ for generic unit--block cells. The total number of assignment cells is $L=NB$. With one-lag temporal carryover, the first block is used only to generate lagged treatment exposure for block 2. Outcomes are therefore evaluated on the post-burn-in set $\Ocal=[N]\times\{2,\ldots,B\}$, with $M=|\Ocal|=N(B-1)$. Block 1 is still randomized, but it is not included in the estimand; this avoids defining a block-1 outcome in terms of an unrandomized pre-experiment state $Z_{i0}$. If temporal carryover is absent, one may instead evaluate all randomized cells and set $\Ocal=\Zcal$.

\subsection{Treatment probability and covariance notation}

For each assignment cell,
\[
    Z_a=Z_{ib}\in\{0,1\},\qquad \E[Z_a]=p,
\]
where $p\in(0,1)$ is the planned treatment share. Write
\[
    q=1-p,
    \qquad
    v=pq.
\]
The raw treatment covariance matrix is $\Sigma=\Cov(Z)\in\mathbb R^{L\times L}$. We normalize covariance entrywise as
\[
    \Rnorm_{ac}=\frac{\Cov(Z_a,Z_c)}{p(1-p)}=\frac{\Sigma_{ac}}{v}.
\]
Thus $\Rnorm_{aa}=1$. When $p=1/2$, $\Rnorm$ is the covariance of the centered sign treatment $2Z_a-1$; when $p\ne1/2$, it is the Bernoulli-$p$ normalized covariance. This normalized covariance is the design object that enters the bias formulas below.

The main text works directly with unit--block potential outcomes. Appendix~\ref{app:primitive_observations} records how primitive user--item or opportunity--unit observations can be aggregated into this unit--block layer in the recommendation-system validation setting.

\subsection{Network-time exposure graph}

Let $\AI=(A_{ij})_{i,j\in[N]}$ be a weighted adjacency matrix with $A_{ij}\ge0$ and $A_{ii}=0$ by default. The graph may be constructed from social links, spatial adjacency, mobility flows, demand overlap, semantic similarity, collaborative-filtering embeddings, category substitution, query overlap, or other pre-treatment information. In recommendation experiments we call this a competition graph, but the method only requires an exposure graph.

The network-time design uses two exposure matrices indexed by assignment cells. Rows outside $\Ocal$ are set to zero. The contemporaneous network-spillover matrix is
\begin{equation}
    H^{\comp}_{(i,b),(j,r)}=A_{ij}\ind\{b\ge2,\ r=b,\ j\ne i\}.
    \label{eq:hcomp}
\end{equation}
The one-lag temporal matrix is
\begin{equation}
    H^{\lag}_{(i,b),(j,r)}=\ind\{b\ge2,\ j=i,\ r=b-1\}.
    \label{eq:hlag}
\end{equation}
Thus neighboring units in the same block may affect each other, and the previous assignment of the same unit may affect its current outcome. The symmetrized exposure graph is used only for locality and HAC calculations: it is a network-time dependency graph for asymptotic and variance-estimation arguments, not an application-specific semantic restriction.
\section{Potential Outcomes, Estimand, and Estimator}
\label{sec:model_estimator}

\subsection{Working potential outcome model}

For $a=(i,b)\in\Ocal$, let the structural potential-outcome schedule be
\begin{equation}
    Y_{ib}^{\star}(z)=\mu_{ib}+\beta_{ib}z_{ib}
    +\gamma\sum_{j\ne i}A_{ij}z_{jb}
    +\eta z_{i,b-1},
    \label{eq:item_model}
\end{equation}
and let the observed outcome be
\begin{equation}
    Y_{ib}^{\mathrm{obs}}=Y_{ib}^{\star}(Z)+\varepsilon_{ib}.
    \label{eq:observed_outcome_noise}
\end{equation}
The disturbance $\varepsilon_{ib}$ is not part of the causal schedule or the GATE. In a finite-population randomization analysis it may be fixed before assignment, in which case all results are conditional on its realization. If it is treated as random measurement noise, every bias statement below assumes at least $\E\{\psi_{ib}(Z_{ib})\varepsilon_{ib}\}=0$; the convenient sufficient condition is $\E(\varepsilon_{ib}\mid Z)=0$. Variance and primitive dependence results impose their additional noise conditions explicitly.

The structural equation is a design model rather than a claim that real platform outcomes are exactly linear. Its role is to expose which entries of the assignment covariance matrix determine design bias and to motivate the envelope-driven variance surrogate used in the optimization objective. Here $\mu_{ib}$ is the structural baseline outcome, $\beta_{ib}$ is the direct treatment effect, $\gamma$ is contemporaneous network spillover interference, and $\eta$ is one-lag own-unit temporal carryover. The temporal term is intentionally minimal: previous treatment affects the next outcome for the same unit only. Section~\ref{sec:experiments} later evaluates the resulting design under this model and two misspecified data-generating processes.

Equivalently, the structural schedule is
\begin{equation}
    Y_a^{\star}(z)=\mu_a+\beta_a z_a+
    \gamma\sum_cH^{\comp}_{ac}z_c+
    \eta\sum_cH^{\lag}_{ac}z_c,
    \qquad a\in\Ocal.
    \label{eq:two_H_model}
\end{equation}
For derivations, we use the generic one-component shorthand
\begin{equation}
    Y_a^{\star}(z)=\mu_a+\beta_a z_a+\theta\sum_cH_{ac}z_c,
    \qquad a\in\Ocal,
    \label{eq:single_H_model}
\end{equation}
where $H\ge0$ has zero rows outside $\Ocal$. Equation~\eqref{eq:single_H_model} only rescales the exposure component. If $\gamma>0$ and $\eta\ge0$, set
\begin{equation}
    \theta=\gamma,\qquad
    \omega_{\mathrm{rel}}=\eta/\gamma,\qquad
    H=H^{\comp}+\omega_{\mathrm{rel}}H^{\lag}.
    \label{eq:relative_reweighting}
\end{equation}
Then $\theta H=\gamma H^{\comp}+\eta H^{\lag}$, so the one-component notation exactly reproduces \eqref{eq:two_H_model}. More generally, for any chosen scale $\theta>0$ one can write $H=(\gamma/\theta)H^{\comp}+(\eta/\theta)H^{\lag}$ whenever the two coefficients are nonnegative. Thus $\omega_{\mathrm{rel}}$ is not a new causal parameter; it is the relative carryover-to-network effect size used to put the temporal and network exposure operators on a common scale. Because $\gamma$ and $\eta$ are rarely known at design time, the COSTA objective below keeps the network and temporal bias proxies separate and lets their weights be tuned. If either exposure effect can be negative, one should keep $H^{\comp}$ and $H^{\lag}$ separate, decompose the signed component into positive and negative parts before applying a nonnegative envelope, or define the bias proxy directly for the signed operator.

\begin{remark}[What the linear working model does and does not claim]
\label{rem:misspecified_outcomes}
The covariance identity is exact for the linear exposure component used in \eqref{eq:single_H_model}. If the true potential outcome surface can be decomposed as
\[
    Y_a^{\star}(z)=\mu_a+\beta_a z_a+\theta H_a^\top z+r_a(z),
\]
then COSTA controls the design bias contributed by \(\theta H_a^\top z\), while the remaining bias contains the residual term involving \(r_a(z)\). The nonlinear and demand-substitution simulations in Section~\ref{sec:experiments} therefore assess robustness to residual misspecification rather than extend the linear bias result.
\end{remark}

\subsection{Lag-adjusted GATE}

The target estimand is the lag-adjusted global average treatment effect over post-burn-in cells:
\begin{equation}
    \tau=\frac{1}{M}\sum_{a\in\Ocal}\{Y_a^{\star}(\one)-Y_a^{\star}(\bm 0)\}.
    \label{eq:estimand}
\end{equation}
Under \eqref{eq:two_H_model},
\begin{equation}
    \tau=\frac{1}{M}\sum_{a\in\Ocal}\left(
    \beta_a+\gamma\sum_cH^{\comp}_{ac}+\eta\sum_cH^{\lag}_{ac}\right).
    \label{eq:estimand_twoH}
\end{equation}
For $B=4$, this is an average over blocks $2,3,4$. The estimand is finite-period and does not require $B\to\infty$.

For weighted network outcomes, we also use a traffic-weighted version in experiments. Given pre-treatment weights $w_a=w_{ib}\ge0$, for example $w_{ib}=m_{ib}$, define
\[
    \tau_w=\frac{1}{W}\sum_{a\in\Ocal}w_a\{Y_a^{\star}(\one)-Y_a^{\star}(\bm 0)\},\qquad W=\sum_{a\in\Ocal}w_a>0.
\]
The weighted formulas are stated formally below. This avoids treating the target used in the empirical analysis as a merely informal extension of the unweighted theorem.

For the design identities, define the common-marginal HT score and baseline-adjusted estimator
\begin{equation}
    \psi_a(Z_a)=\frac{Z_a}{p}-\frac{1-Z_a}{1-p}=\frac{Z_a-p}{p(1-p)}=\frac{Z_a-p}{v},
    \qquad
    \widehat\tau_{\tilde\alpha}=\frac{1}{M}\sum_{a\in\Ocal}\psi_a(Z_a)\{Y_a^{\mathrm{obs}}-\tilde\alpha_a\}.
    \label{eq:psi_estimator_main}
\end{equation}
The fixed adjustment $\tilde\alpha_a$ is chosen before assignment and therefore does not affect the HT bias identity; additional details, including the traffic-weighted version and basic HT moments, are collected in Appendix~\ref{app:additional_discussion}.

\subsection{Exact bias, assignment-polynomial representation, and covariance-level variance bound}

The first result answers a finite-sample design question: which features of a common-marginal assignment law determine bias for the sustained all-treated versus all-control contrast? The short theorem isolates the exact bias identity. The following paragraph then rewrites the assignment-measurable statistic as a linear-plus-quadratic polynomial, explaining why covariance determines exact bias but not exact variance. Theorem~\ref{thm:variance_envelope} deliberately trades sharpness for a covariance-only upper envelope that can guide design. Proofs are in Appendix~\ref{app:derivations}.

For this subsection, consider the structural schedule
\begin{equation}
    Y_a^{\star}(z)=\mu_a+\beta_a z_a+\theta\sum_cH_{ac}z_c,
    \qquad a\in\Ocal,
    \label{eq:thm_one_component_model}
\end{equation}
where $H\ge0$ has zero rows outside $\Ocal$, and let $Y_a^{\rm obs}=Y_a^{\star}(Z)+\varepsilon_a$. If measurement noise is random rather than fixed under the randomization law, impose
\begin{equation}
    \E\{\psi_a(Z_a)\varepsilon_a\}=0,
    \qquad a\in\Ocal.
    \label{eq:noise_ht_orthogonality}
\end{equation}
This follows, for example, from $\E(\varepsilon_a\mid Z)=0$. For any normalized covariance $\Rnorm$, write
\begin{equation}
    b_H(\Rnorm):=\langle H,\Rnorm-\one\one^\top\rangle
    =\sum_{a,c}H_{ac}(\Rnorm_{ac}-1).
    \label{eq:thm_bH}
\end{equation}

For the remainder of this subsection, assume the assignment law has common marginal $p\in(0,1)$, write $q=1-p$, $v=pq$, and set $\Rnorm=\Cov(Z)/v$. Take $\tilde\alpha_a$ to be fixed before assignment, and suppose the disturbances are fixed under randomization or satisfy \eqref{eq:noise_ht_orthogonality}.

\begin{theorem}[HT bias under common marginals]
\label{thm:main}
\begin{equation}
    \E[\widehat\tau_{\tilde\alpha}]-\tau
    =\frac{\theta}{M}b_H(\Rnorm).
    \label{eq:bias_main}
\end{equation}
\end{theorem}

\paragraph{Meaning.}
The theorem gives an exact finite-sample answer: under common marginals, HT bias for the sustained global contrast is determined by exposure-weighted assignment covariance. It applies to any joint Bernoulli law satisfying the stated conditions and does not require the Gaussian parameterization introduced later. \emph{Proof.} See Appendix~\ref{subsec:main_bias_proof}.

\paragraph{Exact assignment-polynomial representation.}
Extend $\beta$, $\mu$, and $\tilde\alpha$ by zero outside $\Ocal$. Let $(\widehat\tau_{\tilde\alpha})^{(0)}$ retain fixed disturbances or set random measurement noise to zero, and define $r_a^{\tilde\alpha}=\mu_a+\varepsilon_a-\tilde\alpha_a$ in the fixed-error regime and $r_a^{\tilde\alpha}=\mu_a-\tilde\alpha_a$ otherwise. With $\deg_H=H^\top\one$ and
\begin{equation}
    g_p^{\tilde\alpha}=\frac{r^{\tilde\alpha}}{v}+\frac{\beta}{p}-\frac{\theta\deg_H}{q},
    \label{eq:gp_generic_adjustment}
\end{equation}
there is a constant $C$, nonrandom under the assignment law, such that
\begin{equation}
    M(\widehat\tau_{\tilde\alpha})^{(0)}
    =(g_p^{\tilde\alpha})^\top Z+\frac{\theta}{v}Z^\top HZ+C.
    \label{eq:exact_estimator_algebra}
\end{equation}
Consequently,
\begin{equation}
    \Var\{(\widehat\tau_{\tilde\alpha})^{(0)}\}
    =M^{-2}\Var\left\{(g_p^{\tilde\alpha})^\top Z+\frac{\theta}{v}Z^\top HZ\right\}.
    \label{eq:exact_variance_main}
\end{equation}
The statistic is quadratic in the assignments, so its exact variance generally involves moments through order four. The next theorem therefore uses a covariance-only upper bound rather than another equality. \emph{Derivation.} See Appendix~\ref{subsec:main_polynomial_proof}.

\begin{remark}[Necessity of HT--noise orthogonality]
\label{rem:noise_orthogonality_necessity}
Without fixed-error conditioning or \eqref{eq:noise_ht_orthogonality}, the exact formula is
\begin{equation}
    \E[\widehat\tau_{\tilde\alpha}]-\tau
    =\frac{\theta}{M}b_H(\Rnorm)
      +\frac1M\sum_{a\in\Ocal}\E\{\psi_a(Z_a)\varepsilon_a\}.
    \label{eq:bias_with_endogenous_noise}
\end{equation}
The extra term need not be small. For example, if
$\varepsilon_a=Z_a-p$, then
$\E\{\psi_a(Z_a)\varepsilon_a\}=1$. Thus the noise condition is part of the theorem, not merely a proof convenience.
\end{remark}

\begin{corollary}[Weighted HT bias identity]
\label{cor:weighted_ht_identity}
Let $w_a\ge0$ be fixed, $W=\sum_{a\in\Ocal}w_a>0$, and
\[
    \widehat\tau_{w,\tilde\alpha}^{HT}
    =W^{-1}\sum_{a\in\Ocal}w_a\psi_a(Z_a)
      \{Y_a^{\rm obs}-\tilde\alpha_a\}.
\]
Under Theorem~\ref{thm:main},
\begin{equation}
    \E\widehat\tau_{w,\tilde\alpha}^{HT}-\tau_w
    =\frac{\theta}{W}\sum_{a,c}w_aH_{ac}(\Rnorm_{ac}-1).
    \label{eq:weighted_exact_bias}
\end{equation}
\end{corollary}

\paragraph{Meaning.}
Target weights determine which exposure edges matter most for bias. The identity therefore supports a target-weighted covariance design; it does not imply that an unweighted implementation already optimizes a weighted estimand. \emph{Proof.} See Appendix~\ref{subsec:weighted_derivations}.

For subsequent precision analysis, let $D_w=\operatorname{diag}(w)$ and define
\begin{equation}
    g_{p,w}^{\tilde\alpha}
    =\frac{w\odot r^{\tilde\alpha}}{v}
      +\frac{w\odot\beta}{p}
      -\frac{\theta H^\top w}{q}.
    \label{eq:weighted_g_vector}
\end{equation}
Appendix~\ref{subsec:weighted_derivations} gives the exact weighted assignment-polynomial representation and the independent-noise variance term. These formulas are used by the weighted covariance envelope in Corollary~\ref{cor:weighted_variance_envelope}.

The exact weighted identity extends to heterogeneous signed exposure operators; the full statement and algebra are given in Appendix~\ref{app:heterogeneous_robust_design_extensions}. The main text focuses on the nonnegative exposure case because only there does the bias become a one-sided expected-cut cost. The exact representation is not yet an optimization criterion, because its quadratic assignment term brings in third- and fourth-order design moments. The next result deliberately sacrifices sharpness for a covariance-only upper envelope.

\begin{theorem}[Covariance-level variance envelope]
\label{thm:variance_envelope}
Under the setup preceding Theorem~\ref{thm:main}, let $s\in\mathbb{R}^L_+$ be a design-chosen nonnegative envelope satisfying $s\ge \deg_H$ componentwise. For fixed disturbances, use the fixed-error definition of $r^{\tilde\alpha}$ above; for random measurement noise, apply the envelope to the no-noise component. If the envelope condition
\begin{equation}
    |g_{p,c}^{\tilde\alpha}|\le \omega_g s_c,
    \qquad c=1,\ldots,L,
    \label{eq:comparability}
\end{equation}
holds for some finite $\omega_g\ge0$, and, in the random-noise regime, if $\varepsilon_a$ are conditionally mean zero, independent of the assignment and mutually independent, with variances $\sigma_a^2$, then
\begin{equation}
    \Var(\widehat\tau_{\tilde\alpha})
    \le
    \frac{2}{M^2}
    \left(\omega_g^2+\frac{\theta^2}{p^2q^2}\right)
    s^\top\left(v\Rnorm+p^2\one\one^\top\right)s
    +\frac{1}{M^2v}\sum_{a\in\Ocal}\sigma_a^2.
    \label{eq:variance_bound_main}
\end{equation}
If the analysis is conditional on a fixed disturbance realization, the final random-noise term is absent and that realization is already included in $r^{\tilde\alpha}$. Therefore the MSE is bounded by the squared bias in \eqref{eq:bias_main} plus \eqref{eq:variance_bound_main}.
\end{theorem}

\paragraph{Meaning.}
\emph{Question:} how can design use only second-order information when exact variance uses higher moments? \emph{Role:} this theorem supplies a conservative covariance-level upper envelope under explicit dominance and comparability conditions. \emph{Boundary:} it is neither an exact variance formula nor automatically numerically sharp; an implemented loss must be checked separately against the quadratic direction appearing here. \emph{Proof.} See Appendix~\ref{subsec:variance_envelope_proof}.

\begin{corollary}[Covariance-optimizable MSE bound]
\label{cor:mse_covariance_bound}
Under Theorems~\ref{thm:main} and~\ref{thm:variance_envelope},
\begin{equation}
\operatorname{MSE}(\widehat\tau_{\tilde\alpha})
\le \frac{\theta^2}{M^2}b_H(\Rnorm)^2
+\frac{2}{M^2}\left(\omega_g^2+\frac{\theta^2}{p^2q^2}\right)
s^\top(v\Rnorm+p^2\one\one^\top)s
+\frac{1}{M^2v}\sum_{a\in\Ocal}\sigma_a^2,
\label{eq:mse_covariance_bound}
\end{equation}
with the final term absent when disturbances are fixed under randomization.
\end{corollary}

Once the working-model scales and envelope direction are fixed, the right-hand side depends on the randomization law only through $\Rnorm$. This is the formal basis for covariance optimization: the exact bias and a conservative precision envelope are expressed in the same second-order design variable. The weighted analogue replaces $H^\top\one$ by $H^\top w$ and requires the positive rescaling stated in Corollary~\ref{cor:weighted_variance_envelope}.

The exact variance expression contains linear, quadratic, cubic, and fourth-order functions of the assignment vector. COSTA does not try to optimize those higher-order moments directly. Instead, the envelope bound motivates a second-order, envelope-driven variance surrogate that depends only on the design covariance matrix. This is the key reduction: after the working-model bias identity identifies which entries of $\Rnorm$ reduce exposure mismatch, the surrogate indicates which covariance directions consume randomization contrast. When $\theta>0$ and $H\ge0$, the bias term is necessarily nonpositive, not merely typically negative: every feasible common-marginal Bernoulli law satisfies $\Rnorm_{ac}\le1$. Section~\ref{subsec:disagreement_geometry} strengthens this observation by representing the bias as an expected weighted cut. Increasing covariance on exposure edges reduces that negative bias, but it can increase treatment-load variance; negative covariance can improve balance and reduce the envelope-driven surrogate, but it worsens exposure mismatch. The feasible range also depends on the treatment share $p$: normalized covariance is always bounded above by one, while its negative lower bound is $-\min\{p/(1-p),(1-p)/p\}$. Hence the bias--variance tradeoff is structurally different when treatment is rare than when $p$ is near one half.

\begin{remark}[Structural baseline and fixed adjustment]
The structural baseline $\mu_a$ is part of the potential outcome model: it is the outcome component that would remain after removing direct treatment, network spillover, temporal carryover, and noise. The fixed adjustment $\tilde\alpha_a$ is what the analyst actually subtracts from the observed outcome, for example a pre-experiment forecast or historical mean. These objects need not be equal. Because the HT score has mean zero, the residual $r_a^{\tilde\alpha}=\mu_a-\tilde\alpha_a$ does not create design bias. It does affect variance: a poor adjustment leaves predictable baseline variation inside $g_p^{\tilde\alpha}$. In the ideal adjustment case $\tilde\alpha=\mu$, condition \eqref{eq:comparability} is implied by
\[
    \left|\beta_c-\frac{p}{q}\theta \deg_{H,c}\right|\le \omega_\beta s_c
\]
with $\omega_g=\omega_\beta/p$. Thus adjustment quality primarily affects precision, not the covariance-bias identity.
\end{remark}

When $p=1/2$, $v=1/4$ and $\Rnorm=4\Sigma$. The bias formula reduces to
\[
    \E[\widehat\tau_{\tilde\alpha}]-\tau
    =\frac{\theta}{M}\{4\langle H,\Sigma\rangle-\langle H,\one\one^\top\rangle\},
\]
which is the unit--time analogue of the usual OCD expression \citep{chen2023ocd}. The bound in \eqref{eq:variance_bound_main} reduces continuously to the $p=1/2$ case, but keeps direct-effect variation and interference strength as separate quantities.

\subsection{Assignment-disagreement geometry and causal bias}
\label{subsec:disagreement_geometry}

The covariance identity admits a sharper geometric interpretation that is valid for every joint Bernoulli law with common marginal $p$, not only for Gaussian copulas or cluster randomization. Define
\begin{equation}
    D_Z(a,c):=\Pr(Z_a\ne Z_c),
    \qquad
    d_Z(a,c):=\frac{D_Z(a,c)}{2v}.
    \label{eq:disagreement_metric_definition}
\end{equation}

\begin{proposition}[Common-marginal expected-cut geometry for positive exposure]
\label{prop:expected_cut_geometry}
For every common-marginal Bernoulli assignment law,
\begin{equation}
    d_Z(a,c)=1-\Rnorm_{ac}.
    \label{eq:disagreement_covariance_identity}
\end{equation}
The function $d_Z$ is a semimetric: it is nonnegative and symmetric, vanishes on the diagonal, and satisfies
\begin{equation}
    d_Z(a,c)\le d_Z(a,b)+d_Z(b,c)
    \qquad\text{for all }a,b,c.
    \label{eq:disagreement_triangle}
\end{equation}
Moreover, the vector $\{D_Z(a,c):a<c\}$ is a convex combination of realized cut vectors and therefore belongs to the cut polytope. The common-marginal restriction selects a subset of that polytope. For the weighted structural model of Corollary~\ref{cor:weighted_ht_identity},
\begin{equation}
    \E\widehat\tau_{w,\tilde\alpha}^{HT}-\tau_w
    =-\frac{\theta}{W}\sum_{a,c}w_aH_{ac}d_Z(a,c)
    =-\frac{\theta}{2vW}\sum_{a,c}w_aH_{ac}\Pr(Z_a\ne Z_c).
    \label{eq:bias_expected_cut}
\end{equation}
Consequently, positive exposure makes the raw-HT design bias for the structural GATE nonpositive.
\end{proposition}

\paragraph{Meaning.}
This proposition turns the signed covariance identity into a nonnegative representation: for positive exposure, every assignment disagreement across an exposure edge incurs causal bias for the global contrast. It is the finite-sample design insight of the paper; it does not extend as a one-sided cut cost to signed exposure without additional cancellation analysis. \emph{Proof.} See Appendix~\ref{subsec:weighted_derivations}.

The cut-polytope connection for symmetric Bernoulli correlation vectors is classical \citep{hubermaric2017cut}; Proposition~\ref{prop:expected_cut_geometry} uses only the elementary realized-cut representation and applies at every common treatment share. Its causal implication is specific: the design bias is the cost of the expected treatment boundary measured on exposure edges.

For an undirected representation, define the symmetrized exposure weight
\begin{equation}
    h_{\{a,c\}}:=w_aH_{ac}+w_cH_{ca},
    \qquad
    \mathcal E_H:=\{\{a,c\}:h_{\{a,c\}}>0\}.
    \label{eq:symmetrized_exposure_weight}
\end{equation}
Diagonal exposure entries are immaterial because $d_Z(a,a)=0$. Then the magnitude of the positive-exposure bias equals
\begin{equation}
    \frac{\theta}{W}\sum_{e\in\mathcal E_H}h_e d_Z(e).
    \label{eq:undirected_cut_bias}
\end{equation}

\begin{corollary}[Zero-bias component collapse]
\label{cor:zero_bias_collapse}
Under the common-marginal assignment law, structural model, and noise conditions of Proposition~\ref{prop:expected_cut_geometry}, suppose $\theta>0$ and $h_e>0$ on every edge of $\mathcal G_H=(\Zcal,\mathcal E_H)$. The raw HT estimator is exactly design-unbiased for the structural GATE if and only if
\begin{equation}
    \Pr(Z_a\ne Z_c)=0
    \qquad\text{for every }\{a,c\}\in\mathcal E_H.
    \label{eq:zero_bias_edge_agreement}
\end{equation}
Hence assignments are almost surely constant on every connected component of $\mathcal G_H$. If $\mathcal G_H$ is connected, then
\begin{equation}
    Z=S\one,
    \qquad S\sim\operatorname{Bernoulli}(p).
    \label{eq:global_rollout_collapse}
\end{equation}
Conditional on fixed potential outcomes and disturbances, the design-generated estimator then has at most two support points and cannot have a nondegenerate continuous Gaussian limit under overlap.
\end{corollary}

\paragraph{Meaning.}
Exact zero bias is not free. On every connected component carrying positive exposure weight, it forces assignments to move together almost surely. The collapse is componentwise rather than necessarily global, and the next path bound quantifies what remains when zero bias is relaxed. \emph{Proof.} See Appendix~\ref{subsec:weighted_derivations}.

Exact zero bias is deliberately extreme, but the same geometry gives quantitative lower bounds. Consider a path of $n$ assignment cells labeled $1,\ldots,n$.

\begin{proposition}[Path disagreement and bias lower bound]
\label{prop:path_bias_lower_bound}
For every common-marginal Bernoulli assignment law on the path and every integer $1\le m<n$,
\begin{equation}
    \sum_{k=1}^{n-1}d_Z(k,k+1)
    \ge \frac1m\sum_{i=1}^{n-m}d_Z(i,i+m).
    \label{eq:path_disagreement_inequality}
\end{equation}
If $\sup_{1\le i\le n-m}\Rnorm_{i,i+m}\le\epsilon_m<1$, then
\begin{equation}
    \frac1{n-1}\sum_{k=1}^{n-1}\{1-\Rnorm_{k,k+1}\}
    \ge \frac{n-m}{m(n-1)}(1-\epsilon_m).
    \label{eq:path_average_disagreement_bound}
\end{equation}
In the unweighted positive-exposure model with $W=n$, $H_{k,k+1}\ge h_0>0$, and $\theta\ge\theta_0>0$,
\begin{equation}
    |\E\widehat\tau^{HT}-\tau|
    \ge \frac{\theta_0h_0}{n}\frac{n-m}{m}(1-\epsilon_m).
    \label{eq:path_bias_rate_bound}
\end{equation}
\end{proposition}

\paragraph{Meaning.}
On a path, decorrelating assignments over a distance $m$ requires enough local disagreement to create an order-$1/m$ bias cost. This is a path-specific lower bound, not a universal lower bound for arbitrary graphs. \emph{Proof.} See Appendix~\ref{subsec:weighted_derivations}.

\begin{remark}[Order sharpness on a path]
\label{rem:path_bound_sharp}
When $n$ is a multiple of $m$, partition the path into consecutive blocks of length $m$ and assign independent $\operatorname{Bernoulli}(p)$ labels to the blocks. Then $d_Z(i,i+m)=1$ for every admissible $i$, while $d_Z(k,k+1)$ is zero except at the $n/m-1$ block boundaries, where it equals one. Thus both sides of \eqref{eq:path_disagreement_inequality} are of order $n/m$. The $m^{-1}$ rate is therefore sharp up to boundary constants and is attained by standard local clustering.
\end{remark}

The path argument extends to arbitrary exposure graphs. Let $\mathcal P$ collect terminal pairs in $G_H$ with required assignment separations $d_Z(u,v)\ge\delta_{uv}$ and define
\begin{equation}
\operatorname{MC}(G_H,h,\mathcal P,\delta)
:=\min_{x_e\ge0}\sum_{e\in E_H}h_ex_e
\quad\text{subject to}\quad
\sum_{e\in P}x_e\ge\delta_{uv}
\label{eq:fractional_multicut_program}
\end{equation}
for every $(u,v)\in\mathcal P$ and every path $P:u\leadsto v$.

\begin{proposition}[General-graph multicut bias bound]
\label{prop:fractional_multicut_bias}
Under positive exposure,
\begin{equation}
\left|\mathbb E\widehat\tau^{HT}_{w,\tilde\alpha}-\tau_w\right|
\ge \frac{\theta}{W}\operatorname{MC}(G_H,h,\mathcal P,\delta).
\label{eq:fractional_multicut_bias_bound}
\end{equation}
\end{proposition}

\paragraph{Meaning.}
The semimetric edge lengths $x_e=d_Z(e)$ satisfy every path constraint, so any design that separates designated terminal pairs must pay at least the graph-specific multicut cost on positive exposure edges. The path result is the one-dimensional special case. Appendix~\ref{app:general_graph_cut_extensions} gives the path-packing dual and explicit grid and expansion corollaries. \emph{Proof.} See Appendix~\ref{subsec:weighted_derivations}.

\section{Covariance-Optimized Design}
\label{sec:objective}

\subsection{Design intuition and target-weighted mismatch}

Positive covariance on a network or lag exposure edge makes the observed local environment more similar to global treatment or global control. Positive covariance everywhere, however, approaches a rollout and removes experimental contrast. The design should therefore concentrate covariance on exposure-relevant pairs while preserving treatment-load stability and adjacent-block variation.

For $k\in\{\comp,\lag\}$, define the weighted bias numerator, exposure mass, and normalized mismatch
\begin{equation}
b_k^{(w)}(\Rnorm)=\sum_{a,c}w_aH^k_{ac}(\Rnorm_{ac}-1),\qquad
E_k^{(w)}=\sum_{a,c}w_aH^k_{ac},\qquad
\bar b_k^{(w)}(\Rnorm)=\frac{b_k^{(w)}(\Rnorm)}{E_k^{(w)}},
\label{eq:weighted_bias_terms_design}
\end{equation}
omitting a channel with zero exposure mass. Target weighting matters under skewed traffic: an unweighted design can look well aligned on average while high-weight cells retain substantial mismatch.

\subsection{Target-weighted objective and robust regularization}
\label{subsec:design_objective}

The exact bias identity and Corollary~\ref{cor:mse_covariance_bound} make $\Rnorm$ the only design variable once the working-model scales are fixed. The resulting criterion is nevertheless a causal working-model surrogate rather than a fully specified outcome model, so COSTA supplements the targeted mismatch and precision terms with soft controls that remain useful under broader outcome variation.

Let $\one_b$ select all assignment cells in block $b$ and $\one_i$ all blocks of unit $i$. Then
\begin{equation}
v\one_b^\top\Rnorm\one_b=\Var\!\left(\sum_iZ_{ib}\right),\qquad
v\one_i^\top\Rnorm\one_i=\Var\!\left(\sum_bZ_{ib}\right).
\label{eq:block_item_balance}
\end{equation}
The first controls treatment-load fluctuations that interact with common block shocks; the second controls cumulative unit exposure that interacts with persistent unit heterogeneity. These are soft analogues of the double-balance logic in multiple-unit switchbacks \citep{masoero2024multiple,missault2025rbsd}. Exact margins may be infeasible for general $p$ and short horizons, and hard balance can conflict with the positive covariance needed on high-weight exposure edges.

For two common-marginal Bernoulli assignments,
\begin{equation}
\Pr(Z_a\ne Z_c)=2v(1-\Rnorm_{ac}).
\label{eq:switch_rate}
\end{equation}
Define the target-weighted adjacent-block switching rate
\begin{equation}
\bar s_{\lag}^{(w)}=\{E_{\lag}^{(w)}\}^{-1}\sum_{a,c}w_aH^{\lag}_{ac}\Pr(Z_a\ne Z_c).
\label{eq:weighted_lag_switch_rate}
\end{equation}

\begin{proposition}[Switching and carryover bias identity]
\label{prop:switching_carryover_frontier}
With $B_{\lag}^{(w)}=(\eta/W)b_{\lag}^{(w)}(\Rnorm)$,
\begin{equation}
\frac{b_{\lag}^{(w)}(\Rnorm)}{E_{\lag}^{(w)}}=-\frac{\bar s_{\lag}^{(w)}}{2v},\qquad
B_{\lag}^{(w)}=-\frac{\eta E_{\lag}^{(w)}}{2vW}\bar s_{\lag}^{(w)}.
\label{eq:switching_carryover_frontier}
\end{equation}
\end{proposition}

\paragraph{Meaning.}
Operational switching and carryover bias are the same disagreement quantity on different scales. A nonvanishing switching requirement under nonvanishing carryover therefore needs debiasing, a bias-aware interval, or a target other than sustained all-treated versus all-control exposure. \emph{Proof.} See Appendix~\ref{subsec:weighted_derivations}.

For a prespecified minimum switching rate $s_{\mathrm{sw,min}}$, define
\begin{equation}
\bar{\Rnorm}_{\lag}(\Rnorm)=\frac{1}{N(B-1)}\sum_{i=1}^N\sum_{b=2}^B\Rnorm_{(i,b),(i,b-1)},\qquad
\rho_{\max}=1-\frac{s_{\mathrm{sw,min}}}{2v}.
\label{eq:rmax}
\end{equation}
The hinge $[\bar{\Rnorm}_{\lag}(\Rnorm)-\rho_{\max}]_+^2$ prevents a nearly static rollout from buying lower lag mismatch by sacrificing adjacent-block contrast.

Let $d_k^{(w)}=(H^k)^\top w$. A convenient formal precision direction is
\begin{equation}
e_w=\kappa_{\comp}\frac{d_{\comp}^{(w)}}{E_{\comp}^{(w)}}+\kappa_{\lag}\frac{d_{\lag}^{(w)}}{E_{\lag}^{(w)}}+\zeta_{\mathrm{ridge}}\one,
\label{eq:ew_def}
\end{equation}
with a zero-mass channel omitted. When $s_w=c_we_w$ satisfies the weighted envelope conditions, $c_w^2$ is absorbed into the precision weight. The target-weighted objective is
\begin{align}
\mathcal L(\Rnorm;p)=&\ \lambda_{b,\comp}\{\bar b_{\comp}^{(w)}(\Rnorm)\}^2
+\lambda_{b,\lag}\{\bar b_{\lag}^{(w)}(\Rnorm)\}^2
+\lambda_V e_w^\top(v\Rnorm+p^2\one\one^\top)e_w\nonumber\\
&+\lambda_{\rm blk}\sum_b\one_b^\top\Rnorm\one_b
+\lambda_{\rm unit}\sum_i\one_i^\top\Rnorm\one_i
+\lambda_{\rm sw}[\bar{\Rnorm}_{\lag}(\Rnorm)-\rho_{\max}]_+^2.
\label{eq:ideal_objective}
\end{align}
The two mismatch channels are squared separately so that network and carryover errors cannot cancel; Appendix~\ref{app:heterogeneous_robust_design_extensions} gives the corresponding ellipsoidal worst-case interpretation. The remaining terms are regularizers rather than additional exact MSE identities. The block- and unit-load quadratics import the precision logic of balanced multi-unit switchbacks, while the switching hinge rules out near-static assignments.

\paragraph{Prespecified weights and ablations.}
For every reported experiment,
\begin{equation}
(\lambda_{b,\comp},\lambda_{b,\lag},\lambda_V,\lambda_{\rm blk},\lambda_{\rm unit},\lambda_{\rm sw})
=(1,1,0.05,0.05,0,10),\qquad s_{\rm sw,min}=0.10.
\label{eq:default_objective_weights}
\end{equation}
The normalized network and lag mismatches are the primary causal criteria and receive unit weight. The precision term is a conservative, scale-dependent envelope direction rather than an exact variance component: setting $\lambda_V=1$ would treat it as commensurate with the squared mismatch terms and can let it dominate, so $0.05$ provides mild precision regularization while preserving causal alignment. We likewise use $0.05$ for the active block-load control, set the optional full unit-load term to zero because the implemented optimizer does not evaluate that complete quadratic, and place weight $10$ on the smaller-scale switching hinge to enforce the $0.10$ target. All coefficients were fixed before simulated outcomes were generated and reused across both datasets and all three outcome models. Prespecified leave-one-term-out ablations on both datasets remove the active balance and switching controls in turn; they document their RMSE benefits rather than tune the defaults.

\subsection{Relation to canonical designs}

Independent Bernoulli randomization has $\Rnorm_{ac}=0$ off the diagonal, maximizing local contrast but leaving every positive exposure edge with the full disagreement cost $1-\Rnorm_{ac}=1$. Cluster randomization sets many within-cluster correlations near one; it can work well for separated communities but treats all pairs in a large cluster similarly and can be inefficient when important edges cross cluster boundaries or cluster sizes differ. COSTA instead aligns high-weight exposure pairs without forcing an entire cluster to move together.

Balanced switchbacks motivate the load and switching controls above, but they do not identify which units affect one another. Static optimized covariance design acts within one period: independently optimized blocks leave lag bias unchanged, whereas reusing an assignment suppresses switching. COSTA targets both channels jointly. These comparisons concern finite-sample MSE, not inferential compatibility; Section~\ref{sec:inference} verifies the additional conditions needed for IC-0, IC-1, and IC-2.

\section{Gaussian-Copula and Kronecker Parameterization}
\label{sec:param_opt}

\subsection{\texorpdfstring{General-$p$ Gaussian copula}{General-p Gaussian copula}}

Let $G\sim N(0,\Omega)$ and threshold
\begin{equation}
Z_a=\ind\{G_a\le\Phi^{-1}(p)\}.
\label{eq:gaussian_threshold_assignment}
\end{equation}
Then $Z_a\sim\operatorname{Bernoulli}(p)$ and
\begin{equation}
\Rnorm_{ac}=\Gp(\Omega_{ac})
=\frac{\Phi_2\{\Phi^{-1}(p),\Phi^{-1}(p);\Omega_{ac}\}-p^2}{p(1-p)}.
\label{eq:gaussian_copula_map}
\end{equation}
At $p=1/2$, $\Gp(\rho)=2\arcsin(\rho)/\pi$. Monotonicity carries latent-correlation changes into binary covariance while retaining a valid positive-semidefinite Gaussian law. A fully flexible cell-factor construction assigns a unit vector $u_a$ to each assignment cell and sets $\Omega_{ac}=u_a^\top u_c$; it is useful for moderate problems but has order $NBd$ parameters and repeats network estimation block by block.

\subsection{Kronecker network--time covariance}

A fully flexible $NB\times NB$ correlation matrix has order $N^2B^2$ pairwise entries and a global positive-semidefinite constraint. Even free cell embeddings ignore the repeated causal architecture: network exposure connects units within a block, whereas lag exposure connects blocks within a unit. COSTA therefore uses unit factors $f_i\in\mathbb R^{d_I}$ and time factors $t_b\in\mathbb R^{d_T}$ with unit norms and imposes
\begin{equation}
\Omega_{(i,b),(j,r)}=(f_i^\top f_j)(t_b^\top t_r).
\label{eq:kronecker_cell_corr}
\end{equation}
Thus $\Omega_{(i,b),(j,b)}=f_i^\top f_j$, $\Omega_{(i,b),(i,r)}=t_b^\top t_r$, and other space--time correlations are their product. Unit structure is shared across blocks and temporal structure across units, matching the two exposure operators directly.

Equivalently, the cell factor is $f_i\otimes t_b$. The parameter count falls to $Nd_I+Bd_T$, and optimization can evaluate only observed network edges, adjacent lags, and fixed diagnostic batches rather than all cell pairs. The construction is causal-structure-driven rather than merely low rank: it directly parameterizes the covariance directions identified by the bias and MSE bounds, while the remaining correlations arise from a valid product Gram matrix. Gaussian scores can be sampled without forming an $NB\times NB$ matrix. Common Kronecker factors can nevertheless transmit weak covariance across many distant cells, so Section~\ref{sec:inference} treats computational scalability and inferential locality separately.

\begin{remark}[When the Kronecker form may be restrictive]
The separable form implies that the unit covariance pattern is stable across blocks up to a temporal-pair scalar. This is natural when the network exposure graph is stable. Strongly time-varying networks may instead require block-specific unit factors, a sum of Kronecker components, or sparse edge corrections.
\end{remark}

\paragraph{COSTA-LocalPenalty.}
Kronecker sharing is computationally valuable, but it can diffuse small correlations over many unit--time pairs that are neither network edges nor adjacent lags. Such covariance contributes little to the working-model exposure objective, may increase estimator variance, and can create a nonlocal covariance tail omitted by a short HAC window. COSTA-LocalPenalty addresses this finite-sample issue by shrinking sampled distant pairwise covariances while retaining the same valid Gaussian-copula law.

Before optimization, fix batches of ordered distinct-unit pairs, time pairs separated by at least two blocks, and their product-cell combinations. For unit pairs $(i_\ell,j_\ell)$, time pairs $(b_\ell,r_\ell)$, and product pairs $(i'_\ell,j'_\ell,b'_\ell,r'_\ell)$, define
\begin{equation}
R^I_\ell=\Gp(f_{i_\ell}^\top f_{j_\ell}),\qquad
R^T_\ell=\Gp(t_{b_\ell}^\top t_{r_\ell}),\qquad
R^P_\ell=\Gp\{(f_{i'_\ell}^\top f_{j'_\ell})(t_{b'_\ell}^\top t_{r'_\ell})\},
\label{eq:local_penalty_components}
\end{equation}
and set
\begin{equation}
J_{\rm loc}=m_I^{-1}\sum_\ell(R^I_\ell)^2+m_T^{-1}\sum_\ell(R^T_\ell)^2+m_P^{-1}\sum_\ell(R^P_\ell)^2,
\qquad
\mathcal L_{\rm LP}(\Rnorm;p)=\mathcal L(\Rnorm;p)+\lambda_{\rm loc}J_{\rm loc}.
\label{eq:local_penalty_exact}
\end{equation}
The batches are fixed before optimization, so the criterion is stable and its cost scales with sampled pairs rather than $(NB)^2$. Because the unit pairs are generic distinct pairs rather than graph-certified nonedges, $J_{\rm loc}$ is a finite-sample pairwise regularizer, not an IC-0 certificate. We prespecify $\lambda_{\rm loc}=1$. The experiments show that this adaptation usually leaves point-estimation RMSE close to COSTA while materially narrowing model-assisted intervals; the fixed-batch sweep reports the associated locality--accuracy trade-off.

\subsection{Sampling and optimization}

Draw a standard Gaussian vector $\xi\in\mathbb R^{d_Id_T}$ and set
\begin{equation}
G_{ib}=(f_i\otimes t_b)^\top\xi.
\label{eq:kron_sampling_score}
\end{equation}
Equivalently, reshape $\xi$ into a Gaussian matrix $E\in\mathbb R^{d_I\times d_T}$ and compute $G_{ib}=f_i^\top Et_b$. Assign $Z_{ib}=\ind\{G_{ib}\le\Phi^{-1}(p)\}$. The shared Gaussian matrix creates both unit and temporal dependence without materializing the full covariance matrix.

Algorithm~\ref{alg:stocd} optimizes row-normalized unit and time factors and samples assignments by thresholding the resulting Gaussian scores. Exact finite-batch construction, sparse precision and sampled block-load proxies, and the complete theory--implementation map are given in Appendix~\ref{subsec:implementation_details_appendix}; no full unit-load quadratic is evaluated.

\begin{algorithm}[!htbp]
\caption{COSTA covariance optimization and assignment sampling}
\label{alg:stocd}
\begin{algorithmic}[1]
\REQUIRE Network matrix $\AI$, blocks $B$, treatment probability $p$, dimensions $(d_I,d_T)$, objective weights, switching target, and optional COSTA-LocalPenalty batches.
\STATE Initialize unit factors $F\in\mathbb R^{N\times d_I}$ and time factors $T\in\mathbb R^{B\times d_T}$; normalize rows.
\REPEAT
    \STATE Compute latent network-edge and lag correlations and map them through $\mathcal G_p$.
    \STATE Evaluate mismatch, precision, balance, switching, and optional pairwise-shrinkage terms.
    \STATE Take an optimizer step and renormalize all factor rows.
\UNTIL convergence
\STATE Sample Gaussian factors, form Kronecker scores, and threshold at $\Phi^{-1}(p)$.
\RETURN $Z\in\{0,1\}^{N\times B}$.
\end{algorithmic}
\end{algorithm}

\subsection{Computational scaling and parameter choices}
\label{subsec:scaling}

The Kronecker parameterization uses $O(Nd_I+Bd_T)$ parameters rather than $O(NBd)$ free cell embeddings or an unrestricted $O(N^2B^2)$ covariance matrix. For fixed $B$, the optimization is essentially linear in the number of units because only sparse exposure edges, adjacent lags, and fixed pair batches are evaluated. The network dimension $d_I$ controls the number of unit-side covariance directions; the time dimension is smaller because $B$ is finite. The default finite-sample design uses $d_I=16$ and $d_T=B-1$, with dimension and reduction ablations reported in Section~\ref{sec:experiments} and the appendix.

The shared-factor representation implies $\operatorname{rank}(\Omega)\le d_Id_T$. Low rank is a useful finite-sample regularizer but not an automatic weak-dependence certificate: distant sets can recover the same latent factors even when pairwise correlations are small. Proposition~\ref{prop:global_factor_certifiability} and Corollary~\ref{cor:tight_frame_block_saturation} formalize this limitation. Pairwise shrinkage can improve finite-sample variance, but formal inference still requires the separated-set conditions in Section~\ref{sec:inference}.

\section{Inference for Covariance-Structured Assignments}
\label{sec:inference}

Assignment covariance and interference jointly determine the centered HT contribution field. The dependence arises in two steps. First, the assignment law may deliberately correlate treatments across unit--time cells. Second, network and temporal interference make each contribution a nonlinear function of assignments in a local exposure neighborhood. Even independent treatment indicators can therefore generate dependent contributions when those neighborhoods overlap; covariance engineering adds dependence between the assignment neighborhoods themselves. The inferential object is the random field produced by this composition.

This section asks four separate questions. First, when do the resulting centered HT contributions satisfy an existing network central limit theorem? Second, how can a structured assignment covariance verify the required dependence condition for \emph{groups} of contributions, rather than only for individual treatment pairs? Third, when is the design expectation close enough to the GATE for the centered limit to be causal? Fourth, how can the design-centered variance be estimated after allowing the contribution means to vary across cells?

Our probabilistic engine is the graph-$\psi$ theory of \citet{kojevnikov2021network}, which begins once a dependence coefficient for the contribution field is available. The new results provide the bridge from design to that coefficient: the outcome-locality architecture identifies the Gaussian assignment blocks generating each contribution, canonical correlations control dependence between two separated blocks, and a spectral far-row-mass bound supplies a primitive sufficient condition. Because the argument separates the local outcome map from the covariance sequence, it applies to COSTA and to other structured Gaussian covariance models whenever the same conditions hold. The section then combines this dependence analysis with the exact bias identity and develops feasible variance-estimation routes.

The analysis is conditionally design-based in a possibly random pre-treatment environment. Let $\mathcal F_N$ contain pre-treatment data, the graph and optimized design, target weights, baseline adjustments, the structural potential-outcome schedule, and all disturbances treated as fixed under randomization. Write
\begin{equation}
\mathbb P_N(\cdot):=\Pr(\cdot\mid\mathcal F_N),\qquad
\mathbb E_N,\quad\Cov_N,\quad\Var_N,\qquad
\|U\|_{r,N}:=\{\mathbb E_N|U|^r\}^{1/r}.
\label{eq:conditional_probability_operators}
\end{equation}
A \emph{good-environment sequence} is a collection $E_N^{\rm good}\in\mathcal F_N$ with $\Pr(E_N^{\rm good})\to1$. All graph, design, and pilot quantities below are $\mathcal F_N$-measurable, and deterministic-looking restrictions need only hold on $E_N^{\rm good}$; displayed rates involving only such quantities are understood under the outer environment law. An external pilot completed before the current randomization is included in $\mathcal F_N$, so its realized fit is fixed under $\mathbb P_N$, while its $O_p$ rate refers to the outer pilot/environment randomness. Deterministic environments are covered by taking $E_N^{\rm good}$ to be the whole sample space.

For a randomization statistic $U_N$, write
\begin{align}
U_N\rightsquigarrow_{\mathcal F}N(0,1)
&\quad\Longleftrightarrow\quad
\sup_{t\in\mathbb R}
\left|\mathbb P_N(U_N\le t)-\Phi(t)\right|\xrightarrow{p}0,
\label{eq:conditional_weak_convergence_notation}\\
R_N\xrightarrow{p\mid\mathcal F}0
&\quad\Longleftrightarrow\quad
\mathbb P_N(|R_N|>\epsilon)\xrightarrow{p}0
\quad\text{for every }\epsilon>0.
\label{eq:conditional_probability_convergence_notation}
\end{align}
The corresponding $o_{p\mid\mathcal F}$ and $O_{p\mid\mathcal F}$ notation has its usual conditional-probability meaning. The outer probability in these displays integrates only the random environment. Because the conditional Kolmogorov error is bounded by one, the first convergence implies the usual unconditional Gaussian limit; Appendix~\ref{app:conditional_and_primitive_extensions} records the transfer formally. The main Gaussian certificate conditions on all outcome noise. The same coefficient survives additional independent own-cell noise, as shown in the appendix; shared or locally shared random noise must instead enter an enlarged joint certificate.

\begin{figure}[H]
\centering
\fbox{\begin{minipage}{0.94\linewidth}
\textbf{From covariance structure to inference.}

\textbf{IC-0:} structured assignment covariance $+$ local outcome neighborhoods $\Longrightarrow$ dependence bounds for separated HT contributions $\Longrightarrow$ graph-$\psi$ rates $\Longrightarrow$ design-centered normality.

\vspace{0.75em}
\textbf{IC-1:} expected-cut bias at the standard-error scale determines whether the design expectation may be replaced by the GATE.

\vspace{0.75em}
\textbf{IC-2:} valid network-HAC $+$ accurate contribution means $\Longrightarrow$ feasible studentization.
\end{minipage}}
\caption{The theoretical sequence. IC-0 controls the contribution dependence created jointly by the assignment covariance and the network--temporal outcome map; IC-1 is the separate causal-centering check; IC-2 additionally requires a valid variance estimator and feasible contribution-mean centering.}
\label{fig:design_verification_chain}
\end{figure}

\subsection{Inferential layers and the existing probabilistic engine}
\label{subsec:inference_layers_engine}

We consider a triangular array with
\begin{equation}
N\to\infty,\qquad B_N\ge2,\qquad
M_N:=|\Ocal_N|=N(B_N-1),\qquad L_N:=|\Zcal_N|=NB_N,
\qquad 1<\frac{L_N}{M_N}\le2.
\label{eq:inference_asymptotic_regime}
\end{equation}
Let
\begin{equation}
\mathcal G_N=(\Zcal_N,\mathcal E_N),\qquad
\mathcal E_N^{\rm nat}\subseteq\mathcal E_N,\qquad
\mathcal G_N\text{ connected},
\label{eq:connected_dependence_supergraph}
\end{equation}
be a pre-treatment dependence supergraph containing the natural network--time exposure graph. Added edges affect only the inferential metric; their enlarged balls and shells must be paid for in the rates below. Write $d_N^{\Zcal}$ for shortest-path distance in the full assignment-cell graph $\mathcal G_N$, and let $\mathcal G_N^{\Ocal}:=\mathcal G_N[\Ocal_N]$ be the subgraph induced by the contribution cells.

The contribution array is indexed only by post-burn-in outcome cells, whereas the Gaussian assignment lives on all unit--block cells. The following condition prevents burn-in or auxiliary assignment cells from creating shortest-path shortcuts that are invisible to the contribution graph; without it, separation of contribution indices would not imply separation of the Gaussian neighborhoods generating their scores.

\begin{assumption}[Contribution-metric compatibility]
\label{ass:contribution_metric_compatibility}
The induced graph $\mathcal G_N^{\Ocal}$ is connected and is geodesically closed in $\mathcal G_N$:
\begin{equation}
 d_{\mathcal G_N^{\Ocal}}(a,c)=d_N^{\Zcal}(a,c),
 \qquad a,c\in\Ocal_N.
\label{eq:geodesic_closure}
\end{equation}
We write $d_N$ for this common shortest-path distance on $\Ocal_N$.
\end{assumption}

For $a\in\Ocal_N$, define
\begin{equation}
\mathcal B_N(a;s)=\{c\in\Ocal_N:d_N(a,c)\le s\},\qquad
\partial\mathcal B_N(a;s)=\{c\in\Ocal_N:d_N(a,c)=s\}.
\label{eq:cell_balls_shells}
\end{equation}

\begin{remark}[Compatibility of the assignment and contribution index sets]
\label{rem:index_set_compatibility}
Assumption~\ref{ass:contribution_metric_compatibility} makes the KMS invocation literal: the contribution array is indexed by the connected graph $\mathcal G_N^{\Ocal}$, and its graph distance agrees with the full assignment-cell distance whenever both endpoints are outcome cells. The standard one-lag construction used here has this property when contemporaneous network edges are replicated in every block: any path between cells with block index at least two that visits the burn-in block can be projected from block one to block two without increasing its length. Added inferential edges need not preserve this argument and must be checked separately. If geodesic closure is not imposed, a valid alternative is to extend the contribution array by zeros to all of $\Zcal_N$ and verify every ball, shell, canonical-correlation separation, and HAC topology condition on the full graph. No equivalence of those topology rates follows merely from $L_N/M_N\le2$.
\end{remark}

Let $w_{Na}\ge0$, $W_N=\sum_{a\in\Ocal_N}w_{Na}>0$, and define the triangular-array score and adjusted observed outcome by
\begin{equation}
\psi_{Na}(z):=\frac{z}{p_N}-\frac{1-z}{1-p_N},\qquad
X_{Na}^{\rm obs}:=Y_{Na}^{\rm obs}-\widetilde\alpha_{Na}.
\label{eq:triangular_score_adjusted_outcome}
\end{equation}
Set
\begin{equation}
\omega_{Na}:=\frac{M_Nw_{Na}}{W_N},\qquad
Q_{Na}:=\omega_{Na}\psi_{Na}(Z_{Na})X_{Na}^{\rm obs},\qquad
\bar q_{Na}:=\mathbb E_N Q_{Na},\qquad
\xi_{Na}:=Q_{Na}-\bar q_{Na}.
\label{eq:stabilized_contributions}
\end{equation}
Then
\begin{equation}
\widehat\tau^{HT}_{w,\tilde\alpha}=M_N^{-1}\sum_aQ_{Na},\qquad
T_N:=\sum_a\xi_{Na},\qquad
\sigma_N^2:=\Var_N(T_N),\qquad
V_N:=\sigma_N^2/M_N^2.
\label{eq:inference_variance_scaling}
\end{equation}

\begin{definition}[Three inferential compatibility levels]
\label{def:ic_hierarchy}
A design--estimator sequence is \emph{IC-0} if
\[
(\widehat\tau^{HT}_{w,\tilde\alpha}-\mathbb E_N\widehat\tau^{HT}_{w,\tilde\alpha})/\sqrt{V_N}
\rightsquigarrow_{\mathcal F}N(0,1).
\]
It is \emph{IC-1} if IC-0 holds and
$|\mathbb E_N\widehat\tau^{HT}_{w,\tilde\alpha}-\tau_{w,N}|/\sqrt{V_N}\xrightarrow{p}0$.
It is \emph{IC-2} if IC-1 holds and a data-based $\widehat V_N$ satisfies
$\widehat V_N/V_N\xrightarrow{p\mid\mathcal F}1$. A one-sided conservative variance calibration may yield conservative coverage without being IC-2.
\end{definition}

The remainder of this subsection first supplies conditions for IC-0. Subsection~\ref{subsec:bias_locality_trilemma} then shows why IC-0 need not imply IC-1, and Subsection~\ref{subsec:network_hac} gives conditional sufficient routes toward IC-2.

\begin{definition}[Graph-$\psi$ dependence]
\label{def:graph_psi_dependence}
Conditional on $\mathcal F_N$, the centered array is graph-$\psi$ dependent with $\mathcal F_N$-measurable coefficients $\theta_N(s)\in[0,1]$ if $\theta_N(0)=1$, $\theta_N(s)$ is nonincreasing, and for disjoint nonempty $A,B\subseteq\Ocal_N$ with $d_N(A,B)\ge s$ and bounded Lipschitz $f,g$,
\begin{equation}
|\Cov_N\{f(\xi_{N,A}),g(\xi_{N,B})\}|
\le C_\psi |A||B|
\{\|f\|_\infty+\operatorname{Lip}(f)\}
\{\|g\|_\infty+\operatorname{Lip}(g)\}\theta_N(s).
\label{eq:graph_psi_definition}
\end{equation}
\end{definition}

For $k\ge1$, define
\begin{align}
\delta_N^{\partial}(s;k)
&:=M_N^{-1}\sum_{a\in\Ocal_N}|\partial\mathcal B_N(a;s)|^k,
\label{eq:delta_shell_definition}\\
\Delta_N(s,m;k)
&:=M_N^{-1}\sum_{a\in\Ocal_N}
\max_{c\in\partial\mathcal B_N(a;s)}
|\mathcal B_N(a;m)\setminus\mathcal B_N(c;s-1)|^k,
\label{eq:Delta_shell_definition}\\
c_N(s,m;k)
&:=\inf_{\alpha>1}
\{\Delta_N(s,m;k\alpha)\}^{1/\alpha}
\left\{\delta_N^{\partial}\left(s;\frac{\alpha}{\alpha-1}\right)\right\}^{1-1/\alpha},
\label{eq:c_topology_definition}
\end{align}
with empty-shell maxima equal to zero and $\mathcal B_N(c;-1)=\varnothing$.

\begin{assumption}[Overlap and moments]
\label{ass:weights_moments}
For constants $\underline p>0$, $\nu>4$, and $C_\xi<\infty$, on $E_N^{\rm good}$,
\begin{equation}
p_N\in[\underline p,1-\underline p],\qquad
\max_{a\in\Ocal_N}\|\xi_{Na}\|_{\nu,N}\le C_\xi,
\qquad \sigma_N^2>0.
\label{eq:overlap_moment_condition}
\end{equation}
\end{assumption}

\begin{assumption}[KMS topology--dependence rates]
\label{ass:weak_dependence_rates}
The array is graph-$\psi$ dependent, and there exists a deterministic sequence of positive integers $m_N\to\infty$ such that, for each $k\in\{1,2\}$,
\begin{equation}
\frac{M_N}{\sigma_N^{2+k}}
\sum_{s\ge0}c_N(s,m_N;k)\theta_N(s)^{1-(2+k)/\nu}\xrightarrow{p}0,
\qquad k=1,2,
\label{eq:weak_dep_clt_rate_1}
\end{equation}
and, for that same sequence,
\begin{equation}
\frac{M_N^2}{\sigma_N}\theta_N(m_N)^{1-1/\nu}\xrightarrow{p}0.
\label{eq:weak_dep_clt_rate_2}
\end{equation}
\end{assumption}

\begin{theorem}[Design-based specialization of the KMS network CLT]
\label{thm:cell_clt}
Under Assumptions~\ref{ass:contribution_metric_compatibility}, \ref{ass:weights_moments}, and \ref{ass:weak_dependence_rates},
\begin{equation}
\frac{\widehat\tau^{HT}_{w,\tilde\alpha}-\mathbb E_N\widehat\tau^{HT}_{w,\tilde\alpha}}{\sqrt{V_N}}
=\frac{T_N}{\sigma_N}\rightsquigarrow_{\mathcal F} N(0,1).
\label{eq:clt_centered}
\end{equation}
If also
\begin{equation}
\frac{|\mathbb E_N\widehat\tau^{HT}_{w,\tilde\alpha}-\tau_{w,N}|}{\sqrt{V_N}}\xrightarrow{p}0,
\label{eq:bias_negligible}
\end{equation}
then
\begin{equation}
\frac{\widehat\tau^{HT}_{w,\tilde\alpha}-\tau_{w,N}}{\sqrt{V_N}}
\rightsquigarrow_{\mathcal F} N(0,1).
\label{eq:clt_tau}
\end{equation}
\end{theorem}

\paragraph{Meaning.}
\emph{Question:} when do local centered HT contributions obey an existing network CLT? \emph{Role:} the theorem delivers IC-0 after the topology, moment, variance, and graph-$\psi$ rates are verified. \emph{Boundary:} it centers at the design expectation; replacing that center by the GATE is the separate IC-1 condition displayed in \eqref{eq:bias_negligible}. \emph{Proof.} See Appendix~\ref{subsec:clt_proofs}.

Theorem~\ref{thm:cell_clt} is a direct design-based specialization of the existing KMS probabilistic result. The paper's new theoretical work begins by connecting assignment structure to the contribution-dependence coefficient. Under root-$M_N$ sum variance, polynomial graph-ball growth and geometric dependence verify the KMS rates with $m_N\asymp\log M_N$; Appendix~\ref{app:conditional_and_primitive_extensions} gives the rate calculation.

\subsection{Dependence control for separated sets of contributions}
\label{subsec:design_side_certificate}

The interference map first determines which assignments generate each contribution. This local-measurability interface also turns finite-range assignment designs into finite-range contribution arrays.

\begin{assumption}[Conditional local contributions]
\label{ass:local_outcomes}
For a finite $r_Y$, conditional on $\mathcal F_N$, $Q_{Na}$ is measurable with respect to assignments in
$\mathcal B_N^{\Zcal}(a;r_Y):=\{c\in\Zcal_N:d_N^{\Zcal}(a,c)\le r_Y\}$ and $\mathcal F_N$-measurable quantities. Proposition~\ref{prop:gaussian_canonical_certificate} makes no claim for additional random outcome noise outside $\mathcal F_N$. Appendix~\ref{app:conditional_and_primitive_extensions} shows that independent own-cell noise preserves the same coefficient; shared or locally shared noise requires an enlarged joint assignment--noise certificate.
\end{assumption}

For the semi-synthetic demand model, a finite radius can be asserted only on a supergraph containing the $A^\star$ edges, lag-two own-unit links, and lagged-competitor paths. The default sampler is not certified on that enlarged graph; those simulations therefore provide robustness evidence for finite-sample MSE rather than an inference certificate.

Finite-range designs now admit a concrete sufficient result.

\begin{corollary}[Finite-range local-factor IC-0]
\label{cor:local_factor_ic0}
Suppose Assumption~\ref{ass:local_outcomes} holds and latent assignments admit
\begin{equation}
G_{Na}=\sum_{u\in\Lambda_{Na}}\ell_{Na,u}\epsilon_{Nu},\qquad \sum_u\ell_{Na,u}^2=1,
\label{eq:local_factor_sampler}
\end{equation}
with independent Gaussian shocks and disjoint loading supports whenever assignment-cell distance exceeds a fixed $r_F$. If $\sigma_N^2\asymp M_N$, fourth moments are uniformly bounded, and
\begin{equation}
\max_{a\in\Ocal_N}|\mathcal B_N(a;2r_Y+r_F)|=o_p(M_N^{1/4}),
\label{eq:local_factor_ic0_rate}
\end{equation}
then the design-centered limit in Theorem~\ref{thm:cell_clt} holds.
\end{corollary}

Disjoint factor supports make contributions beyond radius $2r_Y+r_F$ conditionally independent, so a dependency-graph normal approximation applies \citep{baldirinott1989}. This also identifies a certified design class over which the finite-sample COSTA objective can be optimized: local-factor loadings with a fixed primitive-shock radius. The reported low-rank Kronecker sampler is evaluated for finite-sample MSE and is not automatically covered; causal centering and studentization remain separate. Exact locality is not necessary: Proposition~\ref{prop:local_approximation_transfer} transfers the CLT from local approximants when their aggregate conditional-$L^2$ error is negligible relative to $\sigma_N$.

Pairwise assignment correlations cannot answer the required dependence question because each HT contribution is a nonlinear function of a neighborhood of assignments. The relevant proof object is therefore maximal dependence between the Gaussian blocks generating two separated sets of contributions.

Let $G_N\sim N(0,\Omega_N)$ and $Z_{Na}=\ind\{G_{Na}\le\Phi^{-1}(p_N)\}$. Throughout this section, $\Omega_N$ denotes the \emph{latent Gaussian correlation}, $\Rnorm_N=\Gp(\Omega_N)$ denotes the \emph{binary normalized covariance}, and $\theta_N^{\rm GC}$ denotes a \emph{setwise sufficient upper bound} for the actual graph-$\psi$ coefficient. These three objects are not interchangeable. For $U,V\subseteq\Zcal_N$, define
\begin{equation}
\rho_{\Omega_N}(U,V)
:=\sup_{x,y}
\frac{|x^\top\Omega_{N,UV}y|}
{(x^\top\Omega_{N,UU}x)^{1/2}(y^\top\Omega_{N,VV}y)^{1/2}},
\label{eq:gaussian_canonical_correlation}
\end{equation}
with zero-denominator ratios set to zero, and write
$A^{+r_Y}=\cup_{a\in A}\mathcal B_N^{\Zcal}(a;r_Y)$.

\begin{proposition}[Canonical-correlation design certificate]
\label{prop:gaussian_canonical_certificate}
Under Assumptions~\ref{ass:contribution_metric_compatibility} and~\ref{ass:local_outcomes}, the contribution array is graph-$\psi$ dependent with
\begin{equation}
\theta_N^{\rm GC}(0):=1,\qquad
\theta_N^{\rm GC}(s):=
\sup_{\substack{\varnothing\ne A,B\subseteq\Ocal_N,\ A\cap B=\varnothing\\d_N(A,B)\ge s}}
\rho_{\Omega_N}(A^{+r_Y},B^{+r_Y}),\quad s\ge1,
\label{eq:gaussian_theta_certificate}
\end{equation}
where an empty supremum is zero. Hence Theorem~\ref{thm:cell_clt} applies whenever the KMS rates hold with $\theta_N=\theta_N^{\rm GC}$. The result covers the discontinuous threshold map because the multidimensional Gaussian maximal-correlation (Gebelein) inequality applies to square-integrable measurable transformations of the two Gaussian blocks \citep{veraar2009correlation}.
\end{proposition}

\paragraph{Meaning.}
The relevant dependence object is not a single pairwise assignment correlation. It is maximal dependence between the Gaussian blocks that generate two separated sets of nonlinear local contributions. The proposition gives a sufficient upper bound for graph-$\psi$ dependence; failure of this bound does not prove nonnormality. \emph{Proof.} See Appendix~\ref{subsec:gaussian_certificate_proofs}.

This is a setwise rather than pairwise condition: it controls nonlinear transformations of arbitrary separated sets of contributions. It is sufficient, not necessary. A statistic may have another conditional, stable, or mixed-normal limit even when this uniform approach fails.

Define the absolute far-row mass
\begin{equation}
\alpha_N(t):=\sup_{u\in\Zcal_N}\sum_{v:d_N^{\Zcal}(u,v)\ge t}|\Omega_{N,uv}|.
\label{eq:absolute_far_row_mass}
\end{equation}

The canonical-correlation supremum is the relevant setwise object, but it is difficult to verify directly over all separated set pairs. The next proposition replaces it with a conservative primitive bound based on a spectral floor and absolute far-row mass.

\begin{proposition}[Spectral far-row-mass certificate]
\label{prop:spectral_rowsum_certificate}
Under Assumption~\ref{ass:contribution_metric_compatibility}, if $\lambda_{\min}(\Omega_N)\ge\kappa_N>0$, then for $s>2r_Y$,
\begin{equation}
\theta_N^{\rm GC}(s)\le
\min\left\{1,\frac{\alpha_N(s-2r_Y)}{\kappa_N}\right\}.
\label{eq:spectral_rowsum_bound}
\end{equation}
Thus a uniform spectral floor and geometric far-row-mass decay give geometric setwise dependence. Average far-pair covariance is not a substitute for this rowwise bound.
\end{proposition}

\paragraph{Meaning.}
Canonical correlation is the relevant setwise dependence object but is difficult to maximize over all separated sets. This proposition replaces it by a conservative primitive check using a spectral floor and absolute far-row mass. Separate topology and variance rates are still needed for a CLT, and kernel conditions are still needed for HAC. \emph{Proof.} See Appendix~\ref{subsec:gaussian_certificate_proofs}.

The primitive check can be reduced factorwise for Kronecker covariance. Suppose $\Omega_N=\Omega_{I,N}\otimes\Omega_{T,N}$ and the assignment metric is additive,
$d_N^{\Zcal}\{(i,b),(j,r)\}=d_{I,N}(i,j)+d_{T,N}(b,r)$. Define
\[
S_{X,N}=\sup_u\sum_v|\Omega_{X,N,uv}|,\qquad
\alpha_{X,N}(t)=\sup_u\sum_{v:d_{X,N}(u,v)\ge t}|\Omega_{X,N,uv}|,
\quad X\in\{I,T\}.
\]
\begin{corollary}[Kronecker reduction of far-row mass]
\label{cor:kronecker_rowsum}
The cell-level far-row mass satisfies
\begin{equation}
\alpha_N(t)\le \alpha_{I,N}(\lceil t/2\rceil)S_{T,N}+S_{I,N}\alpha_{T,N}(\lceil t/2\rceil).
\label{eq:kronecker_rowsum}
\end{equation}
If the factor matrices have spectral floors $\kappa_{I,N}$ and $\kappa_{T,N}$, then $\Omega_N$ has floor $\kappa_{I,N}\kappa_{T,N}$, so Proposition~\ref{prop:spectral_rowsum_certificate} applies with this product.
\end{corollary}

Additive distance at least $t$ implies that the unit pair or the time pair is separated by at least $\lceil t/2\rceil$; summing the absolute Kronecker products gives \eqref{eq:kronecker_rowsum}, while eigenvalues of a Kronecker product multiply. Thus spatial and temporal decay can be checked separately. Fixed-rank Gram factors generally have zero spectral floor, so local loadings, growing dimensions, structured regularization, or another dependence argument may still be required.

When the spectral floor is supplied asymptotically only by a Gaussian nugget, Corollary~\ref{cor:nugget_supported_certificate_frontier} shows that root-$M_N$ raw-GATE IC-1 and this row-sum route jointly require $\alpha_N(s_N-2r_Y)=o_p(\lambda_N)=o_p(M_N^{-1})$. This joint rate need not hold when the structured covariance itself supplies a fixed spectral floor.

\subsection{Why global-factor samplers may fail the certificate}
\label{subsec:global_factor_barriers}

For a correlation matrix $\Omega\in\mathbb R^{L\times L}$ define
\begin{equation}
r_{\rm eff}(\Omega):=\frac{\{\tr(\Omega)\}^2}{\|\Omega\|_F^2}=\frac{L^2}{\|\Omega\|_F^2},
\label{eq:effective_rank}
\end{equation}
and let $D_N^{\Zcal}(h)=\max_{a\in\Zcal_N}|\{c\in\Zcal_N:d_N^{\Zcal}(a,c)\le h\}|$.

The row-mass requirement is intentionally stronger than pairwise decay. The next results show why weaker summaries---small pairwise correlations, bounded nominal factor dimension, or even increasing effective rank---cannot generally replace a uniform bound over separated sets.

\begin{proposition}[Global-factor certifiability barriers]
\label{prop:global_factor_certifiability}
The following statements hold.
\begin{enumerate}[label=(\roman*),leftmargin=2.6em]
\item \emph{Far-mass/effective-rank barrier.} For every $h\ge0$,
\label{prop:rank_barrier}
\begin{equation}
\sum_{d_N^{\Zcal}(a,c)>h}\Omega_{N,ac}^2
\ge \frac{L_N^2}{r_{\rm eff}(\Omega_N)}-L_ND_N^{\Zcal}(h)
\ge \frac{L_N^2}{\operatorname{rank}(\Omega_N)}-L_ND_N^{\Zcal}(h).
\label{eq:rank_barrier_bound}
\end{equation}
Moreover, suppose there is a sequence of nonnegative integers $h_N$ such that
\begin{equation}
\frac{D_N^{\Zcal}(h_N)}{L_N}\xrightarrow{p}0,
\qquad
\epsilon_N^{\rm far}:=
\max_{d_N^{\Zcal}(a,c)>h_N}|\Omega_{N,ac}|\xrightarrow{p}0.
\label{eq:rank_barrier_locality_conditions}
\end{equation}
Then $r_{\rm eff}(\Omega_N)\xrightarrow{p}\infty$. Thus a bounded-rank sequence cannot have both sparse $h_N$-balls and uniform far-pair decay. For the Kronecker sampler,
\begin{equation}
\operatorname{rank}(\Omega_N)\le d_{I,N}d_{T,N}.
\label{eq:kronecker_rank_bound}
\end{equation}
\item \emph{Exact shared-factor saturation.} If $\Omega_N=U_NU_N^\top$ and disjoint blocks $A,B$ each span the full factor space, then
\label{cor:shared_factor_saturation}
\begin{equation}
\rho_{\Omega_N}(A,B)=1.
\label{eq:shared_factor_saturation}
\end{equation}
\item \emph{Weak common-factor saturation.} If
\begin{equation}
\Omega_N=a_NU_NU_N^\top+(1-a_N)I_{L_N},\qquad0<a_N<1,
\label{eq:weak_common_factor_model}
\end{equation}
with unit-norm rows and full-column-rank block matrices, then
\label{prop:weak_common_factor_saturation}
\begin{equation}
\rho_{\Omega_N}(A,B)\ge
\frac{a_N}
{\sqrt{a_N+(1-a_N)/\lambda_{\min}(U_{N,A}^\top U_{N,A})}
 \sqrt{a_N+(1-a_N)/\lambda_{\min}(U_{N,B}^\top U_{N,B})}}.
\label{eq:weak_common_factor_lower_bound}
\end{equation}
If both block Gram eigenvalues are of their natural linear orders and $a_N\min(|A|,|B|)\to\infty$, the canonical correlation tends to one even when $a_N\to0$ and $r_{\rm eff}(\Omega_N)\to\infty$.
\end{enumerate}
\end{proposition}

\paragraph{Meaning.}
Small pairwise correlations, nominally low factor dimension, and favorable-looking rank summaries can coexist with strong dependence between large separated sets. These examples refute those summaries as universal sufficient conditions; they do not establish that every resulting estimator is nonnormal or that no factor-aware limit theory exists. \emph{Proof.} See Appendix~\ref{subsec:gaussian_certificate_proofs}.

\begin{corollary}[Separated tight-frame-plus-nugget obstruction]
\label{cor:tight_frame_block_saturation}
For every fixed factor dimension $r$, there exist size sequences $L_N\to\infty$, integer separations $s_N\to\infty$ with $s_N=o(L_N)$, connected assignment graphs with $\Ocal_N=\Zcal_N$, and blocks $A_N,B_N$ satisfying
\begin{equation}
|A_N|,|B_N|\asymp L_N,
\qquad d_N^{\Zcal}(A_N,B_N)\ge s_N>2r_Y,
\label{eq:tight_frame_separation}
\end{equation}
and unit-row tight frames $U_N\in\mathbb R^{L_N\times r}$ such that, for
$\Omega_N=L_N^{-1/4}U_NU_N^\top+(1-L_N^{-1/4})I_{L_N}$,
\begin{equation}
\max_{a\ne c}|\Omega_{N,ac}|\to0,
\qquad r_{\rm eff}(\Omega_N)\asymp rL_N^{1/2}\to\infty,
\qquad \theta_N^{\rm GC}(s_N)\to1.
\label{eq:tight_frame_certificate_failure}
\end{equation}
Thus pairwise decay and diverging effective rank do not certify decay of the latent Gaussian separated-set coefficient. The graph/frame construction and calculation are in Appendix~\ref{app:conditional_and_primitive_extensions}. This is a certifiability obstruction; by itself it does not prove that thresholded HT contributions have maximal correlation tending to one or that every estimator is nonnormal.
\end{corollary}

\paragraph{Meaning.}
Even increasing effective rank can hide a recoverable shared block factor. Effective rank is therefore a descriptive diagnostic, not a substitute for a uniform bound over separated sets. \emph{Proof.} See Appendix~\ref{subsec:gaussian_certificate_proofs}.

The matrix inequalities in Proposition~\ref{prop:global_factor_certifiability} use standard frame-potential and Gaussian-factor algebra; the inferential implication is the new point. Small pairwise covariance and growing effective rank are diagnostics, not a uniform separated-set condition. Sampled block correlations are lower bounds on the supremum in \eqref{eq:gaussian_theta_certificate}: one large draw can refute a proposed decay rate, but a finite collection of small draws cannot certify it. Implementations should therefore report \texttt{CERTIFIED}, \texttt{NOT CERTIFIED}, or \texttt{ROUTE REFUTED}, together with the proof object or stress test supporting that route-specific verdict.

The fixed-dimensional Kronecker sampler and its soft pairwise locality penalty are therefore not automatically certified. We call the penalized implementation \emph{COSTA-LocalPenalty}; it acquires IC-0, IC-1, or IC-2 status only after the corresponding conditions are verified. Failure of this latent Gaussian certificate does not prove that every statistic under the design is nonnormal; it shows only that the uniform local graph-$\psi$ route is unavailable without additional statistic-specific or factor-aware theory.

\subsection{When centered inference cannot target the raw GATE}
\label{subsec:bias_locality_trilemma}

The distinction between IC-0 and IC-1 is structural. The next result combines the path disagreement inequality in Proposition~\ref{prop:path_bias_lower_bound} with the same topology rate used for the centered CLT and an explicit assignment-decorrelation premise. It is sampler-agnostic and concerns only the raw HT estimator for the all-treated versus all-control GATE.

IC-0 centers only at the design expectation. To decide whether the same Gaussian approximation is causal, the expected-cut bias must be compared with the standard-error scale. The next theorem performs exactly this comparison on a canonical path.

\begin{theorem}[Raw-GATE incompatibility under far-assignment decorrelation on a path]
\label{thm:path_trilemma}
For each $N$, let $\Ocal_N=\Zcal_N=\{1,\ldots,M_N\}$ and let $\mathcal G_N$ be the path graph. Assume:
\begin{enumerate}[label=(\roman*),leftmargin=2.2em]
\item the common-marginal assignment law, structural model, and disturbance-orthogonality conditions of Proposition~\ref{prop:expected_cut_geometry} hold; $w_{Nk}=1$; writing $\vartheta_N$ for the structural exposure coefficient, $\vartheta_N\ge\vartheta_0>0$; and the nearest-neighbor exposure weights satisfy $H_{N,k,k+1}\ge h_0>0$ for $k=1,\ldots,M_N-1$;
\item $\sigma_N^2\asymp M_N$ and Assumption~\ref{ass:weights_moments} holds;
\item Assumption~\ref{ass:weak_dependence_rates} holds with its deterministic positive-integer sequence $m_N\to\infty$, satisfying $m_N<M_N$ eventually; and
\item assignment correlations decorrelate at that same scale:
\begin{equation}
\varepsilon_N^Z:=\sup_{1\le i\le M_N-m_N}|\Rnorm_{N,i,i+m_N}|\xrightarrow{p}0.
\label{eq:far_assignment_decorrelation}
\end{equation}
\end{enumerate}
Then \eqref{eq:clt_centered} holds,
\begin{equation}
m_N=o(M_N^{1/2}),
\label{eq:path_radius_upper_rootM}
\end{equation}
and
\begin{equation}
\frac{|\mathbb E_N\widehat\tau_N^{HT}-\tau_N|}{\sqrt{V_N}}\xrightarrow{p}\infty.
\label{eq:path_trilemma_divergence}
\end{equation}
Thus the sequence is IC-0 but not IC-1.
\end{theorem}

\paragraph{Meaning.}
\emph{Question:} does a local centered CLT automatically justify raw-GATE inference? \emph{Role:} in the stated positive path and root-$M_N$ regime, the theorem compares design bias directly with the standard-error scale and shows a route-specific incompatibility between IC-0 locality and raw-HT IC-1. \emph{Boundary:} debiasing, different estimands, nonlocal limits, signed cancellation, or other variance scales are outside the obstruction. \emph{Proof.} See Appendix~\ref{subsec:causal_centering_proofs}.

\begin{corollary}[Gaussian-certificate path incompatibility]
\label{cor:gaussian_path_trilemma}
In Theorem~\ref{thm:path_trilemma}, suppose assignments are thresholded Gaussian, Assumption~\ref{ass:local_outcomes} holds, and the KMS rates are verified using $\theta_N^{\rm GC}$. Then the second rate implies \eqref{eq:far_assignment_decorrelation}, so the theorem applies.
\end{corollary}

\paragraph{Meaning.}
This corollary instantiates the path obstruction with the Gaussian separated-set bound. It says that a particular sufficient Gaussian-locality approach cannot simultaneously certify raw-GATE centering in the stated regime; it is not a universal impossibility theorem for thresholded-Gaussian designs. \emph{Proof.} See Appendix~\ref{subsec:causal_centering_proofs}.

The scope is deliberately narrow: a canonical path geometry, positive nonvanishing local interference, the raw HT estimator, root-$M$ scaling, same-scale far-assignment decorrelation, and the stated local graph-$\psi$ approach. The theorem does not cover signed spillovers, debiased estimators, alternative exposure estimands, superlinear sum-variance regimes, or global-factor stable/mixed-normal theory. Appendix~\ref{app:causal_corrections_extensions} gives the variance-window refinement, a fixed-range bias floor, and external/model-assisted debiasing with pilot-uncertainty propagation.

\paragraph{Restoring causal centering.}
The path theorem identifies a missing design-bias component rather than a failure of centered normality. Under the exact two-channel model, define
\begin{equation}
b_{k,N}^{(w)}(\Rnorm_N)=\sum_{a,c}w_{Na}H^k_{N,ac}(\Rnorm_{N,ac}-1),\qquad k\in\{\comp,\lag\},
\label{eq:unnormalized_bias_loadings}
\end{equation}
set $c_N=W_N^{-1}\{b_{\comp,N}^{(w)}(\Rnorm_N),b_{\lag,N}^{(w)}(\Rnorm_N)\}^\top$, and let $\delta_N=(\gamma_N,\eta_N)^\top$. The exact bias is $\delta_N^\top c_N$, separating a known design loading from exposure coefficients that must be estimated or bounded.

\begin{proposition}[External or model-assisted debiasing]
\label{prop:external_debiasing}
Define the known design-bias loading and exposure-effect vector
\begin{equation}
c_N:=\frac1{W_N}
\begin{pmatrix}
 b^{(w)}_{\comp,N}(\Rnorm_N)\\
 b^{(w)}_{\lag,N}(\Rnorm_N)
\end{pmatrix},
\qquad
\delta_N:=\begin{pmatrix}\gamma_N\\\eta_N\end{pmatrix}.
\label{eq:design_bias_loading_vector}
\end{equation}
Let $\widehat\delta_N$ be estimated from pre-assignment information or an external sample and define
\begin{equation}
\widehat\tau_N^{\rm BC}:=\widehat\tau^{HT}_{w,\tilde\alpha}-\widehat\delta_N^\top c_N.
\label{eq:debiased_ht_estimator}
\end{equation}
If the design-centered CLT holds and
\begin{equation}
c_N^\top(\widehat\delta_N-\delta_N)=o_p(\sqrt{V_N}),
\label{eq:debiased_nuisance_rate}
\end{equation}
then
\begin{equation}
\frac{\widehat\tau_N^{\rm BC}-\tau_{w,N}}{\sqrt{V_N}}\Rightarrow N(0,1).
\label{eq:debiased_causal_clt}
\end{equation}
\end{proposition}

The proof subtracts the exact design bias and applies Slutsky. Nonnegligible pilot error must instead be propagated under the joint pilot--experiment law in Proposition~\ref{prop:external_pilot_uncertainty}; other remedies include changing the estimand, shrinking exposure effects, or using a valid bias envelope. The preceding incompatibility concerns IC-1, not variance estimation. Once causal centering has been restored---or when design-centered inference is itself the target---the remaining and logically separate IC-2 task is to estimate $V_N$.

\subsection{Feasible studentization with heterogeneous contribution means}
\label{subsec:network_hac}

Let $K:[0,\infty)\to\mathbb R$ be bounded, compactly supported on $[0,1]$, and normalized by $K(0)=1$; set $K(\infty)=0$. For a bandwidth $b_N\to\infty$, let $\mathsf K_N(a,c)=K\{d_N(a,c)/b_N\}$ on the connected dependence metric, and define the maximum bandwidth-ball size
\begin{equation}
D_N(b_N):=\max_{a\in\Ocal_N}|\mathcal B_N(a;b_N)|.
\label{eq:hac_bandwidth_ball_size}
\end{equation}
Put
\begin{equation}
\Sigma_N:=\sigma_N^2/M_N=M_NV_N,
\qquad
\widehat\Sigma_N^{\rm or}:=M_N^{-1}\sum_{a,c}\mathsf K_N(a,c)\xi_{Na}\xi_{Nc}.
\label{eq:oracle_network_hac}
\end{equation}

\begin{proposition}[Connected network-HAC specialization]
\label{prop:oracle_network_hac_connected}
\label{prop:oracle_network_hac}
Under Assumption~\ref{ass:contribution_metric_compatibility}, suppose graph-$\psi$ dependence and the moment condition hold, $\inf_N\Sigma_N>0$, and
\begin{align}
\sum_{s\ge1}|K(s/b_N)-1|\delta_N^{\partial}(s;1)\theta_N(s)^{1-2/\nu}&\xrightarrow{p}0,
\label{eq:hac_kernel_bias_rate}\\
M_N^{-1}\sum_{s\ge0}c_N(s,b_N;2)\theta_N(s)^{1-4/\nu}&\xrightarrow{p}0.
\label{eq:hac_stochastic_rate}
\end{align}
Then
\begin{equation}
\widehat\Sigma_N^{\rm or}-\Sigma_N\xrightarrow{p\mid\mathcal F}0,
\qquad
(\widehat\Sigma_N^{\rm or}/M_N)/V_N\xrightarrow{p\mid\mathcal F}1.
\label{eq:oracle_hac_consistency}
\end{equation}
\end{proposition}

\paragraph{Meaning.}
Given the true heterogeneous contribution means and separately verified kernel approximation and stochastic rates, the oracle quadratic form estimates the design-centered long-run variance. It does not justify replacing those means by a global sample mean or repair a failure of IC-1. \emph{Proof.} See Appendix~\ref{subsec:hac_proof_appendix}.

This is again an existing network-HAC specialization. Appendix~\ref{app:disconnected_hac_extensions} gives the disconnected-metric extension, which separately controls omitted cross-component covariance and cross-component covariance of included pair-products; empty finite shells are not sufficient. A generic radial graph-distance kernel need not generate a positive-semidefinite matrix, so oracle-HAC validity and positive semidefiniteness are separate requirements. The latter is invoked only for the primitive feasible-centering route below.

Under the exact two-channel model,
\begin{equation}
\bar q_{Na}=\omega_{Na}\left[\beta_{Na}
+\gamma_N\sum_cH^{\comp}_{N,ac}\Rnorm_{N,ac}
+\eta_N\sum_cH^{\lag}_{N,ac}\Rnorm_{N,ac}\right].
\label{eq:contribution_mean_linear_model}
\end{equation}
Thus $\bar q_{Na}$ is design-specific and generally heterogeneous; it is not an untreated baseline regression. For a structured estimate $\widehat q_N^0$, define
\begin{equation}
\widehat\xi_{Na}^{\rm MA}=Q_{Na}-\widehat q_{Na}^0,\qquad
\widehat\Sigma_N^{\rm MA}=M_N^{-1}(\widehat\xi_N^{\rm MA})^\top\mathsf K_N\widehat\xi_N^{\rm MA},\qquad
\widehat V_N^{\rm MA,+}=[\widehat\Sigma_N^{\rm MA}]_+/M_N.
\label{eq:model_assisted_hac}
\end{equation}

\medskip
\noindent\textbf{Oracle-HAC consistency does not by itself justify global demeaning.} The oracle form uses the true cell-specific means $\bar q_{Na}$. Estimating these heterogeneous, design-specific means is a separate first-stage problem, addressed by the next result.

\begin{proposition}[Feasible heterogeneous-mean centering]
\label{prop:model_assisted_hac}
Let $r_N=\widehat q_N^0-\bar q_N$. Under Proposition~\ref{prop:oracle_network_hac_connected}, if
\begin{equation}
r_N^\top\mathsf K_Nr_N=o_p(\sigma_N^2),\qquad
r_N^\top\mathsf K_N\xi_N=o_{p\mid\mathcal F}(\sigma_N^2),
\label{eq:centering_error_conditions}
\end{equation}
then $\widehat V_N^{\rm MA,+}/V_N\xrightarrow{p\mid\mathcal F}1$.
\end{proposition}

\paragraph{Meaning.}
This proposition isolates the additional first-stage problem required for feasible HAC: estimated heterogeneous design means must be small in both the kernel quadratic norm and the associated cross term. PSD kernels provide a useful seminorm inequality, but the result remains a sufficient condition rather than a default recipe. \emph{Proof.} See Appendix~\ref{subsec:hac_proof_appendix}.

The next corollary is an existence-style sufficient route. It shows how a pre-randomization, design-specific mean model can make an already valid oracle HAC feasible under a PSD kernel; it is not the default empirical procedure and does not establish oracle-HAC validity by itself.

\begin{corollary}[Abstract PSD-kernel external-pilot sufficient route to IC-2]
\label{cor:external_pilot_centering}
Under Proposition~\ref{prop:oracle_network_hac_connected}, suppose an external pilot is completed before the current design randomization and its fitted quantities are included in $\mathcal F_N$. Assume
\begin{equation}
\bar q_N=\mathsf P_N\vartheta_{0,N},\qquad
\lambda_{\max}(M_N^{-1}\mathsf P_N^\top\mathsf P_N)\le C_P,
\label{eq:external_pilot_mean_model}
\end{equation}
and an external pilot of effective size $n_{0,N}$ gives, under the outer pilot/environment law,
\begin{equation}
\|\widehat\vartheta_N-\vartheta_{0,N}\|_2^2=O_p(d_{q,N}/n_{0,N}),
\qquad \widehat q_N^0=\mathsf P_N\widehat\vartheta_N.
\label{eq:external_pilot_parameter_rate}
\end{equation}
If $\mathsf K_N\succeq0$, $\|\mathsf K_N\|_{\rm op}\le\Lambda_N$, and
\begin{equation}
\frac{\Lambda_NM_Nd_{q,N}}{n_{0,N}\sigma_N^2}\xrightarrow{p}0,
\label{eq:external_pilot_hac_rate}
\end{equation}
then $\widehat V_N^{\rm MA,+}/V_N\xrightarrow{p\mid\mathcal F}1$. In the root-$M_N$ regime with $\Lambda_N\lesssim D_N(b_N)$, it is enough that
\begin{equation}
\frac{D_N(b_N)d_{q,N}}{n_{0,N}}\xrightarrow{p}0.
\label{eq:external_pilot_rootM_rate}
\end{equation}
Together with IC-1 this gives IC-2. The pilot must estimate the contribution mean under the actual assignment law. Positive semidefiniteness is essential for this primitive seminorm argument; Appendix~\ref{app:disconnected_hac_extensions} gives a stronger operator-norm route for a possibly indefinite distance kernel. The identical PSD argument applies to Proposition~\ref{prop:oracle_network_hac_disconnected} after replacing $\mathsf K_N$ by $\widetilde{\mathsf K}_N$ and verifying its two infinite-shell remainders. This corollary is an abstract existence result, not a default HAC method.
\end{corollary}

\paragraph{Meaning.}
An external, pre-randomization, design-specific mean model can make an already valid oracle HAC feasible when the stated PSD-kernel rates hold. This is an existence-style IC-2 construction; it neither establishes oracle-HAC validity nor makes a pilot trained under a different assignment law automatically transportable. \emph{Proof.} See Appendix~\ref{subsec:hac_proof_appendix}.

\paragraph{An explicit PSD matrix requiring separate oracle-HAC verification.}
Let $\{\mathcal C_{Ng}\}_{g\in\mathcal I_N}$ be overlapping graph blocks and let $\omega_{Ng}\ge0$. The Gram kernel
\begin{equation}
\mathsf K_N^{\rm blk}(a,c)
=\sum_{g\in\mathcal I_N}\omega_{Ng}\ind\{a,c\in\mathcal C_{Ng}\}
=(\mathsf S_N\mathsf S_N^\top)(a,c),
\qquad
\mathsf S_N(a,g)=\sqrt{\omega_{Ng}}\ind\{a\in\mathcal C_{Ng}\},
\label{eq:overlapping_block_psd_kernel}
\end{equation}
is positive semidefinite by construction. This example shows that the PSD premise is nonempty. It does not, by itself, establish oracle-HAC consistency: the block sequence and weights must still satisfy an appropriate kernel-approximation condition and the stochastic rate in Proposition~\ref{prop:oracle_network_hac_connected} (or their separately proved analogues).

\paragraph{Exact linear-model pilot map.}
Under \eqref{eq:contribution_mean_linear_model}, define the design-specific exposure summaries
\begin{equation}
s^{\comp}_{Na}=\sum_cH^{\comp}_{N,ac}\Rnorm_{N,ac},
\qquad
s^{\lag}_{Na}=\sum_cH^{\lag}_{N,ac}\Rnorm_{N,ac}.
\label{eq:pilot_design_specific_exposure_summaries}
\end{equation}
If $\beta_{Na}=x_{Na}^\top\delta_{0,N}$ for pre-treatment regressors $x_{Na}$, a concrete mean model has row
\begin{equation}
\mathsf P_N(a,\cdot)=\omega_{Na}\bigl(x_{Na}^\top,s^{\comp}_{Na},s^{\lag}_{Na}\bigr),
\qquad
\vartheta_{0,N}=\bigl(\delta_{0,N}^\top,\gamma_N,\eta_N\bigr)^\top.
\label{eq:exact_linear_pilot_map}
\end{equation}
A pilot may come from the same assignment law, transportable historical experiments, or a validated external structural model. Because the summaries in \eqref{eq:pilot_design_specific_exposure_summaries} depend on the current assignment law, retuning the covariance design changes $\mathsf P_N$ and generally requires recomputing the summaries and re-estimating, or separately justifying transport of, the pilot coefficients.

\begin{corollary}[Studentized inference]
\label{cor:studentized_model_assisted}
Under Theorem~\ref{thm:cell_clt}, a certified oracle-HAC route, and Proposition~\ref{prop:model_assisted_hac}, define the displayed statistic as zero on $\{\widehat V_N^{\rm MA,+}=0\}$. Then
\begin{equation}
\frac{\widehat\tau^{HT}_{w,\tilde\alpha}-\mathbb E_N\widehat\tau^{HT}_{w,\tilde\alpha}}
{\sqrt{\widehat V_N^{\rm MA,+}}}\rightsquigarrow_{\mathcal F} N(0,1).
\label{eq:studentized_centered_clt}
\end{equation}
If \eqref{eq:bias_negligible} also holds, the expectation may be replaced by $\tau_{w,N}$.
\end{corollary}

\paragraph{Meaning.}
Once IC-0, IC-1, oracle-HAC validity, and feasible heterogeneous-mean centering all hold, Slutsky yields the usual studentized statistic. The corollary packages the three layers; none of them is implied merely by optimizing the COSTA point-estimation objective. \emph{Proof.} See Appendix~\ref{subsec:hac_proof_appendix}.

\paragraph{Design-aware covariance-tail correction.}
A local HAC can omit a material nonlocal covariance tail. Let $\Sigma_{K,N}$ be its population target and $\mathcal T_N=\Sigma_N-\Sigma_{K,N}$. A plug-in potential-outcome surface and fresh draws from the known design estimate full and local design variances, yielding $\widehat{\mathcal T}_{\rm pl}$.

\begin{proposition}[Design-aware covariance-tail correction]
\label{prop:design_aware_hac}
If $\widehat\Sigma_{\rm loc}-\Sigma_{K,N}=o_p(\Sigma_N)$ and $\widehat{\mathcal T}_{\rm pl}-\mathcal T_N=o_p(\Sigma_N)$, then
\[
\widehat\Sigma_{DA}^{\rm sgn}=\widehat\Sigma_{\rm loc}+\widehat{\mathcal T}_{\rm pl}
\]
is ratio-consistent. Define
\begin{equation}
\widehat\Sigma_{DA,\mathrm{raw}}^{+}=\widehat\Sigma_{\rm loc}+[\widehat{\mathcal T}_{\rm pl}]_+,
\qquad
\widehat V_{DA}^{+}=\frac{[\widehat\Sigma_{DA,\mathrm{raw}}^{+}]_+}{M_N}.
\label{eq:design_aware_positive_tail}
\end{equation}
Then $\widehat V_{DA}^{+}$ is nonnegative and one-sided conservative. It is ratio-consistent if and only if the negative part $[-\mathcal T_N]_+$ is $o_p(\Sigma_N)$.
\end{proposition}

Sufficient plug-in conditions control the randomization variance of the aggregate contribution error and, for a positive-semidefinite local kernel, its quadratic norm. Monte Carlo estimates from independent design draws are asymptotically negligible under bounded standardized fourth moments. Fixed-outcome plug-ins are only diagnostic because interference changes outcomes under counterfactual assignments. The correction is therefore model-assisted and cannot substitute for IC-0. Appendix~\ref{subsec:design_aware_variance} gives the full conditions and proof. The experiments evaluate signed and outer-truncated positive-tail corrections together with heterogeneous-mean local HAC, reporting coverage, interval width, standard-error calibration, and studentized quantiles.

\paragraph{What Section~\ref{sec:inference} establishes.}
Assumption~\ref{ass:local_outcomes} identifies the assignments generating each contribution. Propositions~\ref{prop:gaussian_canonical_certificate} and~\ref{prop:spectral_rowsum_certificate} control dependence between separated generating blocks, while Corollary~\ref{cor:kronecker_rowsum} reduces the primitive check to spatial and temporal factors. Theorem~\ref{thm:cell_clt} gives IC-0. The exact expected-cut bias determines whether IC-1 holds; Theorem~\ref{thm:path_trilemma} exhibits a canonical failure, and Proposition~\ref{prop:external_debiasing} gives a debiasing route. Propositions~\ref{prop:oracle_network_hac_connected} and~\ref{prop:model_assisted_hac} yield model-assisted IC-2 through Corollary~\ref{cor:studentized_model_assisted}, while Proposition~\ref{prop:design_aware_hac} supplies a separate covariance-tail correction.

The ordering prevents three conflations: pairwise covariance optimization is not setwise weak dependence; a centered Gaussian limit need not be centered at the GATE; and a local covariance quadratic need not estimate variance when contribution means are heterogeneous or nonlocal covariance remains. The same modular framework applies to sparse, block, Kronecker, locally factored, and other covariance sequences after a design-specific dependence bound is verified. A failed sufficient certificate proves neither nonnormality nor impossibility; it only shows that the proposed route has not established the next IC layer.

\section{Experiments}
\label{sec:experiments}

\subsection{Validation setting: item-side multiple-unit switchbacks}
\label{subsec:validation_setting}

The empirical study instantiates the general network--temporal framework in multiple-item switchback experiments for recommendation systems. Here the spatial units $i$ are items connected through a competition graph, and the temporal indices $b$ are intervention blocks. This setting is a useful test bed because item traffic is highly unequal, treatments can shift demand across substitutes, and carryover can persist across blocks. Multiple-unit switchbacks and regular balanced switchback designs were developed for related e-commerce settings \citep{missault2025rbsd}; our experiments use this application to evaluate joint covariance optimization rather than to restrict the theory to switchbacks. For reproduction, we provide the code at \url{https://github.com/Cqyiiii/COSTA-Covariance-Optimized-Spatiotemporal-Treatment-Allocation}.

We use RetailRocket and MovieLens as two semi-synthetic benchmarks. Real timestamped interactions determine traffic weights, item embeddings, and competition graphs, while potential outcomes are generated from known linear, nonlinear, and demand-substitution models. The global treatment effect is therefore exactly computable, allowing direct evaluation of bias, standard deviation, and RMSE. RetailRocket contains visitor--item events in an e-commerce environment \citep{retailrocket}; MovieLens provides timestamped user--movie ratings and is documented by \citet{harper2015movielens}. The main application results evaluate point estimation. \IfFileExists{tables/tab_inference_validation.tex}{An inference analysis, reported in Subsection~\ref{subsec:inference_experiments}, evaluates coverage and interval width for the heterogeneous-mean and design-aware variance estimators without rerunning the point-estimation experiments.}{}

\subsection{Competition graph and potential outcome models}
\label{subsec:experiment_pom}

For each dataset, we sort interactions by time and use a pre-treatment prefix for graph construction and baseline estimation. Event weights define historical interaction intensity. In RetailRocket, the default weights are view $=1$, add-to-cart $=3$, and transaction $=8$; in MovieLens, ratings are converted to interaction weights. We retain the top $N$ items by weighted interaction volume and split the simulation window into $B$ consecutive blocks.

The default competition graph is built from collaborative-filtering item embeddings. Let $X_{ui}$ be the weighted user--item matrix from the graph-construction period. We compute truncated-SVD item embeddings $g_i\in\mathbb R^{64}$ and cosine similarities $s_{ij}$. For each item $i$, we keep its top-$K$ most similar neighbors, truncate negative similarities, and row-normalize:
\[
    A_{ij}=\frac{\max(s_{ij},0)\ind\{j\in\mathrm{NN}_K(i)\}}
    {\sum_{\ell\in\mathrm{NN}_K(i)}\max(s_{i\ell},0)+10^{-12}}.
\]
The default is $K=10$. The graph is a pre-treatment approximation to substitution or demand overlap; it is not assumed to be a perfectly specified structural response model. The top-$K$ ablation below studies the sensitivity of COSTA to graph sparsity.

We evaluate three potential outcome models, and all main tables report them separately. Let $m_{ib}(Z)$ denote the conditional mean outcome under a simulated data-generating process. For any such mean function, the finite-population target used in the experiments is the traffic-weighted global contrast
\begin{equation}
    \tau_{w,m}=\frac{1}{W_N}
    \sum_{(i,b)\in\Ocal}w_{ib}\{m_{ib}(\one)-m_{ib}(\bm 0)\},
    \qquad
    W_N=\sum_{(i,b)\in\Ocal}w_{ib},
    \label{eq:experiment_true_gate}
\end{equation}
where $w_{ib}$ is the pre-treatment item--block traffic weight. Setting $w_{ib}\equiv1$ gives the unweighted GATE. This definition is used to compute the oracle target in all three models, so the Monte Carlo bias is measured against the same global treatment regime that motivates the covariance design.

\paragraph{Linear model.}
The first model matches the covariance-design working model. For $b=2,\ldots,B$,
\begin{equation}
    m^{\mathrm{lin}}_{ib}(Z)=\mu_{ib}+\beta_iZ_{ib}
    +\gamma\sum_jA_{ij}Z_{jb}
    +\eta Z_{i,b-1}.
    \label{eq:experiment_linear_model}
\end{equation}
The graph in \eqref{eq:experiment_linear_model} is the same pre-treatment competition graph used by the design; $\gamma$ is contemporaneous network spillover and $\eta$ is one-lag own-item temporal carryover. We allow heterogeneous direct effects, with default form $\beta_i=\beta_0(1+0.1\xi_i)$ and $\xi_i\sim N(0,1)$, and add mean-zero simulation noise to observed outcomes. The true linear target is
\begin{equation}
    \tau_{w,\mathrm{lin}}=\frac{1}{W_N}
    \sum_{(i,b)\in\Ocal}w_{ib}
    \left(\beta_i+\gamma\sum_jA_{ij}+\eta\right).
    \label{eq:experiment_tau_linear}
\end{equation}
If $A_I$ is row-normalized and the weights do not concentrate on unusual rows, this target is close to the weighted average direct effect plus $\gamma+\eta$.

\paragraph{Nonlinear saturated competition model.}
For compactness, write $C_{ib}(Z)=\sum_jA_{ij}Z_{jb}$ and $C^\star_{ib}(Z)=\sum_jA^\star_{ij}Z_{jb}$. The second model deliberately violates the additive linear exposure assumption while preserving the same observed competition graph:
\begin{equation}
    m^{\mathrm{nl}}_{ib}(Z)=\mu_{ib}+\beta_iZ_{ib}+\gamma_i\tanh\{\kappa C_{ib}(Z)\}+\lambda_{\mathrm{int}}Z_{ib}C_{ib}(Z)+\eta_iZ_{i,b-1}.
    \label{eq:experiment_nonlinear_model}
\end{equation}
The hyperbolic-tangent term creates saturation in network exposure, $\gamma_i$ and $\eta_i$ introduce heterogeneous spillover and carryover strengths, and $\lambda_{\mathrm{int}}Z_{ib}C_{ib}(Z)$ creates interaction between own treatment and neighbors' treatments. The true target $\tau_{w,\mathrm{nl}}$ is computed from \eqref{eq:experiment_true_gate} by evaluating \eqref{eq:experiment_nonlinear_model} at $Z=\one$ and $Z=\bm 0$.

\paragraph{Demand-substitution model.}
The third model is closer to a demand-substitution process and introduces stronger misspecification. Let $L_{ib}(Z)=\eta_1Z_{i,b-1}+\eta_2\ind\{b\ge3\}Z_{i,b-2}-\chi_TC^\star_{i,b-1}(Z)$. The latent demand mean is
\begin{equation}
    \mu^{\mathrm{dem}}_{ib}(Z)=\exp\{\mu^\ell_{ib}+\beta_iZ_{ib}-\chi C^\star_{ib}(Z)+L_{ib}(Z)\}.
    \label{eq:experiment_demand_mean}
\end{equation}
We then generate observed outcomes as
\begin{equation}
    Y_{ib}(Z)\sim \Poisson\{c_{ib}\mu^{\mathrm{dem}}_{ib}(Z)\},
    \label{eq:experiment_demand_poisson}
\end{equation}
where $c_{ib}$ is item--block traffic exposure. The graph $A^\star$ may differ from the design graph $A_I$, for example by using a denser top-$K$ graph or a demand-overlap graph. The specification is motivated by standard count-demand models and by the economics/operations literature on substitutable products: Poisson-type outcome models are common for count responses, and competitor substitution is a central mechanism in inventory and demand models \citep{camerontrivedi1986count,netessinerudi2003substitution}. The negative coefficients $-\chi$ and $-\chi_T$ encode substitution: treated neighboring items can draw demand away from item $i$ contemporaneously or through lagged competitor pressure. This model is misspecified relative to COSTA because it has an exponential link, graph mismatch, lag-2 own-item carryover, and lag-1 competitor carryover. The target $\tau_{w,\mathrm{dem}}$ is computed from the Poisson means $m^{\mathrm{dem}}_{ib}(Z)=c_{ib}\mu^{\mathrm{dem}}_{ib}(Z)$ under global treatment and global control.

The linear model tests whether the covariance formula behaves as intended when the working approximation is accurate. The nonlinear and demand-substitution models test robustness to saturation, interaction effects, heterogeneous spillovers, graph mismatch, non-Gaussian outcomes, and additional temporal spillovers beyond the one-lag design model.

\subsection{Designs, estimators, and metrics}
\label{subsec:experiment_designs_metrics}

We compare six main designs. \textsc{Independent} assigns each item--block cell independently with probability $p$. \textsc{RBSD} is a regular balanced switchback design that balances assignments across items and blocks but does not use the competition graph; implementation details for general $p$ are given in Appendix~\ref{app:additional_experiments}. \textsc{Cluster} follows graph cluster randomization \citep{ugander2013graph}: it uses Leiden communities \citep{traag2019leiden} of the pre-treatment competition graph and randomizes the clusters independently by block. \textsc{Static-OCD} optimizes graph-aware covariance within each block following optimized covariance design \citep{chen2023ocd}, but samples blocks independently, isolating the value of temporal adaptation. \textsc{COSTA} is the default covariance-optimized network-time design. \textsc{COSTA-LocalPenalty} denotes the finite-sample soft-locality-penalized variant and is abbreviated \textsc{COSTA-LP} in compact figures and tables; in the default RMSE summaries we use the fixed, untuned $\lambda=1$ configuration and suppress the tuning-value suffix. This empirical label should not be read as evidence that the fixed-rank sampler satisfies the formal canonical-correlation certificate. The separate inference-diagnostic grid below is a sensitivity analysis over the simulated locality-penalty values available in the $B=8$ outputs, not a tuning step.

The covariance objective and the exact bias identity are derived for the HT score. The default-setting summaries use $B=8$ and report weighted HT so that the headline RMSE decomposition is aligned with the theory and makes the variance effect of COSTA-LocalPenalty visible. Weighted H{\'a}jek and unweighted DIM are reported side-by-side in Section~\ref{subsec:estimator_comparison}. The optimizer in the reported simulations uses the unweighted design approximation described in Section~\ref{subsec:design_objective}, while the estimand and the weighted estimators are traffic-weighted; this intentionally evaluates robustness of the covariance design under traffic heterogeneity. A fully target-weighted design objective is the coherent extension when launch-time design weights should prioritize high-traffic cells. For each dataset, design, and potential outcome model, we report bias, standard deviation, and RMSE over Monte Carlo assignments. We do not average across potential outcome models: the three models are separate cases throughout the main results and appendix tables.

\subsection{Default-setting RMSE}
\label{subsec:default_results}

Figure~\ref{fig:default_rmse} reports the default setting with $N=10{,}000$, $B=8$, $K=10$, $p=0.5$, $d_I=16$, and $d_T=B-1$. The legend is placed above the panels because RBSD has relatively large RMSE in the nonlinear model. In this longer finite-block setting, COSTA-LocalPenalty often improves RMSE relative to default COSTA under weighted HT. The reason is not that long-range covariance is irrelevant: penalizing it can increase residual design bias. Rather, COSTA already controls the most important exposure-edge terms, while COSTA-LocalPenalty reduces the sampled pairwise locality surrogates in \eqref{eq:local_penalty_exact}; in these simulations that reduction is associated with lower Monte Carlo variance. The resulting empirical association yields a better finite-sample bias--variance balance in the default weighted-HT summaries, without establishing a general variance mechanism.

\begin{figure}[!htbp]
    \centering
    \includegraphics[width=0.98\linewidth]{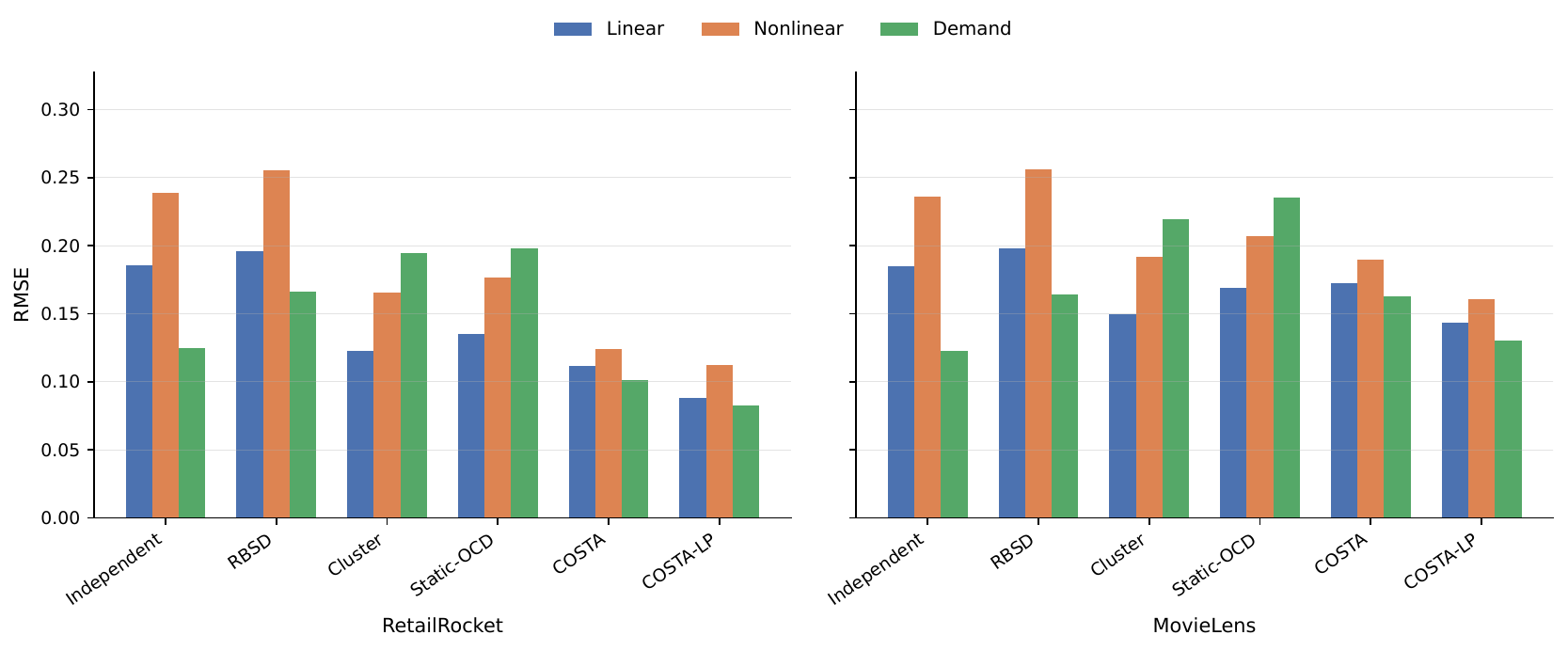}
    \caption{Default-setting RMSE by dataset, design, and potential outcome model. The estimator is weighted HT; the setting is $N=10{,}000$, $B=8$, top-$K=10$, $p=0.5$, Kronecker network dimension $d_I=16$, and an unweighted design objective with traffic-weighted estimands. Bars summarize Monte Carlo assignment draws from the processed RetailRocket and MovieLens simulation outputs.}
    \label{fig:default_rmse}
\end{figure}

\begin{table}[!htbp]
    \centering
    \caption{RetailRocket default-setting statistical performance. Linear, nonlinear, and demand are three separate potential outcome models; each model reports Bias, Std., and RMSE. The estimator is weighted HT.}
    \label{tab:default_retail}
    \small
    \setlength{\tabcolsep}{3.6pt}
    \input{tables/tab_default_retail.tex}
\end{table}

\begin{table}[!htbp]
    \centering
    \caption{MovieLens default-setting statistical performance. Linear, nonlinear, and demand are three separate potential outcome models; each model reports Bias, Std., and RMSE. The estimator is weighted HT.}
    \label{tab:default_movielens_main}
    \small
    \setlength{\tabcolsep}{3.6pt}
    \input{tables/tab_default_movielens.tex}
\end{table}

Tables~\ref{tab:default_retail}--\ref{tab:default_movielens_main} give the corresponding numerical decompositions. On RetailRocket, COSTA-LocalPenalty has RMSE $0.088$ in the linear model, $0.112$ in the nonlinear model, and $0.082$ in the demand model, improving over COSTA primarily through lower standard deviation. On MovieLens, the same pattern appears: COSTA-LocalPenalty reduces RMSE to $0.143$, $0.161$, and $0.130$ in the three models. Independent and RBSD have small standard deviations but large negative bias; Cluster and Static-OCD reduce some network mismatch but do not jointly adapt to carryover; COSTA moves the estimates closer to the global estimand by aligning covariance with both network and temporal exposure edges, while COSTA-LocalPenalty reduces the sampled pairwise locality surrogates and, in these experiments, is associated with lower Monte Carlo variance.

\subsection{Estimator comparison: HT, H{\'a}jek, and DIM}
\label{subsec:estimator_comparison}

The three estimators used in the simulation code differ only in weighting and normalization. Let $X_a^{\mathrm{obs}}=Y_a^{\mathrm{obs}}-\tilde\alpha_a$. The weighted HT estimator is
\begin{equation}
    \widehat\tau^{HT}_{w,\tilde\alpha}
    =\frac{1}{W_N}\sum_{a\in\Ocal}w_a\frac{Z_a-p}{p(1-p)}X_a^{\mathrm{obs}}.
    \label{eq:exp_ht_def}
\end{equation}
The weighted H{\'a}jek estimator replaces inverse-probability normalization by realized treatment mass,
\begin{equation}
    \widehat\tau^{H}_{w,\tilde\alpha}
    =\frac{\sum_{a\in\Ocal}w_aZ_aX_a^{\mathrm{obs}}}{\sum_{a\in\Ocal}w_aZ_a}
    -\frac{\sum_{a\in\Ocal}w_a(1-Z_a)X_a^{\mathrm{obs}}}{\sum_{a\in\Ocal}w_a(1-Z_a)}.
    \label{eq:exp_hajek_def}
\end{equation}
If either realized weighted denominator is zero, the simulation code assigns a fixed prespecified value and records the event; under the overlap and denominator-convergence conditions of Proposition~\ref{prop:hajek_delta}, its probability vanishes asymptotically. The DIM diagnostic is the same ratio estimator with $w_a\equiv1$. Let
$\bar X^{\mathrm{obs}}_{1,0}=\sum_{a\in\Ocal}Z_aX_a^{\mathrm{obs}}/\sum_{a\in\Ocal}Z_a$ and
$\bar X^{\mathrm{obs}}_{0,0}=\sum_{a\in\Ocal}(1-Z_a)X_a^{\mathrm{obs}}/\sum_{a\in\Ocal}(1-Z_a)$. Then
\begin{equation}
    \widehat\tau^{DIM}_{\tilde\alpha}=\bar X^{\mathrm{obs}}_{1,0}-\bar X^{\mathrm{obs}}_{0,0}.
    \label{eq:exp_dim_def}
\end{equation}
Weighted HT and weighted H{\'a}jek target the traffic-weighted GATE; DIM is included as a graph-agnostic unweighted diagnostic. Figure~\ref{fig:estimator_retail} and Table~\ref{tab:estimator_retail} show how the point-estimation comparison changes with normalization: the weighted-HT default highlights the variance benefit of COSTA-LocalPenalty, whereas H{\'a}jek and DIM expose the effect of ratio normalization and target weighting.

\begin{figure}[htbp]
    \centering
    \includegraphics[width=0.95\textwidth]{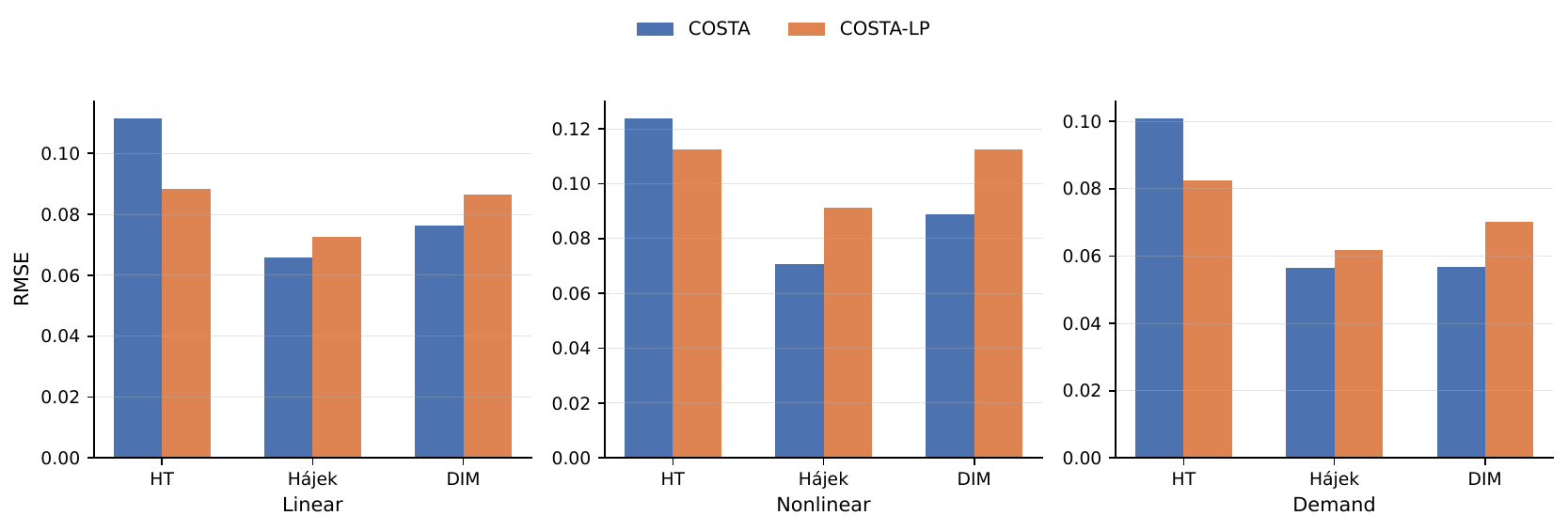}
    \caption{Estimator comparison on RetailRocket in the default $B=8$ setting. Bars report RMSE for HT, H{\'a}jek, and DIM under COSTA and COSTA-LocalPenalty across the three potential outcome models. Weighted HT and weighted H{\'a}jek use traffic weights; DIM is unweighted.}
    \label{fig:estimator_retail}
\end{figure}
\FloatBarrier

\begin{table}[htbp]
    \centering
    \caption{RetailRocket estimator comparison in the default $B=8$ setting. HT and H{\'a}jek are traffic-weighted; DIM is unweighted.}
    \label{tab:estimator_retail}
    \small
    \input{tables/tab_estimator_retail.tex}
\end{table}

\subsection{Optimized covariance diagnostic}
\label{subsec:covariance_diagnostic}

Figure~\ref{fig:covariance_diagnostic} is a stylized five-item, four-block mechanism illustration of selective positive covariance alignment. It shows the network and adjacent-lag directions targeted by the design without claiming to be a subgraph extracted from either empirical dataset. Balance, contrast, and locality terms act on additional pair directions that are not displayed.

\begin{figure}[!htbp]
    \centering
    \includegraphics[width=0.99\linewidth]{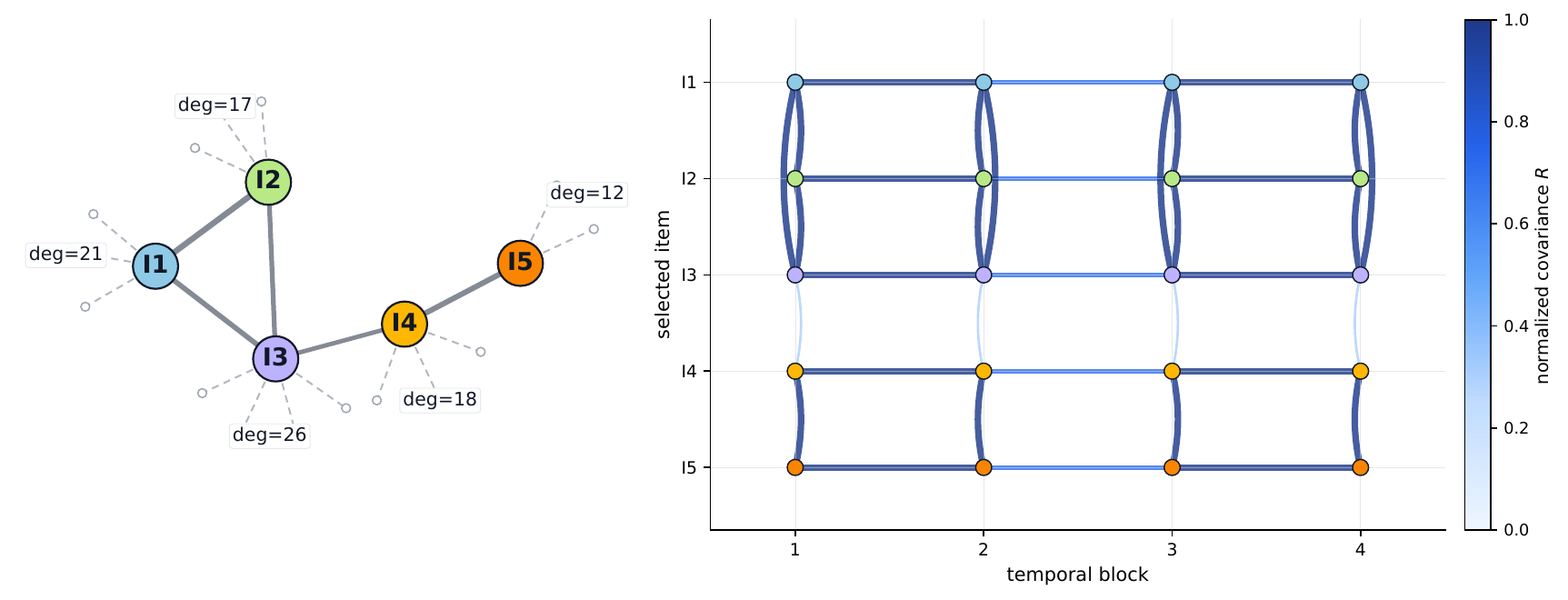}
    \caption{Stylized five-item, four-block illustration of selective positive covariance alignment on the displayed network and lag edges. It is a mechanism illustration rather than a subgraph extracted from RetailRocket or MovieLens; unshown and negative pair directions are not encoded by the $0$--$1$ color scale.}
    \label{fig:covariance_diagnostic}
\end{figure}

\FloatBarrier

\subsection{Setting ablations on RetailRocket}
\label{subsec:setting_ablations}

Figure~\ref{fig:setting_ablation} summarizes the RetailRocket setting ablations for all three potential outcome models.

\begin{figure}[!htbp]
    \centering
    \includegraphics[width=0.88\linewidth]{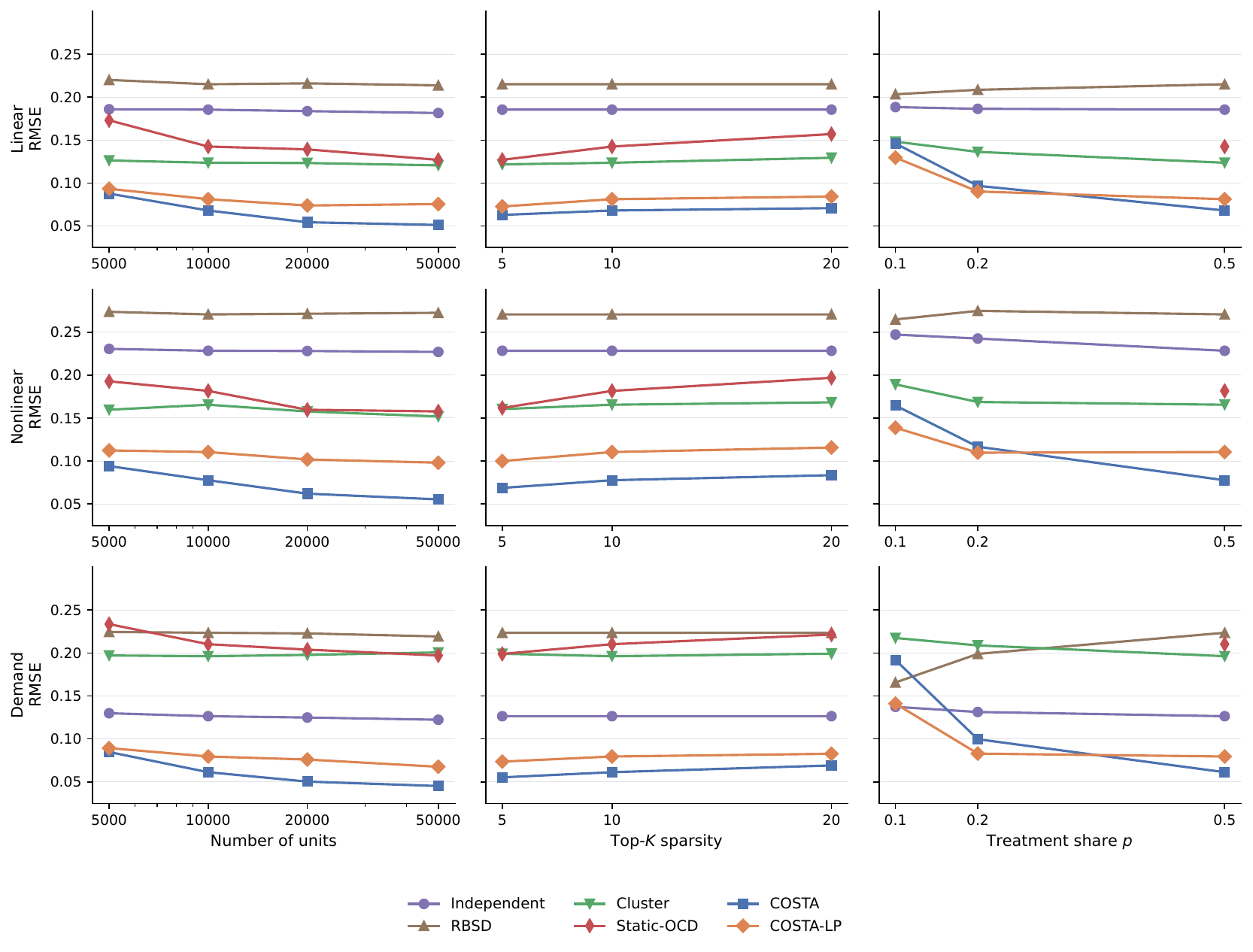}
    \caption{RetailRocket setting ablations by potential outcome model. Rows correspond to the linear, nonlinear, and demand models. Columns vary item count $N\in\{5{,}000,10{,}000,20{,}000,50{,}000\}$, top-$K\in\{5,10,20\}$ graph sparsity, and treatment probability $p\in\{0.1,0.2,0.5\}$. The estimator is weighted H{\'a}jek; the design objective is unweighted and the estimand is traffic-weighted.}
    \label{fig:setting_ablation}
\end{figure}

The next group of experiments varies design inputs that are relevant to all methods. The item-count grid is $N\in\{5{,}000,10{,}000,20{,}000,50{,}000\}$. This checks whether the covariance-design advantage persists as the number of units grows. The top-$K$ grid is $K\in\{5,10,20\}$ and changes the sparsity of the competition graph. The treatment-share grid is $p\in\{0.1,0.2,0.5\}$ and tests the general-$p$ Gaussian-copula construction. Static-OCD is only shown where the implemented sign-Gaussian benchmark is available.

COSTA is the strongest design across the RetailRocket item-scale and graph-sparsity grids shown here. The top-$K$ ablation is especially informative: increasing $K$ changes both the number of exposure edges and the difficulty of matching local environments to global treatment. COSTA and COSTA-LocalPenalty respond by reallocating covariance, whereas Independent and RBSD are insensitive to graph structure. The treatment-share ablation should be read more carefully. At $p=0.1$ and $p=0.2$, the Bernoulli covariance region is asymmetric: strong negative covariance is infeasible, exact short-horizon balance is harder, and the advantage of covariance optimization can attenuate. The main value of the general-$p$ Gaussian copula is therefore not a claim of uniform dominance at low treatment shares, but a valid way to search the feasible covariance region and diagnose when the margin constraint is binding.

\FloatBarrier

\subsection{Bias--standard deviation tradeoff}
\label{subsec:bias_std_tradeoff}

Figure~\ref{fig:bias_std} plots absolute bias against standard deviation in the RetailRocket default setting. Each panel corresponds to one potential outcome model. The visual message is consistent across the three panels. Independent and RBSD sit at the low-variance/high-bias end: they provide stable assignments but leave the exposure environment far from the global treatment and control regimes. Cluster and Static-OCD reduce some bias but often pay a variance cost. COSTA is closest to the lower-left region because it directly targets the covariance directions that determine network-spillover and carryover bias. The implemented COSTA-LocalPenalty variant is observed to move leftward in standard deviation relative to COSTA while penalizing the sampled pairwise quantities in \eqref{eq:local_penalty_exact}; the accompanying movement in bias is an empirical finite-sample tradeoff, not by itself a formal weak-dependence guarantee.

\begin{figure}[!htbp]
    \centering
    \includegraphics[width=0.98\linewidth]{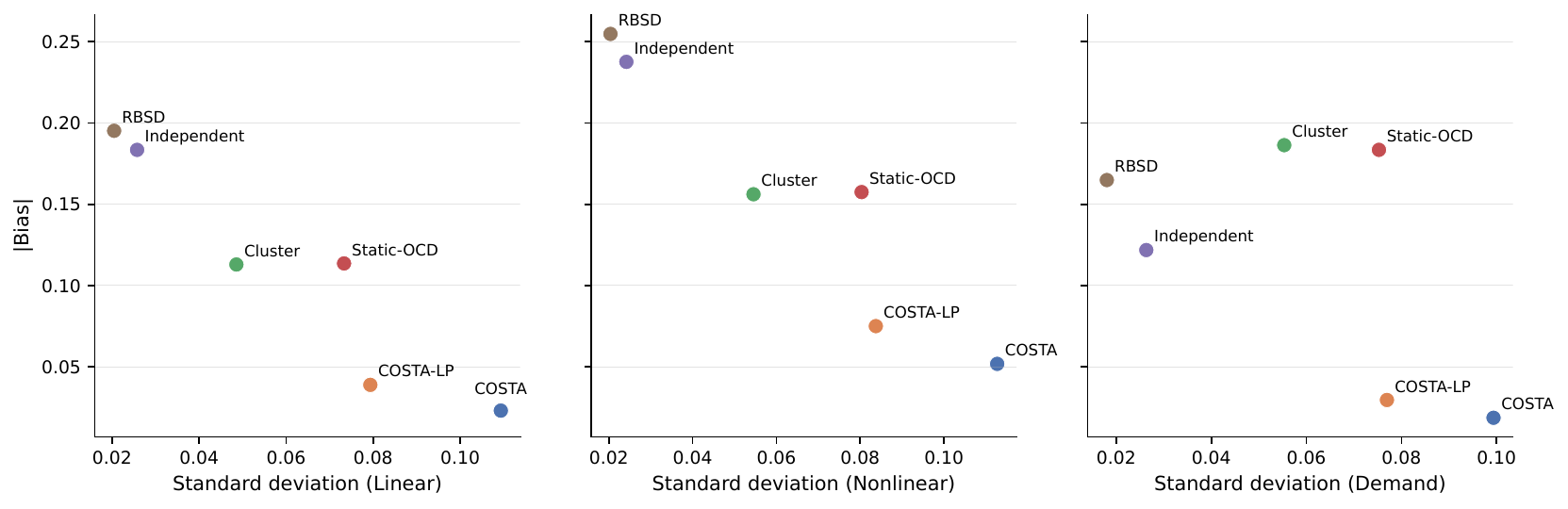}
    \caption{Bias--standard deviation tradeoff in the RetailRocket default setting. Each panel is a separate potential outcome model; the estimator is weighted HT. The setting is $N=10{,}000$, $B=8$, top-$K=10$, $p=0.5$, and an unweighted covariance-design objective with traffic-weighted estimands.}
    \label{fig:bias_std}
\end{figure}

\FloatBarrier

\subsection{COSTA-specific ablations}
\label{subsec:costa_ablations}

We next isolate design choices that are specific to COSTA. Default results use $B=8$; the larger $N$/top-$K$/$p$ setting grids above use $B=4$ for computational comparability. Figure~\ref{fig:kro_dim} varies the Kronecker network dimension $d_I\in\{16,32,64\}$ with $d_T=B-1$. The grid is plotted on logarithmic ticks because the dimensions double. The main finite-sample conclusion is that $d_I$ has limited impact on RMSE in the default RetailRocket setting: dimension $16$ already captures the dominant covariance directions used by the point-estimation objective. This supports low rank as an MSE regularizer, not as an inferential certificate. Proposition~\ref{prop:global_factor_certifiability}(i) shows that a fixed $d_I$ cannot justify uniform weak-dependence asymptotics on a sparse growing graph merely because its finite-sample RMSE has stabilized.

\begin{figure}[!htbp]
    \centering
    \includegraphics[width=0.98\linewidth]{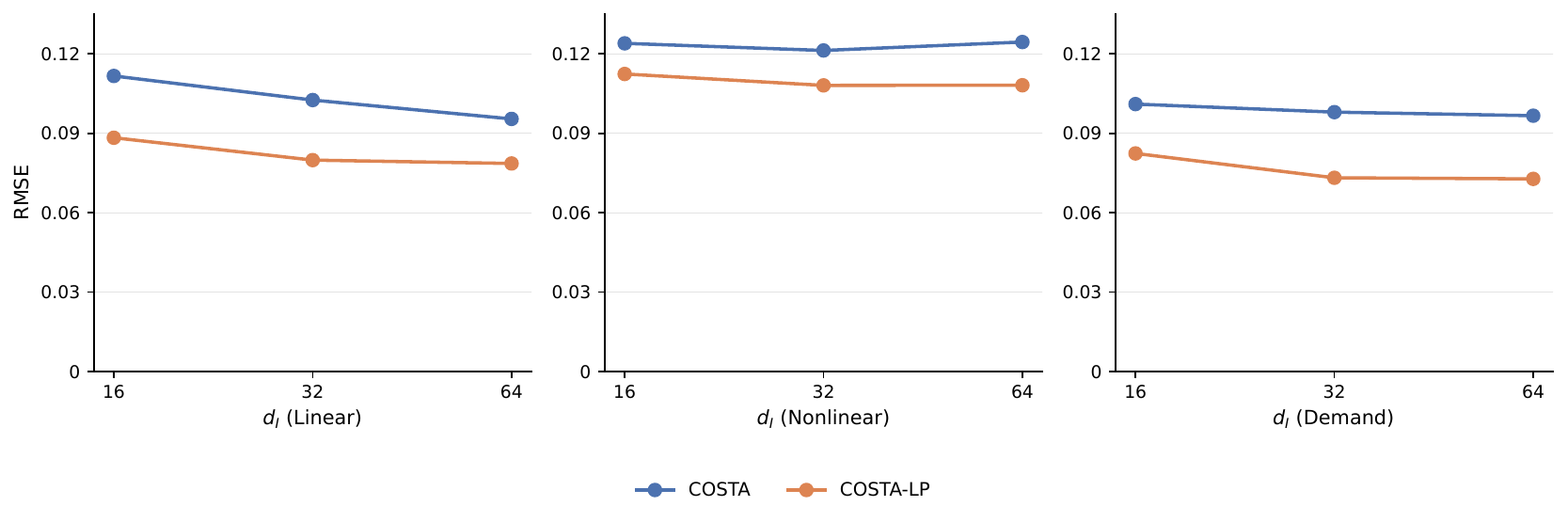}
    \caption{COSTA Kronecker network-dimension ablation on RetailRocket. Each panel is a separate potential outcome model; $d_I\in\{16,32,64\}$, $d_T=B-1$, $N=10{,}000$, $B=8$, top-$K=10$, and $p=0.5$. The estimator is weighted HT.}
    \label{fig:kro_dim}
\end{figure}

Figure~\ref{fig:lambda_ablation} varies the COSTA-LocalPenalty strength over the prespecified RetailRocket grid. The left axis reports weighted-HT RMSE and the right axis reports the in-sample fixed-batch diagnostic $\texttt{far\_product\_mean\_abs\_R}=m_P^{-1}\sum_\ell|R^P_\ell|$. The diagnostic is neither the full three-component penalty $J_{\rm loc}$ nor the setwise coefficient $\theta_N^{\rm GC}(s)$. Stronger shrinkage lowers diffuse product-cell covariance in all three outcome models, whereas RMSE is nonmonotone because the penalty can also restrict useful exposure-edge covariance. The default $\lambda_{\rm loc}=1$ was prespecified rather than selected from this sweep.

\begin{figure}[!htbp]
    \centering
    \includegraphics[width=0.98\linewidth]{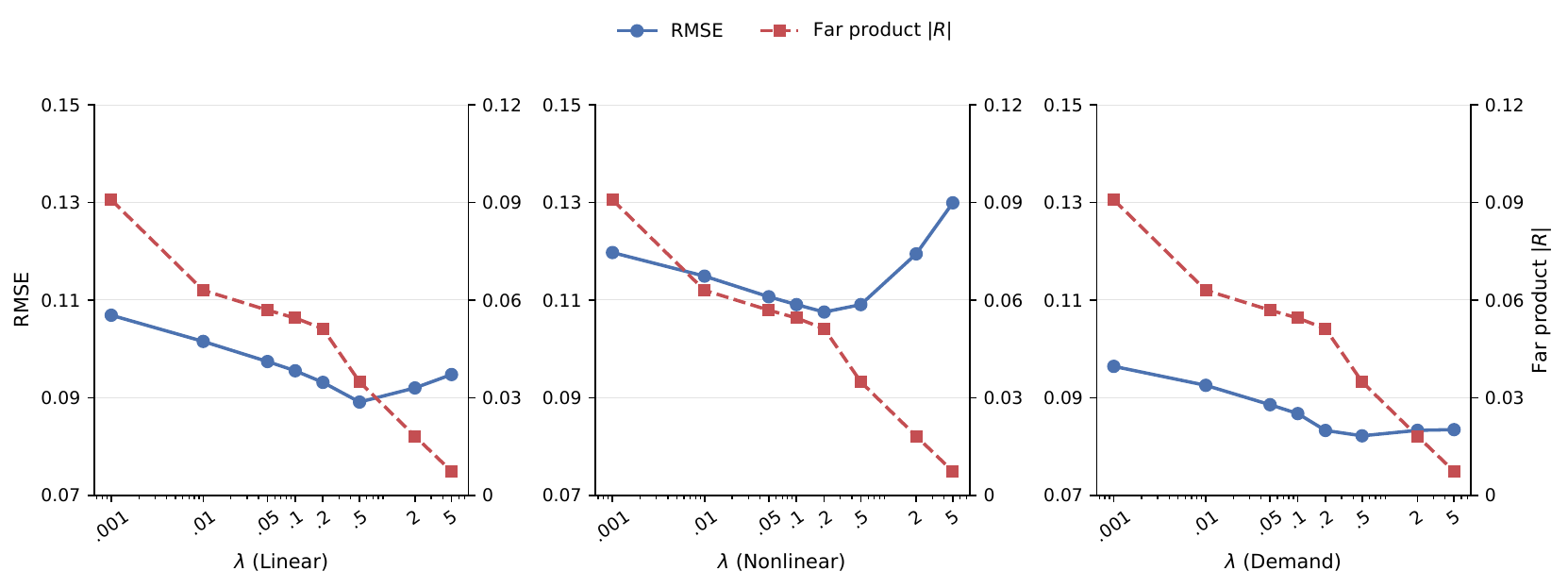}
    \caption{RetailRocket COSTA-LocalPenalty sweep using weighted HT. Circles report RMSE and dashed squares report the in-sample fixed-batch mean absolute product-cell covariance. Moderate shrinkage substantially reduces diffuse pairwise covariance with limited RMSE movement; the diagnostic is not the setwise coefficient $\theta_N^{\rm GC}$.}
    \label{fig:lambda_ablation}
\end{figure}

In the six default H{\'a}jek comparisons, five COSTA-LocalPenalty RMSE values differ from COSTA by at most $0.010$, with a larger $0.020$ cost in the RetailRocket nonlinear case. In the inference grid, COSTA-LocalPenalty reduces average $90\%$ interval width by approximately $25$--$38\%$ across all six dataset--outcome combinations. The narrower intervals reflect lower Monte Carlo variance after design-aware calibration, not uncertainty underestimation.

Table~\ref{tab:cluster_kro} compares Cluster, COSTA, and two cluster-level COSTA variants. The ablation separates two ways of reducing the design space: clustering units before randomization and using the Kronecker parameterization. Clustering alone can reduce bias but can substantially increase variance when cluster-level assignments are too coarse. The cluster-level COSTA variants show that preserving item-level covariance degrees of freedom matters; Kronecker structure is a more targeted reduction than collapsing units into graph clusters.

\begin{table}[!htbp]
    \centering
    \caption{Cluster and Kronecker design-space ablation on RetailRocket. Rows report separate potential outcome models in the default setting with weighted H{\'a}jek.}
    \label{tab:cluster_kro}
    \footnotesize
    \input{tables/tab_cluster_kro_retail.tex}
\end{table}

\subsection{Balance and switching penalty ablation}
\label{subsec:component_ablation}

The soft-balance and switching terms serve different purposes. Balance penalties discourage large fluctuations in treatment load by block and exposure by item; switching penalties prevent the optimizer from solving carryover bias by making adjacent blocks almost identical. Figure~\ref{fig:component_ablation} removes these terms from the soft-locality COSTA-LocalPenalty variant in the $B=8$, $K=10$, $p=0.5$ RetailRocket design grid. Removing balance raises RMSE in all three potential-outcome models shown, while removing the switching guardrail has a smaller but still adverse RMSE effect in this grid. These are descriptive point-estimation comparisons; inference is evaluated separately with the heterogeneous-mean and design-aware procedures below.

\begin{figure}[!htbp]
    \centering
    \includegraphics[width=0.98\linewidth]{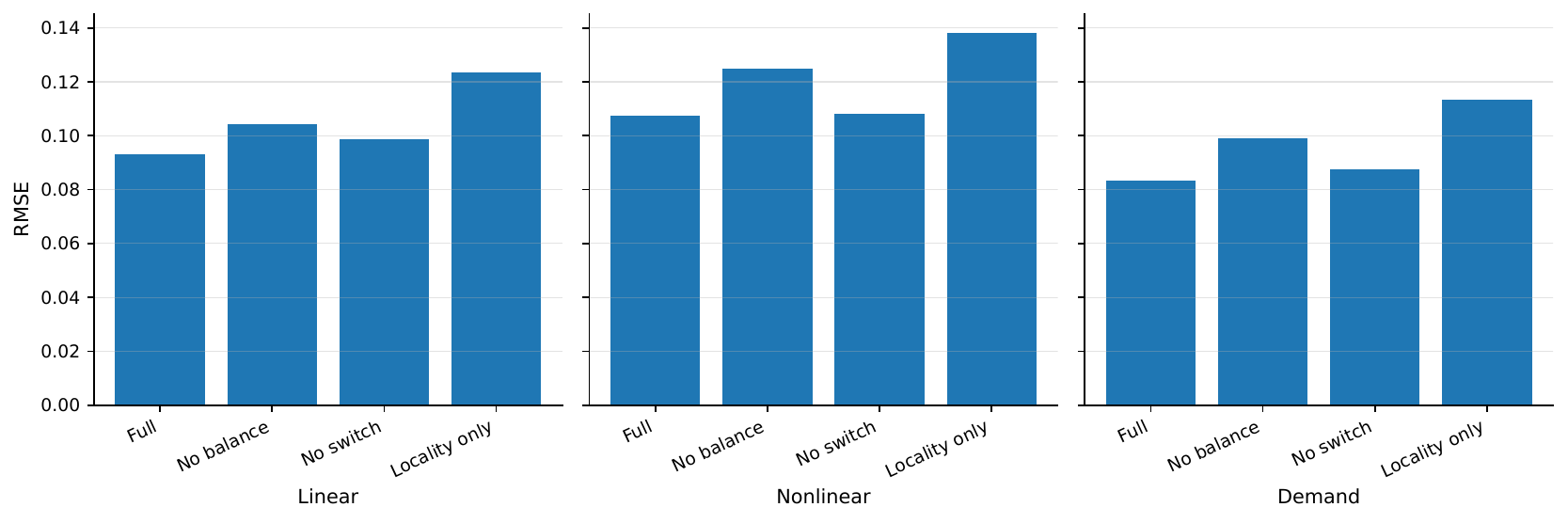}
    \caption{COSTA-LocalPenalty balance and switching penalty ablation on RetailRocket. Each panel is a separate potential outcome model; the setting is $N=10{,}000$, $B=8$, top-$K=10$, $p=0.5$, weighted HT, and the design grid previously used for inference diagnostics.}
    \label{fig:component_ablation}
\end{figure}

\IfFileExists{tables/tab_inference_validation.tex}{
\subsection{Inference validation}
\label{subsec:inference_experiments}
\label{subsec:inference_results}

The inference experiment evaluates the linear, nonlinear, and demand-substitution potential-outcome models. It compares local HAC centered by independently simulated cell-specific contribution means with signed and positive-part design-aware covariance-tail corrections. The mean-calibration, tail-calibration, and main Monte Carlo streams are independent. Both datasets use $N=10{,}000$, $B=8$, top-$K=10$, $p=0.5$, $d_I=16$, and $d_T=7$; the main experiment uses $1{,}000$ assignments and each calibration stage uses $500$ independent assignments. For the nonlinear and demand models, the known semi-synthetic conditional mean is used as the model-assisted calibration target. Those two cases assess finite-sample robustness of the variance workflow; they do not extend the exact linear expected-cut identity to nonlinear outcomes.

Table~\ref{tab:inference_validation} reports the positive-tail design-aware procedure. In addition to design-centered and GATE coverage, it reports the average interval width and its ratio to the empirical central $90\%$ Monte Carlo width, defined by the fifth and ninety-fifth percentiles of the estimator draws. With $1{,}000$ main draws, the Monte Carlo standard error of a $90\%$ coverage estimate is about $\{0.9(0.1)/1000\}^{1/2}=0.0095$. Figures~\ref{fig:inference_design_coverage} and~\ref{fig:inference_width} compare the three variance procedures across outcome models.

\begin{table}[!htbp]
\centering
\caption{Inference validation under linear, nonlinear, and demand-substitution outcomes. Coverage is evaluated at the nominal $90\%$ level. The table reports the positive-tail design-aware procedure; widths and standard-error ratios are Monte Carlo averages. DC denotes coverage at the independently calibrated design expectation, GATE denotes uncorrected raw-GATE coverage, RC denotes recentered GATE coverage, and Width/MC divides average interval width by the empirical fifth-to-ninety-fifth-percentile width. The same independent mean-calibration stream defines the design-center reference and recentering adjustment, so RC and DC are algebraically identical conditional on that stream.}
\label{tab:inference_validation}
\footnotesize
\input{tables/tab_inference_validation.tex}
\end{table}

\FloatBarrier

\IfFileExists{figures/fig_inference_design_coverage.pdf}{
\begin{figure}[!htbp]
\centering
\includegraphics[width=0.98\linewidth]{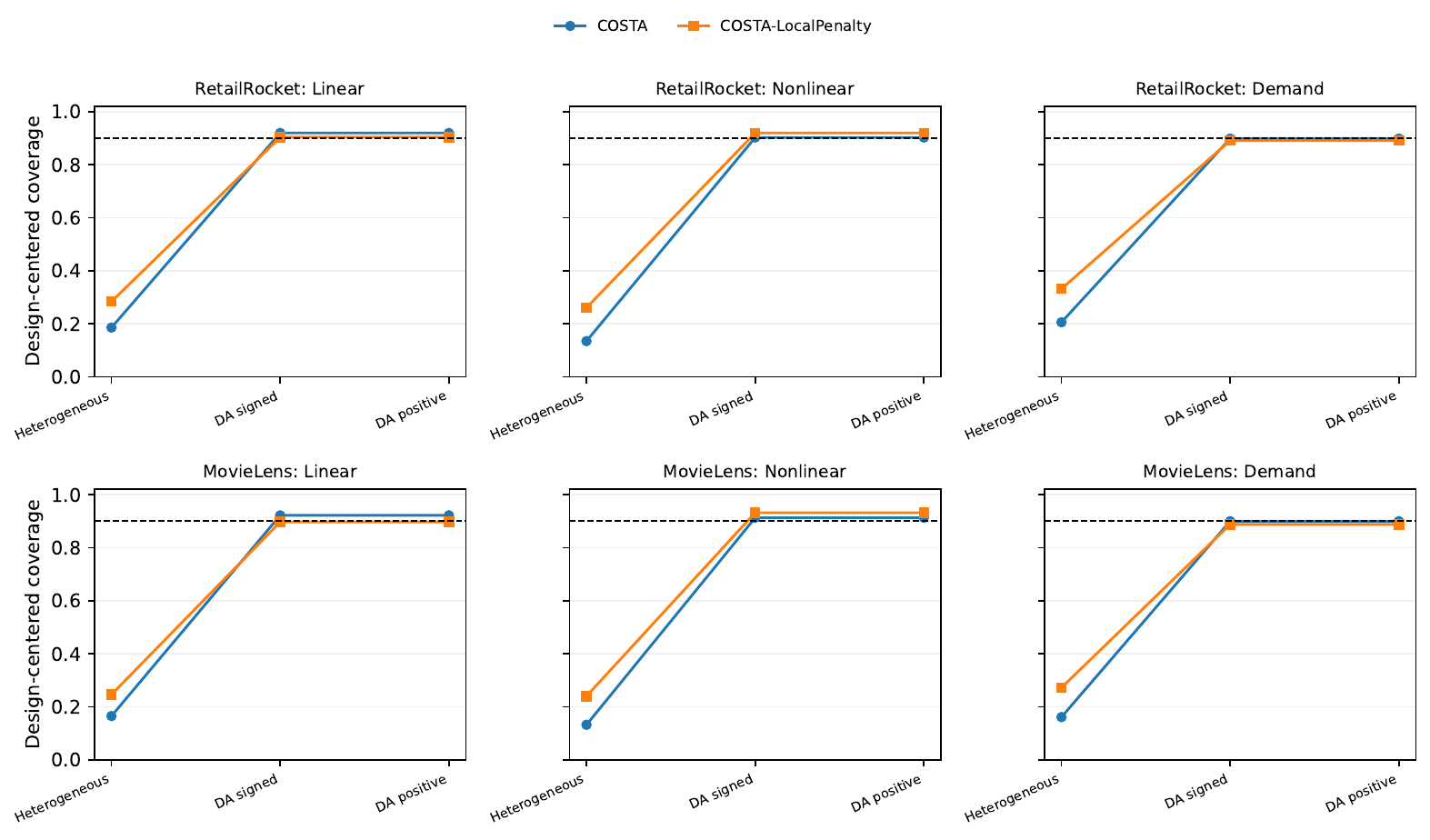}
\caption{Design-centered coverage at the nominal $90\%$ level. Rows are datasets and columns are potential-outcome models. Heterogeneous uses independently calibrated cell-specific means; DA signed and DA positive add the signed and positive-part nonlocal covariance tails. The dashed line marks the nominal level.}
\label{fig:inference_design_coverage}
\end{figure}

\begin{figure}[!htbp]
\centering
\includegraphics[width=0.98\linewidth]{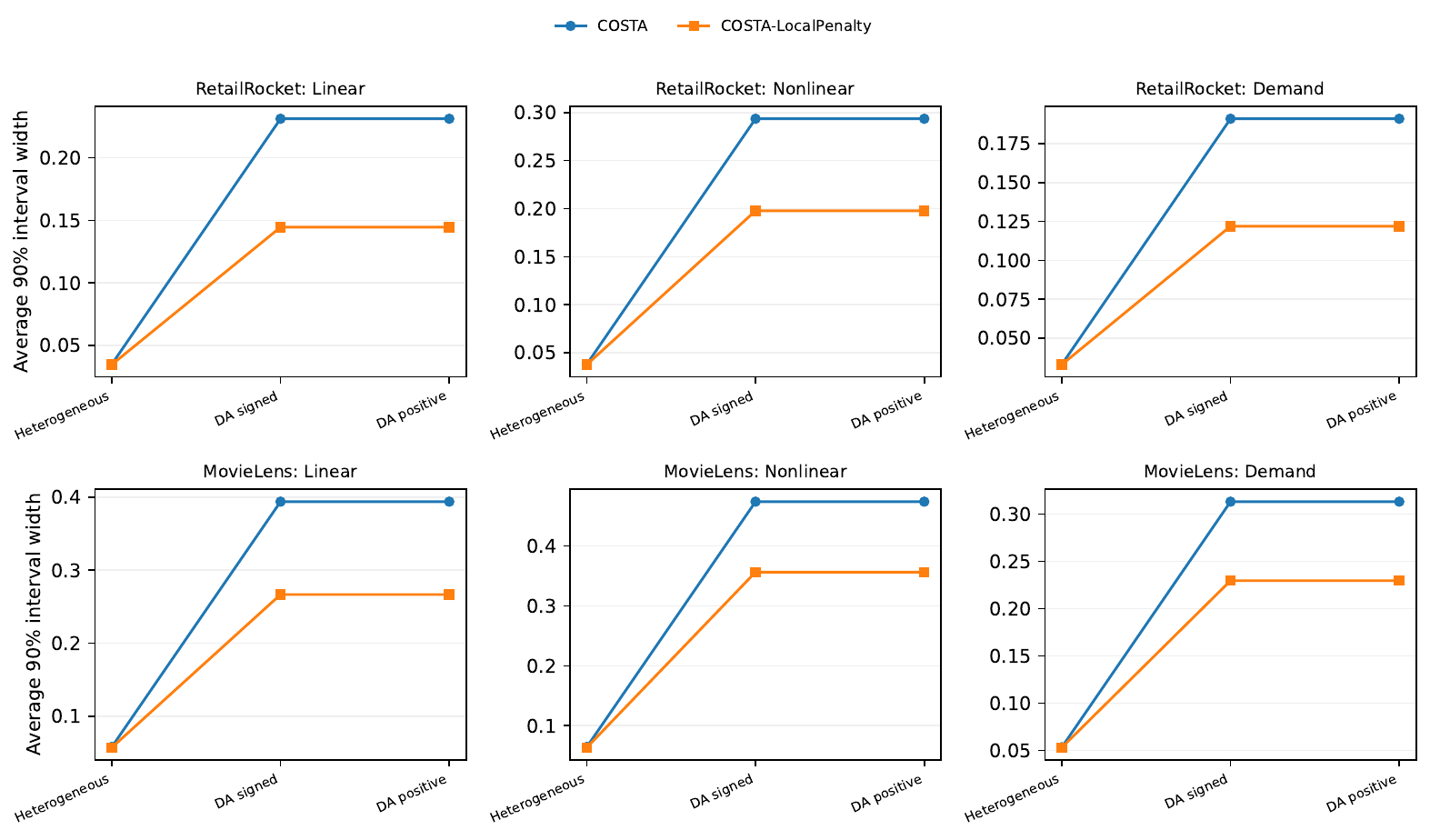}
\caption{Average $90\%$ confidence-interval width for the same variance procedures and outcome models as Figure~\ref{fig:inference_design_coverage}. Width must be read together with coverage: the narrow local intervals substantially understate uncertainty, whereas the design-aware widths closely track the empirical central Monte Carlo ranges reported in Table~\ref{tab:inference_validation}.}
\label{fig:inference_width}
\end{figure}
}{}

The design-aware corrections deliver stable finite-sample calibration across all three outcome models. Across COSTA and COSTA-LocalPenalty, design-centered coverage ranges from $0.888$ to $0.932$, the estimated-standard-error to Monte Carlo standard-deviation ratio ranges from $0.968$ to $1.078$, and the reported-width to empirical-width ratio ranges from $0.967$ to $1.100$. The corresponding studentized standard deviations in Table~\ref{tab:inference_studentized} range from $0.927$ to $1.033$. By contrast, the local heterogeneous-mean estimator substantially understates uncertainty: correcting contribution means alone does not recover covariance outside the local product-graph window. Signed and positive-part design-aware corrections coincide for the two reported designs because their calibrated covariance tails are positive, and no variance estimate is truncated at zero.

Interval width also distinguishes the two optimized designs. COSTA-LocalPenalty intervals are $25\%$--$38\%$ narrower than COSTA intervals across the six dataset--model comparisons, reflecting the smaller Monte Carlo variance of its estimator in this grid rather than an uncalibrated local-HAC shortcut. Raw-GATE coverage remains a separate question and ranges from $0.696$ to $0.909$. The same independent mean-calibration stream defines both the design-center reference and recentering adjustment, so recentered GATE coverage is algebraically identical to design-centered coverage conditional on that stream. This illustrates the IC-1 recentering step; same-order calibration uncertainty is treated theoretically under a joint pilot--experiment law rather than evaluated here. Thus the experiment illustrates rather than collapses the IC hierarchy: the design-aware variance correction addresses the omitted covariance tail, while causal centering still depends on design bias. The nonlinear and demand results show that the model-assisted workflow can be calibrated under richer semi-synthetic outcome surfaces; they are not, by themselves, asymptotic IC-0 certificates for every fixed-rank sampler.

The studentized diagnostics in Table~\ref{tab:inference_studentized} provide a distributional check beyond coverage: across the twelve dataset--model--design cells, centered means range from $-0.066$ to $0.083$, standard deviations from $0.927$ to $1.033$, and the fifth and ninety-fifth percentiles remain close to their Gaussian benchmarks. These diagnostics examine the centered distribution rather than coverage alone.

\FloatBarrier
}{}

\section{Conclusion}
\label{sec:conclusion}

COSTA makes the joint assignment law a first-class design object for experiments with network and temporal interference. Its finite-sample foundation is exact: under common marginals and positive linear exposure, HT bias for the sustained global contrast is the negative expected weight of the treatment boundary. Together with a covariance-level variance envelope, this representation gives a covariance-optimizable MSE bound and identifies the network and lag directions that matter for causal alignment. The Kronecker Gaussian-copula construction then turns that design principle into a scalable assignment law.

The inference theory is reusable beyond the proposed parameterization. Existing graph-$\psi$ CLT and network-HAC results provide the probabilistic engine; our contribution is to translate structured assignment covariance and local interference into dependence bounds for separated sets of HT contributions, and then to separate design-centered normality, causal centering, and feasible studentization. Sparse, block, Kronecker, locally factored, and other covariance structures can use the same framework when their separation, topology, moment, variance, and kernel conditions are verified.

The path result clarifies rather than weakens covariance optimization: a design may deliver large finite-sample MSE gains while raw sustained-treatment GATE inference still needs an additional centering argument. Model-assisted debiasing, alternative targets, bias-aware intervals, and nonlocal-factor limit theory address different versions of that problem. The experiments show substantial and robust point-estimation gains, while the heterogeneous-mean and design-aware variance analysis demonstrates how the known randomization law can improve interval calibration.

\clearpage
\bibliographystyle{plainnat}
\bibliography{COSTA_references_arxiv}
\clearpage
\appendix

\section{Supplementary Theory Statements}
\label{app:supplementary_theory}

The random-environment convention of Section~\ref{sec:inference} applies throughout this supplement, except for the explicit joint-law carve-out in Proposition~\ref{prop:external_pilot_uncertainty}. For compactness, estimator conclusions written with $\to_p$ are conditional-in-randomization statements in the sense of \eqref{eq:conditional_probability_convergence_notation}, whereas rates involving only $\mathcal F_N$-measurable graph, covariance, or pilot quantities are under the outer environment law. Within ordinary conditional inference statements and their proofs, unsubscripted $\E$, $\Cov$, and $\Var$ are shorthand for $\mathbb E_N$, $\Cov_N$, and $\Var_N$. Proposition~\ref{prop:external_pilot_uncertainty} instead conditions only on the pre-pilot field $\mathcal E_N$ and marks all joint operators by the superscript $J$.

This appendix records supporting extensions whose proofs appear in Appendix~\ref{app:inference_proofs} and Appendix~\ref{app:technical_derivations}. They are separated from the main theoretical narrative because they close technical interfaces or provide optional inferential routes rather than defining the paper's principal design-verification contribution.

\begin{table}[H]
\centering
\scriptsize
\caption{Interpretive map for supplementary results. The table states the role and boundary of every supplementary theorem-level statement without promoting secondary extensions to headline contributions.}
\label{tab:supplement_result_map}
\renewcommand{\arraystretch}{1.15}
\setlength{\tabcolsep}{3.5pt}
\begin{tabular}{@{}>{\raggedright\arraybackslash}p{0.25\linewidth}>{\raggedright\arraybackslash}p{0.31\linewidth}>{\raggedright\arraybackslash}p{0.36\linewidth}@{}}
\toprule
Result & Main meaning & Boundary / what it does not establish \\
\midrule
Cor.~\ref{cor:weighted_variance_envelope} & Weighted formal precision envelope & Requires a dominating scaled direction; not the code edge-load term \\
Cor.~\ref{cor:signed_exposure_operator} and Prop.~\ref{prop:robust_bias_objective} & Signed/heterogeneous bias algebra and robust calibration & One-sided cut geometry requires nonnegative exposure \\
Props.~\ref{prop:fractional_multicut_bias}--\ref{prop:fractional_multicut_dual}, Cors.~\ref{cor:grid_bias_floor}--\ref{cor:expander_bias_floor} & General-graph disagreement lower bounds & Geometry-specific; no universal path-rate claim \\
Lemma~\ref{lem:conditional_unconditional_transfer}, Cors.~\ref{cor:poly_geometric_clt}--\ref{cor:local_factor_ic0}, Prop.~\ref{prop:local_approximation_transfer} & Conditional transfer and primitive IC-0 routes & Centered normality only; no automatic GATE centering \\
Cor.~\ref{cor:gaussian_independent_noise} & Independent own-cell noise preserves the assignment certificate & Shared noise requires an enlarged certificate \\
Prop.~\ref{prop:nugget_bias_floor}, Cors.~\ref{cor:nugget_supported_certificate_frontier}--\ref{cor:fixed_range_bias_floor} & Spectral-locality and bias frontiers & Route-specific certificate boundaries, not universal impossibility \\
Prop.~\ref{prop:external_debiasing}, Cor.~\ref{cor:external_debiasing_pilot_size}, Prop.~\ref{prop:external_pilot_uncertainty} & Negligible- and same-order-pilot debiasing routes & Design-specific pilot; joint law required for same-order uncertainty \\
Prop.~\ref{prop:misspecification_bias} & Exact working-model plus remainder decomposition & Remainder envelope may be conservative \\
Prop.~\ref{prop:oracle_network_hac_disconnected}, Cor.~\ref{cor:component_independent_hac} & Disconnected-graph oracle HAC routes & Requires explicit cross-component control or independence \\
Cors.~\ref{cor:primitive_centering_rate}--\ref{cor:indefinite_kernel_centering}, Prop.~\ref{prop:uncentered_conservative} & Feasible-centering and conservative-variance routes & Do not establish oracle-HAC validity \\
Props.~\ref{prop:bias_aware_coverage}--\ref{prop:design_aware_hac} & Bias-aware and tail-corrected intervals & Coverage depends on valid envelopes/tail estimates \\
Prop.~\ref{prop:hajek_delta} & High-level H{\'a}jek delta-method route & Assumes projection-wise weak-dependence conditions \\
\bottomrule
\end{tabular}
\end{table}

\subsection{Heterogeneous exposures and robust bias calibration}
\label{app:heterogeneous_robust_design_extensions}

\begin{corollary}[Weighted covariance-level variance envelope]
\label{cor:weighted_variance_envelope}
Under the weighted identity in Corollary~\ref{cor:weighted_ht_identity}, let $s_w\in\mathbb R_+^L$ satisfy $s_w\ge H^\top w$ componentwise and
\begin{equation}
 |g_{p,w,c}^{\tilde\alpha}|\le \omega_{g,w}s_{w,c},
 \qquad c=1,\ldots,L.
 \label{eq:weighted_envelope_comparability}
\end{equation}
Then the assignment-measurable component satisfies
\begin{equation}
 \Var\{(\widehat\tau^{HT}_{w,\tilde\alpha})^{(0)}\}
 \le \frac{2}{W^2}
 \left(\omega_{g,w}^2+\frac{\theta^2}{p^2q^2}\right)
 s_w^\top(v\Rnorm+p^2\one\one^\top)s_w.
 \label{eq:weighted_variance_envelope}
\end{equation}
Under the independent random-measurement-error regime of Corollary~\ref{cor:weighted_ht_identity}, add $W^{-2}v^{-1}\sum_aw_a^2\sigma_{\varepsilon,a}^2$ to the right-hand side.
\end{corollary}

\paragraph{Meaning.}
The formal target-weighted objective has the same envelope structure as the unweighted theorem, but the dominating degree vector is $H^\top w$. A normalized direction such as $e_w$ is not itself the envelope unless a positive rescaling satisfies the dominance and comparability conditions.

\begin{proof}
Apply the proof of Theorem~\ref{thm:variance_envelope} to the weighted polynomial in \eqref{eq:weighted_exact_algebra}, replacing $M$, $H$, $g_p^{\tilde\alpha}$, and $s$ by $W$, $D_wH$, $g_{p,w}^{\tilde\alpha}$, and $s_w$. Since $(D_wH)^\top\one=H^\top w$, the same coordinatewise domination argument gives \eqref{eq:weighted_variance_envelope}; \eqref{eq:weighted_noise_variance} supplies the independent-noise term.
\end{proof}

\begin{corollary}[Heterogeneous signed exposure operator]
\label{cor:signed_exposure_operator}
Replace the scalar exposure component in Theorem~\ref{thm:main} by
\begin{equation}
    Y_a^{\star}(z)=\mu_a+\beta_az_a+\sum_c\Gamma_{ac}z_c,
    \label{eq:signed_exposure_model}
\end{equation}
where $\Gamma\in\mathbb R^{L\times L}$ may be heterogeneous and signed and has zero rows outside $\Ocal$. Under the same noise condition,
\begin{equation}
    \E\widehat\tau_{w,\tilde\alpha}^{HT}-\tau_w
    =\frac1W\langle D_w\Gamma,
      \Rnorm-\one\one^\top\rangle.
    \label{eq:signed_exposure_bias}
\end{equation}
For the assignment-measurable component, define
\begin{equation}
    g_{p,w,\Gamma}^{\tilde\alpha}
    =\frac{w\odot r^{\tilde\alpha}}v
      +\frac{w\odot\beta}p
      -\frac{\Gamma^\top w}q.
    \label{eq:signed_exposure_g}
\end{equation}
Then, up to a nonrandom constant,
\begin{equation}
    W(\widehat\tau_{w,\tilde\alpha}^{HT})^{(0)}
    =(g_{p,w,\Gamma}^{\tilde\alpha})^\top Z
      +\frac1vZ^\top D_w\Gamma Z.
    \label{eq:signed_exposure_algebra}
\end{equation}
The positive-exposure cut representation is recovered when
$\Gamma=\theta H$ with $\theta>0$ and $H\ge0$; for signed $\Gamma$, the exact identity remains valid but cancellation and sign information must be handled through an explicit uncertainty set rather than a nonnegative cut cost.
\end{corollary}

For the two-channel model, write
\[
    \bar b(\Rnorm)
    =\bigl(\bar b_{\comp}^{(w)}(\Rnorm),
            \bar b_{\lag}^{(w)}(\Rnorm)\bigr)^\top,
    \qquad
    \delta=
    \left(\frac{\gamma E_{\comp}^{(w)}}W,
          \frac{\eta E_{\lag}^{(w)}}W\right)^\top,
\]
so that the working-model bias is $\delta^\top\bar b(\Rnorm)$.

\begin{proposition}[Minimax interpretation of separated bias penalties]
\label{prop:robust_bias_objective}
For a positive-definite matrix $\Lambda$, let
\[
    \mathcal U_\Lambda
    =\{\delta:\delta^\top\Lambda^{-1}\delta\le1\}.
\]
Then
\begin{equation}
    \sup_{\delta\in\mathcal U_\Lambda}
    \{\delta^\top\bar b(\Rnorm)\}^2
    =\bar b(\Rnorm)^\top\Lambda\bar b(\Rnorm).
    \label{eq:ellipsoidal_robust_bias}
\end{equation}
In particular, diagonal $\Lambda$ yields exactly a weighted sum of the two squared bias channels, with no cancellation. If instead
$|\delta_1|\le\bar\delta_1$ and $|\delta_2|\le\bar\delta_2$, then
\begin{equation}
    \sup_{\delta}
    \{\delta^\top\bar b(\Rnorm)\}^2
    =\left\{\bar\delta_1|\bar b_{\comp}^{(w)}(\Rnorm)|
      +\bar\delta_2|\bar b_{\lag}^{(w)}(\Rnorm)|\right\}^2.
    \label{eq:box_robust_bias}
\end{equation}
\end{proposition}

\subsection{General-graph expected-cut extensions}
\label{app:general_graph_cut_extensions}

Proposition~\ref{prop:fractional_multicut_bias} in the main text gives the general-graph lower bound. This appendix records its exact path-packing dual and two useful graph-specific consequences.

The lower bound has an exact path-packing dual. For each terminal pair $(u,v)\in\mathcal P$, let $\mathcal Q_{uv}$ denote the finite collection of simple $u$--$v$ paths in $\mathcal G_H$.

\begin{proposition}[Fractional path-packing dual]
\label{prop:fractional_multicut_dual}
The multicut value in \eqref{eq:fractional_multicut_program} equals
\begin{equation}
\begin{split}
\operatorname{MC}(\mathcal G_H,h,\mathcal P,\delta)
=\max_{f_{uv,P}\ge0}\quad
&\sum_{(u,v)\in\mathcal P}\sum_{P\in\mathcal Q_{uv}}
\delta_{uv} f_{uv,P}\\
\text{subject to}\quad
&\sum_{(u,v)\in\mathcal P}
  \sum_{P\in\mathcal Q_{uv}:e\in P}f_{uv,P}\le h_e,
  \qquad e\in\mathcal E_H.
\end{split}
\label{eq:fractional_multicut_dual}
\end{equation}
In particular, if one can select an edge-disjoint path $P_{uv}$ for every terminal pair and $h_e\ge h_0$ on their union, then
\begin{equation}
\operatorname{MC}(\mathcal G_H,h,\mathcal P,\delta)
\ge h_0\sum_{(u,v)\in\mathcal P}\delta_{uv}.
\label{eq:edge_disjoint_path_certificate}
\end{equation}
Thus path packing is a directly verifiable certificate for a causal-bias floor.
\end{proposition}

\begin{corollary}[Opposite-face lower bound on a grid]
\label{cor:grid_bias_floor}
Let $\mathcal G_H$ be the nearest-neighbor graph on $[n]^d$, with $h_e\ge h_0>0$. For each $j=(j_2,\ldots,j_d)\in[n]^{d-1}$, put $u_j=(1,j)$ and $v_j=(n,j)$ and suppose $d_Z(u_j,v_j)\ge\delta$. Then
\begin{equation}
\left|\E\widehat\tau^{HT}_{w,\tilde\alpha}-\tau_w\right|
\ge \frac{\theta h_0\delta}{W}n^{d-1}.
\label{eq:grid_bias_floor}
\end{equation}
If $\Ocal=\Zcal=[n]^d$, so $L=M=n^d$, and $W\asymp L$, the lower bound is of order $n^{-1}=L^{-1/d}$.
\end{corollary}

\begin{corollary}[Expansion bias floor]
\label{cor:expander_bias_floor}
Assume the common-marginal Bernoulli assignment law and exact treatment balance $\sum_{a\in\Zcal}Z_a=pL$ almost surely, with $pL$ an integer, and define the weighted edge expansion
\begin{equation}
\phi_h:=\min_{\varnothing\ne S\subsetneq\Zcal}
\frac{\sum_{e\in\partial_H S}h_e}{\min\{|S|,L-|S|\}}.
\label{eq:weighted_edge_expansion}
\end{equation}
Under positive exposure,
\begin{equation}
\left|\E\widehat\tau^{HT}_{w,\tilde\alpha}-\tau_w\right|
\ge
\frac{\theta\phi_hL\min\{p,1-p\}}{2p(1-p)W}.
\label{eq:expander_bias_floor}
\end{equation}
Consequently, under the stated common-marginal law, if $W\asymp L$ and $\liminf_N\phi_{h,N}>0$, exact balance entails a nonvanishing raw-GATE bias floor.
\end{corollary}


\subsection{Conditional transfer and primitive CLT extensions}
\label{app:conditional_and_primitive_extensions}

\begin{lemma}[Conditional-to-unconditional Gaussian transfer]
\label{lem:conditional_unconditional_transfer}
Let $S_N$ be real-valued and define the conditional Kolmogorov error
\begin{equation}
\Delta_N(\mathcal F_N)
:=\sup_{t\in\mathbb R}
\left|\Pr(S_N\le t\mid\mathcal F_N)-\Phi(t)\right|.
\label{eq:conditional_kolmogorov_error}
\end{equation}
If $\Delta_N(\mathcal F_N)\to_p0$, then
\begin{equation}
\sup_{t\in\mathbb R}|\Pr(S_N\le t)-\Phi(t)|\longrightarrow0.
\label{eq:conditional_unconditional_transfer}
\end{equation}
\end{lemma}

Thus conditional limit theory transfers unconditionally without an additional uniform-integrability assumption on the Kolmogorov error: it is automatically bounded by one. Primitive moment and covariance conditions may still need to hold on events whose probability tends to one in order to establish the conditional premise.

\begin{corollary}[Polynomial graph growth and geometric dependence]
\label{cor:poly_geometric_clt}
Suppose $\sigma_N^2\asymp M_N$, Assumption~\ref{ass:weights_moments} holds, and for constants $C,d<\infty$ and $\rho\in(0,1)$,
\begin{equation}
    \sup_{N,a}|\mathcal B_N(a;s)|\le C(1+s)^d,
    \qquad
    \theta_N(s)\le C\rho^s
    \quad(s\ge0).
    \label{eq:poly_growth_geometric_decay}
\end{equation}
Then Assumption~\ref{ass:weak_dependence_rates} holds with $m_N=\lceil C_m\log M_N\rceil$ for all sufficiently large constants $C_m$.
\end{corollary}

\begin{corollary}[Finite-range dependency-graph specialization]
\label{cor:finite_range_clt}
Suppose Assumption~\ref{ass:local_outcomes} holds and, conditional on $\mathcal F_N$, assignment vectors indexed by two sets are independent whenever their graph distance exceeds a finite $r_Z$. Then $\{\xi_{Na}\}$ has a conditional dependency graph connecting outcome cells within
$h_*=2r_Y+r_Z$. Let
\begin{equation}
    \Delta_N(h_*):=\max_{a\in\Ocal_N}|\mathcal B_N(a;h_*)|.
    \label{eq:Delta_h}
\end{equation}
If $\sigma_N^2\asymp M_N$, the fourth moments are uniformly bounded, and
$\Delta_N(h_*)/M_N^{1/4}\xrightarrow{p}0$, then \eqref{eq:clt_centered} holds. The same conclusion holds when each contribution additionally depends on a cell-specific measurement error $\varepsilon_{Na}$, provided these errors are conditionally independent across cells and independent of the assignment array; one then includes $\varepsilon_{Na}$ in the local sigma-field attached to vertex $a$.
\end{corollary}

The main-text Corollary~\ref{cor:local_factor_ic0} applies the preceding finite-range result to Gaussian local-factor loadings.

Exact neighborhood interference is sufficient but not necessary. The following transfer result permits approximate neighborhood interference in the spirit of approximate-neighborhood asymptotics \citep{leung2022ani} and near-epoch dependence \citep{jenishprucha2012}.

\begin{proposition}[Locally approximable contributions]
\label{prop:local_approximation_transfer}
For each $N$, let $\ell_N$ be a positive integer and let $\{\xi_{Na}^{(\ell_N)}:a\in\Ocal_N\}$ be conditionally centered local approximants, $\mathbb E_N\xi_{Na}^{(\ell_N)}=0$, whose sum $T_N^{(\ell_N)}$ satisfies the centered conclusion of Theorem~\ref{thm:cell_clt}, with conditional variance $\sigma_{N,\ell}^2:=\Var_N\{T_N^{(\ell_N)}\}$. If
\begin{equation}
    \frac{\left\|\sum_{a\in\Ocal_N}\{\xi_{Na}-\xi_{Na}^{(\ell_N)}\}\right\|_{2,N}}
    {\sigma_N}\xrightarrow{p}0,
    \label{eq:aggregate_local_approximation}
\end{equation}
then
\[
    \frac{\sigma_{N,\ell}}{\sigma_N}\xrightarrow{p}1,
    \qquad
    \frac{T_N-T_N^{(\ell_N)}}{\sigma_N}
    \xrightarrow{p\mid\mathcal F}0,
\]
and \eqref{eq:clt_centered} also holds for the original estimator. One possible construction is to replace $Y_{Na}^{\star}(Z_N)$ by a version measurable with respect to $\mathcal F_N$ and assignments in $\mathcal B_N^{\Zcal}(a;\ell_N)$; condition \eqref{eq:aggregate_local_approximation}, rather than a pointwise locality assertion alone, is the scale-relevant requirement.
\end{proposition}

\begin{corollary}[Independent cellwise measurement noise]
\label{cor:gaussian_independent_noise}
Suppose the assignment-measurable structural part of each contribution satisfies Assumption~\ref{ass:local_outcomes}, but $Q_{Na}$ may additionally depend measurably on its own cell-specific error $\varepsilon_{Na}$ and on no other outcome-error coordinate. Conditional on $\mathcal F_N$, assume that $\{\varepsilon_{Na}:a\in\Ocal_N\}$ are independent across cells and independent of $G_N$. Then Proposition~\ref{prop:gaussian_canonical_certificate} remains valid with the same coefficient $\theta_N^{\rm GC}$. Thus independent outcome noise can be integrated out without weakening the assignment certificate, subject to the moment condition in Assumption~\ref{ass:weights_moments}.
\end{corollary}

\begin{remark}[Why assignment certificates do not absorb arbitrary outcome noise]
\label{rem:noise_certificate_gap}
The proposition is conditional on $\mathcal F_N$ and uses that each contribution is a measurable transformation of the relevant Gaussian subvector. Without that convention, assignment locality alone is insufficient. For example, take $\Omega_N=I$ and let every observed outcome contain the same exogenous mean-zero shock $U$. Then $Q_{Na}=\omega_{Na}\psi_{Na}(Z_{Na})U$ may have zero pairwise linear covariance at distant cells while $Q_{Na}^2$ and $Q_{Nc}^2$ share the same $U^2$ and remain nonlinearly dependent. The assignment canonical correlation is zero, but the contribution coefficient is not. The unchanged Gaussian coefficient is justified for independent own-cell noise by Corollary~\ref{cor:gaussian_independent_noise}. Shared or locally shared random noise must instead enter an enlarged joint assignment--noise coefficient; assignment locality alone is not enough.
\end{remark}

\subsection{Additional design frontiers and causal corrections}
\label{app:causal_corrections_extensions}

A spectral nugget can stabilize the latent covariance but is not a free inferential regularizer. Let
\begin{equation}
\Omega_N=(1-\lambda_N)\Omega_N^{\rm str}+\lambda_N I_{L_N},
\qquad \lambda_N\in(0,1).
\label{eq:gaussian_nugget_regularization}
\end{equation}
The following proposition quantifies its raw-GATE bias cost.

\begin{proposition}[Spectral nugget--causal bias frontier]
\label{prop:nugget_bias_floor}
Suppose the structural exposure coefficient satisfies $\vartheta_N>0$, $H_N\ge0$, and the Gaussian copula uses \eqref{eq:gaussian_nugget_regularization}. For every off-diagonal exposure pair,
\begin{equation}
    \Rnorm_{N,ac}\le
    \mathcal G_{p_N}(1-\lambda_N),
    \qquad a\ne c.
    \label{eq:nugget_edge_correlation_cap}
\end{equation}
Therefore, for fixed target weights with $W_N>0$,
\begin{equation}
    |\mathbb E_N\widehat\tau_{w,\tilde\alpha}^{HT}-\tau_{w,N}|
    \ge
    \frac{\vartheta_N}{W_N}
    \sum_{a\ne c}w_{Na}H_{N,ac}
    \{1-\mathcal G_{p_N}(1-\lambda_N)\}.
    \label{eq:nugget_bias_floor}
\end{equation}
If $p_N$ remains in an overlap set and $\lambda_N\downarrow0$, then
\begin{equation}
    1-\mathcal G_{p_N}(1-\lambda_N)
    =\frac{\varphi\{\Phi^{-1}(p_N)\}}
      {p_N(1-p_N)\sqrt\pi}\sqrt{\lambda_N}
      \{1+o(1)\},
    \label{eq:nugget_sqrt_expansion}
\end{equation}
uniformly over compact overlap sets. Hence, when weighted off-diagonal exposure mass is of order $W_N$, $\vartheta_N$ is bounded away from zero, and $\sqrt{V_N}\asymp M_N^{-1/2}$ with $W_N\asymp M_N$, raw-GATE causal centering requires $\lambda_N=o_p(M_N^{-1})$ under the outer environment law.
\end{proposition}

\begin{corollary}[Nugget-supported certification frontier]
\label{cor:nugget_supported_certificate_frontier}
In the root-$M_N$ setting of Proposition~\ref{prop:nugget_bias_floor}, suppose the spectral floor used in Proposition~\ref{prop:spectral_rowsum_certificate} is supplied asymptotically by the nugget,
$\kappa_N\asymp\lambda_N$, and for a relevant separation sequence $s_N>2r_Y$ the spectral route requires
\begin{equation}
\frac{\alpha_N(s_N-2r_Y)}{\kappa_N}\xrightarrow{p}0.
\label{eq:nugget_supported_rowsum_rate}
\end{equation}
If raw-GATE IC-1 also holds, then necessarily
\begin{equation}
\alpha_N(s_N-2r_Y)=o_p(\lambda_N)=o_p(M_N^{-1}).
\label{eq:nugget_supported_joint_frontier}
\end{equation}
This implication is specific to nugget-supported invertibility. It does not apply when the structured covariance itself has a nonvanishing spectral floor.
\end{corollary}

A nugget sequence can therefore deliver either a fixed spectral floor or negligible raw-HT GATE bias, but generally not both under nonvanishing positive exposure mass. This is a design frontier, not a numerical pathology.

\begin{corollary}[Variance-scale window]
\label{cor:variance_scale_window}
In the path setting of Theorem~\ref{thm:path_trilemma}, replace $\sigma_N^2\asymp M_N$ by an arbitrary positive variance scale and suppose $m_N=o(M_N)$, the graph-$\psi$ CLT rates hold, and \eqref{eq:far_assignment_decorrelation} holds. Then the $k=1$, $s=0$ term in \eqref{eq:weak_dep_clt_rate_1} requires
\begin{equation}
    m_N=o\!\left(\frac{\sigma_N^3}{M_N}\right).
    \label{eq:clt_radius_upper_window}
\end{equation}
If the positive-exposure raw bias is negligible relative to $\sqrt{V_N}=\sigma_N/M_N$, then necessarily
\begin{equation}
    m_N\gg\frac{M_N}{\sigma_N}.
    \label{eq:bias_radius_lower_window}
\end{equation}
A nonempty radius window therefore requires $\sigma_N^2/M_N\to\infty$. Conventional root-$M_N$ inference, for which $\sigma_N^2\asymp M_N$, has no such window.
\end{corollary}

\begin{corollary}[Fixed-range local factors can have a constant GATE bias]
\label{cor:fixed_range_bias_floor}
Under the positive path exposure model of Theorem~\ref{thm:path_trilemma}, suppose assignments are independent whenever path distance exceeds a fixed $r_F$. Then, with $m=r_F+1$,
\begin{equation}
    \liminf_{N\to\infty}
    |\mathbb E_N\widehat\tau_N^{HT}-\tau_N|
    \ge \frac{\vartheta_0h_0}{r_F+1}>0.
    \label{eq:fixed_range_path_bias_floor}
\end{equation}
Thus a fixed-range local-factor sampler may give a clean dependency-graph CLT around the design expectation while remaining separated from the nonvanishing-interference GATE by a constant.
\end{corollary}

Proposition~\ref{prop:external_debiasing} gives the negligible-pilot debiasing result in the main text. The next corollary translates its projected-error condition into a primitive pilot-size requirement.

\begin{corollary}[Pilot-size requirement for external debiasing]
\label{cor:external_debiasing_pilot_size}
Suppose an external pilot of effective size $n_{\delta,N}$ estimates a $d_{\delta,N}$-dimensional nuisance vector and
\begin{equation}
    \|\widehat\delta_N-\delta_N\|_2
    =O_p\!\left(\sqrt{\frac{d_{\delta,N}}{n_{\delta,N}}}\right).
    \label{eq:external_debiasing_parameter_rate}
\end{equation}
A sufficient condition for \eqref{eq:debiased_nuisance_rate} is
\begin{equation}
    \frac{\|c_N\|_2}{\sqrt{V_N}}
    \sqrt{\frac{d_{\delta,N}}{n_{\delta,N}}}\longrightarrow0,
    \quad\text{equivalently}\quad
    \frac{n_{\delta,N}V_N}
         {d_{\delta,N}\|c_N\|_2^2}\longrightarrow\infty
\label{eq:external_debiasing_pilot_size}
\end{equation}
whenever $\|c_N\|_2>0$. If $V_N\asymp M_N^{-1}$, $d_{\delta,N}=O(1)$, and $\|c_N\|_2\asymp1$, this requires $n_{\delta,N}/M_N\to\infty$. A pilot of the same order as the main experiment is therefore generally insufficient for treating plug-in uncertainty as negligible; the requirement relaxes if the design-bias loading shrinks.
\end{corollary}

\paragraph{Joint-law carve-out for pilot uncertainty.}
Unlike the preceding conditional-randomization statements, Proposition~\ref{prop:external_pilot_uncertainty} is formulated under a joint law that leaves both the external pilot and the current randomization random. Let $\mathcal E_N$ be the pre-pilot environment sigma-field, let $\mathcal F_N=\sigma(\mathcal E_N,\text{pilot data},\text{pilot-derived design})$, write $\mathbb P_N^D(\cdot)=\Pr(\cdot\mid\mathcal F_N)$ for the current-randomization law, and write $\mathbb P_N^J(\cdot)=\Pr(\cdot\mid\mathcal E_N)$ for the joint pilot--experiment law. All expectations, variances, probability limits, and weak limits in the proposition carry the superscript $J$ unless explicitly marked $D$.

\begin{proposition}[External-pilot uncertainty propagation under the joint law]
\label{prop:external_pilot_uncertainty}
Under the exact two-channel model, put
\begin{equation}
    U_N:=\widehat\tau^{HT}_{w,\tilde\alpha}
          -\mathbb E_N^D(\widehat\tau^{HT}_{w,\tilde\alpha}\mid\mathcal F_N),
    \qquad
    D_N:=c_N^\top(\widehat\delta_N-\delta_N),
    \label{eq:external_pilot_error_components}
\end{equation}
\begin{equation}
    V_N^D:=\Var_N^D(U_N\mid\mathcal F_N),
    \qquad
    V_{\delta,N}^J:=\Var_N^J(D_N\mid\mathcal E_N).
    \label{eq:external_pilot_joint_variances}
\end{equation}
Assume $V_N^D,V_{\delta,N}^J>0$ eventually and
\begin{equation}
    \begin{pmatrix}U_N/\sqrt{V_N^D}\\D_N/\sqrt{V_{\delta,N}^J}\end{pmatrix}
    \Rightarrow_J N_2(0,I_2),
    \qquad
    \lambda_N^J:=\frac{V_N^D}{V_N^D+V_{\delta,N}^J}
    \xrightarrow{p,J}\lambda\in[0,1].
    \label{eq:external_pilot_joint_limits}
\end{equation}
Then
\begin{equation}
    \frac{\widehat\tau_N^{\rm BC}-\tau_{w,N}}
         {\sqrt{V_N^D+V_{\delta,N}^J}}
    \Rightarrow_J N(0,1).
    \label{eq:external_pilot_variance_propagation}
\end{equation}
If $\widehat V_N^D/V_N^D\xrightarrow{p,J}1$ and
$\widehat V_{\delta,N}^J/V_{\delta,N}^J\xrightarrow{p,J}1$, replacing the denominator by
$\sqrt{\widehat V_N^D+\widehat V_{\delta,N}^J}$ is valid. One primitive route is a stable conditional CLT for $U_N/\sqrt{V_N^D}$ relative to the pilot sigma-field, together with a Gaussian limit for $D_N/\sqrt{V_{\delta,N}^J}$. The stable limit makes the current-randomization component asymptotically independent of pilot error even when the design itself is pilot-adaptive. Exact conditional centering makes $U_N$ uncorrelated with every pilot-measurable $D_N$ under the joint law, but separate marginal Gaussian limits alone do not imply the required joint Gaussian limit.
\end{proposition}

\paragraph{Meaning.}
When pilot error is of the same order as current-randomization uncertainty, it must be propagated under a joint pilot--experiment law rather than treated as fixed under the conditional randomization law. Stable conditional normality supplies a primitive route even when the design adapts to the pilot.

The incompatibility result therefore does not imply that local randomization is unusable. It identifies which additional ingredient is required. Coherent routes include: estimate and remove the design-bias loading, either under the negligible-pilot condition in Corollary~\ref{cor:external_debiasing_pilot_size} or with uncertainty propagation as in Proposition~\ref{prop:external_pilot_uncertainty}; change the causal target from the sustained all-treated/all-control GATE to an exposure regime that the design can reproduce; retain the GATE but use the finite-sample or sensitivity-based bias envelopes below; or develop a nonlocal common-factor limit theory for a global sampler rather than applying local HAC outside its certified regime. These routes answer different scientific questions and should not be presented as interchangeable variants of a single asymptotic theorem.

\subsection{Misspecification-aware design bias}

The exact covariance bias identity is most informative under the linear exposure working model, but a confidence statement should distinguish working-model bias from misspecification. Write
\begin{equation}
\begin{split}
    Y_{Na}^{\star}(z)=\mu_{Na}+\beta_{Na}z_a
    +\gamma_N\sum_cH^{\comp}_{N,ac}z_c
    +\eta_N\sum_cH^{\lag}_{N,ac}z_c+r_{Na}(z),
\end{split}
\label{eq:working_model_with_remainder}
\end{equation}
where $r_{Na}$ is unrestricted except through a sensitivity envelope below. Define
\begin{align}
    b^{(w)}_{\comp,N}(\Rnorm_N)
    &:=\sum_{a\in\Ocal_N}\sum_c
    w_{Na}H^{\comp}_{N,ac}(\Rnorm_{N,ac}-1),\\
    b^{(w)}_{\lag,N}(\Rnorm_N)
    &:=\sum_{a\in\Ocal_N}\sum_c
    w_{Na}H^{\lag}_{N,ac}(\Rnorm_{N,ac}-1).
\end{align}

\begin{proposition}[Exact bias decomposition under misspecification]
\label{prop:misspecification_bias}
Under \eqref{eq:working_model_with_remainder}, and assuming that measurement noise is fixed under randomization or satisfies the HT orthogonality condition \eqref{eq:noise_ht_orthogonality},
\begin{equation}
\begin{split}
    \E\widehat\tau^{HT}_{w,\tilde\alpha}-\tau_{w,N}
    ={}&\frac{\gamma_N}{W_N}b^{(w)}_{\comp,N}(\Rnorm_N)
       +\frac{\eta_N}{W_N}b^{(w)}_{\lag,N}(\Rnorm_N)
       +B_{{\rm rem},N},
\end{split}
\label{eq:weighted_bias_inference}
\end{equation}
where
\begin{equation}
    B_{{\rm rem},N}
    :=\frac{1}{W_N}\sum_{a\in\Ocal_N}w_{Na}
    \left[\E\{\psi_{Na}(Z_{Na})r_{Na}(Z_N)\}
    -\{r_{Na}(\one)-r_{Na}(\mathbf 0)\}\right].
    \label{eq:remainder_bias_definition}
\end{equation}
If $|\gamma_N|\le\bar\gamma_N$, $|\eta_N|\le\bar\eta_N$, and
$|B_{{\rm rem},N}|\le\bar B_{{\rm rem},N}$, then
\begin{equation}
    |\E\widehat\tau^{HT}_{w,\tilde\alpha}-\tau_{w,N}|
    \le \mathcal B_{\mathrm{bias},N}^\star(\Rnorm_N)
    :=\frac{\bar\gamma_N|b^{(w)}_{\comp,N}(\Rnorm_N)|
       +\bar\eta_N|b^{(w)}_{\lag,N}(\Rnorm_N)|}{W_N}
       +\bar B_{{\rm rem},N}.
    \label{eq:bias_bound_inference}
\end{equation}
A crude outcome-scale envelope is available when
$\sup_z|r_{Na}(z)|\le\bar r_{Na}$:
\begin{equation}
    |B_{{\rm rem},N}|
    \le \frac{4}{W_N}\sum_{a\in\Ocal_N}w_{Na}\bar r_{Na}.
    \label{eq:crude_remainder_bias_bound}
\end{equation}
\end{proposition}

The exact linear model sets $r_{Na}\equiv0$. More generally, $\bar B_{{\rm rem},N}$ is a sensitivity parameter or a bound justified by external validation. The causal CLT follows from Theorem~\ref{thm:cell_clt} when
$\mathcal B_{\mathrm{bias},N}^\star=o(\sqrt{V_N})$; otherwise the same quantity can be added to the confidence-interval half-width.

\subsection{Disconnected metrics and additional HAC results}
\label{app:disconnected_hac_extensions}

Connectedness ensures that every covariance in $\Sigma_N$ belongs to a finite shell. If one instead retains a disconnected natural graph, setting $K(\infty)=0$ deletes every cross-component pair. Empty finite shells cannot justify that deletion. The next result gives the exact additional conditions.

Let $\widetilde{\mathcal G}_N=(\Ocal_N,\widetilde{\mathcal E}_N)$ be an undirected graph with extended shortest-path distance $\widetilde d_N\in\mathbb N_0\cup\{\infty\}$. For finite $s$, define
\begin{align}
\widetilde{\mathcal B}_N(a;s)&:=\{c:\widetilde d_N(a,c)\le s\},
&\partial\widetilde{\mathcal B}_N(a;s)&:=\{c:\widetilde d_N(a,c)=s\},\\
\widetilde\delta_N^\partial(s;k)&:=\frac1{M_N}\sum_a
|\partial\widetilde{\mathcal B}_N(a;s)|^k,\\
\widetilde\Delta_N(s,m;k)&:=\frac1{M_N}\sum_a
\max_{c\in\partial\widetilde{\mathcal B}_N(a;s)}
|\widetilde{\mathcal B}_N(a;m)\setminus
 \widetilde{\mathcal B}_N(c;s-1)|^k,\\
\widetilde c_N(s,m;k)&:=\inf_{\alpha>1}
\{\widetilde\Delta_N(s,m;k\alpha)\}^{1/\alpha}
\left\{\widetilde\delta_N^\partial
\left(s;\frac{\alpha}{\alpha-1}\right)\right\}^{1-1/\alpha},
\label{eq:disconnected_finite_shell_complexity}
\end{align}
with the same empty-shell conventions as before and
$\widetilde{\mathcal B}_N(c;-1)=\varnothing$.

\begin{definition}[Finite-shell graph-$\psi$ dependence]
\label{def:finite_shell_graph_psi}
The centered array is finite-shell graph-$\psi$ dependent on
$\widetilde{\mathcal G}_N$ if there are coefficients
$\{\theta_N(s):s\ge0\}\subset[0,1]$ with $\theta_N(0)=1$ and
$\theta_N(s)$ nonincreasing, and a constant $C_\psi<\infty$, such that
\eqref{eq:graph_psi_definition} holds for every pair of disjoint nonempty sets
$A,B$ satisfying
\[
    s\le \widetilde d_N(A,B)<\infty.
\]
No covariance restriction is imposed solely from
$\widetilde d_N(A,B)=\infty$; cross-component dependence is instead measured by
\eqref{eq:infinite_shell_target_remainder} and
\eqref{eq:infinite_shell_stochastic_remainder}.
\end{definition}

Define
\begin{align}
\widetilde{\mathsf K}_N(a,c)&:=K\{\widetilde d_N(a,c)/b_N\},
&\mathcal I_N(b_N)&:=\{(a,c):\widetilde d_N(a,c)\le b_N\},
\label{eq:disconnected_kernel_pairs}\\
\mathfrak B_{\infty,N}^{\rm abs}
&:=\frac1{M_N}\sum_{a,c:\,\widetilde d_N(a,c)=\infty}
|\Cov_N(\xi_{Na},\xi_{Nc})|,
\label{eq:infinite_shell_target_remainder}
\end{align}
and, writing
$\widetilde d_N(\{a,c\},\{k,l\})=
\min_{u\in\{a,c\},v\in\{k,l\}}\widetilde d_N(u,v)$,
\begin{equation}
\mathfrak S_{\infty,N}(b_N)
:=\frac1{M_N^2}
\sum_{\substack{(a,c),(k,l)\in\mathcal I_N(b_N)\\
\widetilde d_N(\{a,c\},\{k,l\})=\infty}}
\left|\Cov_N(\xi_{Na}\xi_{Nc},\xi_{Nk}\xi_{Nl})\right|.
\label{eq:infinite_shell_stochastic_remainder}
\end{equation}

\begin{proposition}[Oracle network HAC on a disconnected graph]
\label{prop:oracle_network_hac_disconnected}
Suppose Assumption~\ref{ass:weights_moments} holds, Definition~\ref{def:finite_shell_graph_psi} holds, $\inf_N\Sigma_N>0$, and
\begin{align}
&\sum_{s\ge1}|K(s/b_N)-1|\,
\widetilde\delta_N^\partial(s;1)\theta_N(s)^{1-2/\nu}\xrightarrow{p}0,
\label{eq:hac_disconnected_finite_bias_rate}\\
&\frac1{M_N}\sum_{s\ge0}\widetilde c_N(s,b_N;2)
\theta_N(s)^{1-4/\nu}\xrightarrow{p}0,
\label{eq:hac_disconnected_finite_stochastic_rate}\\
&\mathfrak B_{\infty,N}^{\rm abs}\xrightarrow{p}0,
\qquad
\mathfrak S_{\infty,N}(b_N)\xrightarrow{p}0.
\label{eq:hac_infinite_shell_rates}
\end{align}
Then the oracle estimator
\begin{equation}
\widehat\Sigma_N^{\rm or,disc}
:=\frac1{M_N}\sum_{a,c}
\widetilde{\mathsf K}_N(a,c)\xi_{Na}\xi_{Nc}
\label{eq:oracle_hac_disconnected}
\end{equation}
satisfies
\begin{equation}
\widehat\Sigma_N^{\rm or,disc}-\Sigma_N\xrightarrow{p\mid\mathcal F}0,
\qquad
\frac{\widehat\Sigma_N^{\rm or,disc}/M_N}{V_N}\xrightarrow{p\mid\mathcal F}1.
\label{eq:oracle_hac_disconnected_consistency}
\end{equation}
\end{proposition}

\begin{corollary}[Conditionally independent connected components]
\label{cor:component_independent_hac}
If, conditional on $\mathcal F_N$, the sigma-fields generated by contributions in distinct connected components of $\widetilde{\mathcal G}_N$ are mutually independent, then
$\mathfrak B_{\infty,N}^{\rm abs}=\mathfrak S_{\infty,N}(b_N)=0$ exactly. Hence the two finite-shell rates in Proposition~\ref{prop:oracle_network_hac_disconnected} suffice. This is the natural disconnected-graph corollary for genuinely local-factor or dependency-graph architectures whose primitive shocks are component-specific.
\end{corollary}

\begin{example}[Why both infinite-shell conditions are necessary]
\label{ex:disconnected_hac_counterexample}
Let the graph contain $M$ isolated vertices. First take
$\xi_i=\epsilon_i+U/\sqrt M$, where
$\epsilon_1,\ldots,\epsilon_M,U$ are independent $N(0,1)$. Then
$\Sigma_M=2$, while the oracle kernel with $K(\infty)=0$ is
$M^{-1}\sum_i\xi_i^2\to_p1$. Every displayed finite-shell rate holds, but
$\mathfrak B_{\infty,M}^{\rm abs}=(M-1)/M\to1$.

The target remainder alone is not enough. Take instead
$\xi_i=E_iU$, where $E_i$ are independent Rademacher variables and
$U\sim N(0,1)$ is independent. Now all cross-vertex covariances vanish, so
$\mathfrak B_{\infty,M}^{\rm abs}=0$ and $\Sigma_M=1$, but the same oracle estimator equals $U^2$ and does not converge to one. Here
$\mathfrak S_{\infty,M}(b_M)=2(M-1)/M\to2$. Thus cross-component covariance of included pair-products must also be controlled.
\end{example}

From this point onward, $\mathsf K_N$ and $\widehat\Sigma_N^{\rm or}$ denote either the connected objects certified by Proposition~\ref{prop:oracle_network_hac_connected} or the disconnected objects certified by Proposition~\ref{prop:oracle_network_hac_disconnected}. The first rate in either proposition controls finite-shell covariance truncation and kernel smoothing; the second controls finite-shell stochastic error. In the disconnected case, \eqref{eq:hac_infinite_shell_rates} is an additional substantive requirement, not a consequence of empty shells.

\begin{corollary}[Primitive PSD-kernel centering rate]
\label{cor:primitive_centering_rate}
Under either oracle-HAC proposition, suppose $\mathsf K_N\succeq0$ and $\|\mathsf K_N\|_{\rm op}\le\Lambda_N$. The single primitive condition
\begin{equation}
    \|\widehat q_N^0-\bar q_N\|_2^2
    =o_p\!\left(\frac{\sigma_N^2}{\Lambda_N}\right)
    \label{eq:primitive_centering_l2_rate}
\end{equation}
is sufficient for \eqref{eq:centering_error_conditions}. If $\sigma_N^2\asymp M_N$ and $\Lambda_N\lesssim D_N(b_N)$, where $D_N(b_N)$ is the maximum bandwidth-ball size for the selected kernel graph, then it is enough that
\begin{equation}
    M_N^{-1}\|\widehat q_N^0-\bar q_N\|_2^2
    =o_p\{D_N(b_N)^{-1}\}.
    \label{eq:average_centering_mse_rate}
\end{equation}
\end{corollary}

\begin{corollary}[Operator-norm route for a possibly indefinite kernel]
\label{cor:indefinite_kernel_centering}
Under either oracle-HAC proposition, let $\mathsf K_N$ be symmetric with
$\|\mathsf K_N\|_{\rm op}\le\Lambda_N$; positive semidefiniteness is not required. Suppose
$\|\xi_N\|_2=O_{p\mid\mathcal F}(\sqrt{M_N})$ and, with
$r_N=\widehat q_N^0-\bar q_N$,
\begin{equation}
\Lambda_N\|r_N\|_2^2=o_p(\sigma_N^2),
\qquad
\Lambda_N\sqrt{M_N}\|r_N\|_2=o_p(\sigma_N^2).
\label{eq:indefinite_kernel_centering_rates}
\end{equation}
Then $\widehat V_N^{\rm MA,+}/V_N\xrightarrow{p\mid\mathcal F}1$.
Under the external-pilot model \eqref{eq:external_pilot_mean_model}--\eqref{eq:external_pilot_parameter_rate}, if $\sigma_N^2\asymp M_N$ and $\mathsf K_N(a,a)=1$, the single sufficient rate
\begin{equation}
\frac{\Lambda_N^2d_{q,N}}{n_{0,N}}\xrightarrow{p}0
\label{eq:indefinite_external_pilot_rate}
\end{equation}
implies \eqref{eq:indefinite_kernel_centering_rates}. This route is more demanding than the PSD seminorm argument and is therefore kept in the supplement.
\end{corollary}

There is also a deliberately conservative alternative that avoids estimating heterogeneous means. Suppose a symmetric kernel matrix $\mathsf K_N\succeq0$ is constructed, for example as a Gram matrix of overlapping graph-block indicators. Define
\begin{equation}
    \widehat\Sigma_N^{\rm UC}:=\frac{1}{M_N}Q_N^\top\mathsf K_NQ_N,
    \qquad
    \widehat V_N^{\rm UC}:=\frac{\widehat\Sigma_N^{\rm UC}}{M_N}.
    \label{eq:uncentered_psd_hac}
\end{equation}
A generic distance kernel need not produce a positive-semidefinite matrix on an arbitrary graph.

\begin{proposition}[PSD uncentered conservative variance]
\label{prop:uncentered_conservative}
Let
\begin{equation}
    \Sigma_{K,N}:=\frac{1}{M_N}
    \sum_{a,c}\mathsf K_N(a,c)\Cov(\xi_{Na},\xi_{Nc}).
    \label{eq:smoothed_long_run_variance}
\end{equation}
If $\mathsf K_N\succeq0$,
$\Sigma_{K,N}\ge\Sigma_N-o(\Sigma_N)$, and
$\widehat\Sigma_N^{\rm UC}-\E\widehat\Sigma_N^{\rm UC}=o_p(\Sigma_N)$, then
\begin{equation}
    \widehat\Sigma_N^{\rm UC}\ge\Sigma_N-o_p(\Sigma_N).
    \label{eq:uncentered_conservative_result}
\end{equation}
Thus $\widehat V_N^{\rm UC}$ gives asymptotically conservative normal inference whenever the centered CLT holds. The price is potentially substantial overcoverage through the nonnegative term
$M_N^{-1}\bar q_N^\top\mathsf K_N\bar q_N$.
\end{proposition}

Let $\widehat V_N\ge0$ denote either the positive-part model-assisted estimator or a conservative PSD estimator. With a nonnegative valid sensitivity bound $\widehat B_N\ge0$, the bias-aware interval is
\begin{equation}
    \mathrm{CI}_{\rm BA}
    =\left[\widehat\tau^{HT}_{w,\tilde\alpha}
    \pm\left\{z_{1-\alpha_0/2}\sqrt{\widehat V_N}+\widehat B_N\right\}\right].
    \label{eq:bias_aware_ci}
\end{equation}
\begin{proposition}[Bias-aware coverage]
\label{prop:bias_aware_coverage}
Suppose the centered CLT in \eqref{eq:clt_centered} holds, $\widehat V_N\ge0$, and $\widehat B_N\ge0$. Assume there is a deterministic sequence $\varepsilon_N\downarrow0$ such that
\begin{equation}
    \Pr\{\widehat V_N\ge(1-\varepsilon_N)V_N\}\longrightarrow1.
    \label{eq:one_sided_variance_calibration}
\end{equation}
Let $s_N\ge0$ be $\mathcal F_N$-measurable with $s_N/\sqrt{V_N}\xrightarrow{p}0$. If, for events $\mathcal E_N$ with $\liminf_N\Pr(\mathcal E_N)\ge1-\delta$,
\begin{equation}
    |\E\widehat\tau^{HT}_{w,\tilde\alpha}-\tau_{w,N}|
    \le\widehat B_N+s_N
    \quad\text{on }\mathcal E_N,
    \label{eq:estimated_bias_envelope_event}
\end{equation}
then the interval in \eqref{eq:bias_aware_ci} has
\begin{equation}
    \liminf_{N\to\infty}
    \Pr\{\tau_{w,N}\in\mathrm{CI}_{\rm BA}\}
    \ge1-\alpha_0-\delta.
    \label{eq:bias_aware_coverage_bound}
\end{equation}
A fixed valid envelope has $\delta=0$ and $s_N=0$.
\end{proposition}

When the exposure coefficients and operator have known sign, the cut representation yields a shorter asymmetric interval. Suppose, for example, that the structural positive-exposure bias obeys
$-\widehat B_N-s_N\le\E\widehat\tau^{HT}_{w,\tilde\alpha}-\tau_{w,N}\le0$ on a high-probability event. Fix tail allocations $\alpha_L,\alpha_U>0$ with $\alpha_L+\alpha_U=\alpha_0$ and define
\begin{equation}
    \mathrm{CI}_{\rm dir}(\alpha_L,\alpha_U)
    =\left[
      \widehat\tau^{HT}_{w,\tilde\alpha}
      -z_{1-\alpha_L}\sqrt{\widehat V_N},\quad
      \widehat\tau^{HT}_{w,\tilde\alpha}
      +z_{1-\alpha_U}\sqrt{\widehat V_N}
      +\widehat B_N
      \right].
    \label{eq:direction_aware_ci}
\end{equation}

\begin{proposition}[Direction-aware bias coverage]
\label{prop:direction_aware_coverage}
Under the centered CLT and one-sided variance calibration \eqref{eq:one_sided_variance_calibration}, suppose there are events $\mathcal E_N$ with
$\liminf_N\Pr(\mathcal E_N)\ge1-\delta$ and an $\mathcal F_N$-measurable
$s_N\ge0$ with $s_N/\sqrt{V_N}\xrightarrow{p}0$ such that
\begin{equation}
    -\widehat B_N-s_N
    \le \E\widehat\tau^{HT}_{w,\tilde\alpha}-\tau_{w,N}
    \le s_N
    \quad\text{on }\mathcal E_N.
    \label{eq:directional_bias_event}
\end{equation}
Then, for every fixed $\alpha_L,\alpha_U>0$ with sum $\alpha_0$,
\begin{equation}
    \liminf_{N\to\infty}
    \Pr\{\tau_{w,N}\in\mathrm{CI}_{\rm dir}(\alpha_L,\alpha_U)\}
    \ge1-\alpha_0-\delta.
    \label{eq:direction_aware_coverage_bound}
\end{equation}
For an exact positive-exposure model, Proposition~\ref{prop:expected_cut_geometry} supplies the nonpositive bias direction. Equal tails are valid but not mandatory; a prespecified length or decision loss can determine $\alpha_L$ and $\alpha_U$. Relative to the symmetric interval \eqref{eq:bias_aware_ci}, \eqref{eq:direction_aware_ci} avoids adding the full bias envelope to the lower endpoint.
\end{proposition}

\subsection{Model-assisted design-aware tail correction}
\label{subsec:design_aware_variance}

A local HAC kernel estimates $\Sigma_{K,N}$ and omits the signed covariance tail
\begin{equation}
    \mathcal T_N:=\Sigma_N-\Sigma_{K,N}.
    \label{eq:signed_covariance_tail}
\end{equation}
For a connected certified sequence satisfying the kernel-bias condition in Proposition~\ref{prop:oracle_network_hac_connected}, this tail is asymptotically negligible at the chosen bandwidth. On a disconnected graph, the same statement additionally requires the infinite-shell target remainder in \eqref{eq:hac_infinite_shell_rates}; otherwise the omitted cross-component covariance is part of the tail. For a finite-sample optimized design the tail can be material, and the known randomization law can be used for a model-assisted correction.

Let $\widehat X_{Na}^{\rm pl}(z)$ be a plug-in potential-outcome surface constructed from pre-treatment information, an external sample, or a dependence-aware cross-fitted model that satisfies the conditions below conditionally on its training sigma-field. For a fresh draw $Z_N^*\sim\mathcal D_N$, define
\begin{equation}
    \widehat Q_{Na}^{\rm pl}(Z_N^*)
    =\omega_{Na}\psi_{Na}(Z_{Na}^*)\widehat X_{Na}^{\rm pl}(Z_N^*).
    \label{eq:plugin_contribution}
\end{equation}
The plug-in full and local long-run variances are
\begin{align}
    \widehat\Sigma_{\rm pl,full}
    &:=\frac{1}{M_N}\Var_{\mathcal D_N}
       \left\{\sum_a\widehat Q_{Na}^{\rm pl}(Z_N^*)\right\},
    \label{eq:plug_full_variance}\\
    \widehat\Sigma_{\rm pl,loc}
    &:=\frac{1}{M_N}\sum_{a,c}\mathsf K_N(a,c)
       \Cov_{\mathcal D_N}\{\widehat Q_{Na}^{\rm pl}(Z_N^*),
                              \widehat Q_{Nc}^{\rm pl}(Z_N^*)\},
    \label{eq:plug_local_variance}
\end{align}
and $\widehat{\mathcal T}_{\rm pl}=\widehat\Sigma_{\rm pl,full}-\widehat\Sigma_{\rm pl,loc}$. For any local estimator $\widehat\Sigma_{\rm loc}$, define the signed and positive-part corrections
\begin{align}
    \widehat\Sigma_{\rm DA}^{\rm sgn}
    &=\widehat\Sigma_{\rm loc}+\widehat{\mathcal T}_{\rm pl},
    \label{eq:design_aware_signed}\\
    \widehat\Sigma_{\rm DA,raw}^{+}
    &=\widehat\Sigma_{\rm loc}+[\widehat{\mathcal T}_{\rm pl}]_+,
    \qquad
    \widehat V_{\rm DA}^{+}=\frac{[\widehat\Sigma_{\rm DA,raw}^{+}]_+}{M_N}.
    \label{eq:design_aware_hac}
\end{align}

The consistency conditions used by Proposition~\ref{prop:design_aware_hac} are
\begin{equation}
    \widehat\Sigma_{\rm loc}-\Sigma_{K,N}=o_p(\Sigma_N),
    \qquad
    \widehat{\mathcal T}_{\rm pl}-\mathcal T_N=o_p(\Sigma_N).
    \label{eq:tail_consistency}
\end{equation}
Under these conditions, the signed correction is ratio-consistent and the outer-truncated positive-tail variance has the one-sided guarantee stated in the main text. The negative true tail determines whether the positive-tail version is also ratio-consistent.

A useful sufficient condition for the full plug-in approximation is
\begin{equation}
    \Var_{\mathcal D_N}\left\{\sum_a
    (\widehat Q_{Na}^{\rm pl}-Q_{Na})\right\}=o_p(\sigma_N^2).
    \label{eq:plugin_full_l2_condition}
\end{equation}
Here and in the next display, $Q_N=Q_N(Z_N^*)$ denotes the true contribution vector evaluated under the same fresh design draw as $\widehat Q_N^{\rm pl}$. For a PSD local kernel, a corresponding sufficient condition is
\begin{equation}
    \E_{\mathcal D_N}\left[(e_N-\E e_N)^\top
    \mathsf K_N(e_N-\E e_N)\right]=o_p(\sigma_N^2),
    \qquad e_N=\widehat Q_N^{\rm pl}-Q_N,
    \label{eq:plugin_local_l2_condition}
\end{equation}
provided the true local quadratic variance is $O(\sigma_N^2)$. Cauchy--Schwarz then controls the cross terms. If the design moments in \eqref{eq:plug_full_variance}--\eqref{eq:plug_local_variance} are estimated with $R_N$ independent randomization draws, Monte Carlo error is $o_p(\Sigma_N)$ whenever $R_N\to\infty$ and the corresponding standardized fourth moments of the aggregate and local quadratic forms are uniformly bounded.

The fixed-outcome substitution
$\widehat X_{Na}^{\rm pl}(z)=Y_{Na}^{\rm obs}-\tilde\alpha_{Na}$ is a useful randomization diagnostic, but under interference it does not by itself satisfy \eqref{eq:plugin_full_l2_condition}: counterfactual assignments generally change the potential outcome surface. Accordingly, the design-aware correction is formally model-assisted. It calibrates a covariance tail; it does not restore a CLT for a fixed-rank common-factor design that fails Assumption~\ref{ass:weak_dependence_rates}.

Combining a centered CLT, one-sided variance calibration, and the misspecification-aware bias envelope gives
\begin{equation}
    \mathrm{CI}_{\rm DA+BA}
    =\left[\widehat\tau^{HT}_{w,\tilde\alpha}
    \pm\left\{z_{1-\alpha_0/2}\sqrt{\widehat V_{\rm DA}^{+}}
    +\widehat B_N\right\}\right].
    \label{eq:design_aware_bias_ci}
\end{equation}

\subsection{Ratio estimators: H{\'a}jek and difference in means}
\label{subsec:ratio_estimators}

Denominator consistency alone does not make the weighted H{\'a}jek estimator first-order equivalent to HT. Define the four-dimensional contribution
\begin{equation}
    R_{Na}:=
    \begin{pmatrix}
    \omega_{Na}Z_{Na}X_{Na}^{\rm obs}\\
    \omega_{Na}(1-Z_{Na})X_{Na}^{\rm obs}\\
    \omega_{Na}Z_{Na}\\
    \omega_{Na}(1-Z_{Na})
    \end{pmatrix},
    \qquad
    \bar R_N=\frac{1}{M_N}\sum_aR_{Na},
    \qquad
    h(x)=\frac{x_1}{x_3}-\frac{x_2}{x_4}.
    \label{eq:hajek_vector}
\end{equation}
Define the nonzero-denominator event
\begin{equation}
    \mathcal D_N^H:=\{\bar R_{N,3}>0,\ \bar R_{N,4}>0\}.
    \label{eq:hajek_nonzero_denominator_event}
\end{equation}
On $\mathcal D_N^H$, set $\widehat\tau^H_{w,\tilde\alpha}=h(\bar R_N)$. On its complement, define the estimator to equal zero, or any other fixed prespecified value. Proposition~\ref{prop:hajek_delta} implies $\Pr\{(\mathcal D_N^H)^c\}\to0$, so this finite-sample convention does not affect the limit.

Define
\begin{equation}
    S_{R,N}:=\frac{1}{\sqrt{M_N}}\sum_{a\in\Ocal_N}
    \{R_{Na}-\E R_{Na}\},
    \qquad \Gamma_N:=\Var(S_{R,N}).
    \label{eq:hajek_score_covariance}
\end{equation}

\begin{proposition}[Multivariate weak-dependence delta method]
\label{prop:hajek_delta}
Suppose $\rho_N:=\E\bar R_N\to\rho$ with $\rho_3,\rho_4>0$, $\bar R_N-\rho_N=o_p(1)$, and $\Gamma_N\to\Gamma_R$ for a finite positive-semidefinite matrix $\Gamma_R$. For every fixed $\lambda\in\mathbb R^4$ with $\lambda^\top\Gamma_R\lambda>0$, suppose the scalar array
$\{\lambda^\top(R_{Na}-\E R_{Na}):a\in\Ocal_N\}$ satisfies the uniform $L^\nu$ requirement of Assumption~\ref{ass:weights_moments} and the graph-$\psi$/rate conditions of Assumption~\ref{ass:weak_dependence_rates}, with sum variance
\begin{equation}
    \sigma_{\lambda,N}^2
    =M_N\lambda^\top\Gamma_N\lambda.
    \label{eq:hajek_projection_variance}
\end{equation}
Then
\begin{equation}
    S_{R,N}\Rightarrow N_4(0,\Gamma_R),
    \label{eq:hajek_multivariate_clt}
\end{equation}
and
\begin{equation}
    \sqrt{M_N}\{\widehat\tau^H_{w,\tilde\alpha}-h(\rho_N)\}
    \Rightarrow N\{0,\nabla h(\rho)^\top\Gamma_R\nabla h(\rho)\}.
    \label{eq:hajek_delta_limit}
\end{equation}
If $h(\rho_N)-\tau_{w,N}=o(M_N^{-1/2})$, the same limit is centered at the causal target.
\end{proposition}

A stronger design-balance condition,
\begin{equation}
    \sum_a\omega_{Na}(Z_{Na}-p_N)=o_p(\sqrt{M_N})
    \label{eq:strong_balance_ratio}
\end{equation}
removes the denominator corrections and yields first-order HT equivalence under the remaining moment conditions; the corresponding control-denominator condition follows identically because $M_N^{-1}\sum_a\omega_{Na}=1$. Under ordinary denominator consistency the corrections are generally of order $M_N^{-1/2}$ and must be retained through Proposition~\ref{prop:hajek_delta}. The unweighted difference-in-means estimator is the special case $\omega_{Na}\equiv1$.

\section{Additional Experimental Details and Results}
\label{app:additional_experiments}

\IfFileExists{tables/tab_inference_studentized.tex}{
\subsection{Studentized diagnostics for the supplementary inference experiment}
The generated table reports the fifth, median, and ninety-fifth percentiles and standard deviation of the statistic centered at the independently calibrated design expectation. These diagnostics complement the coverage and width comparisons in Subsection~\ref{subsec:inference_experiments}. For the positive-part design-aware estimator, the standard deviations are close to one and the central quantiles are close to the corresponding Gaussian benchmarks across the three outcome models.
\begin{table}[H]
\centering
\caption{Studentized diagnostics for the positive-tail design-aware procedure under linear, nonlinear, and demand-substitution outcomes.}
\label{tab:inference_studentized}
\footnotesize
\input{tables/tab_inference_studentized.tex}
\end{table}
}{}

\IfFileExists{figures/fig_inference_gate_coverage.pdf}{
\begin{figure}[!htbp]
\centering
\includegraphics[width=0.98\linewidth]{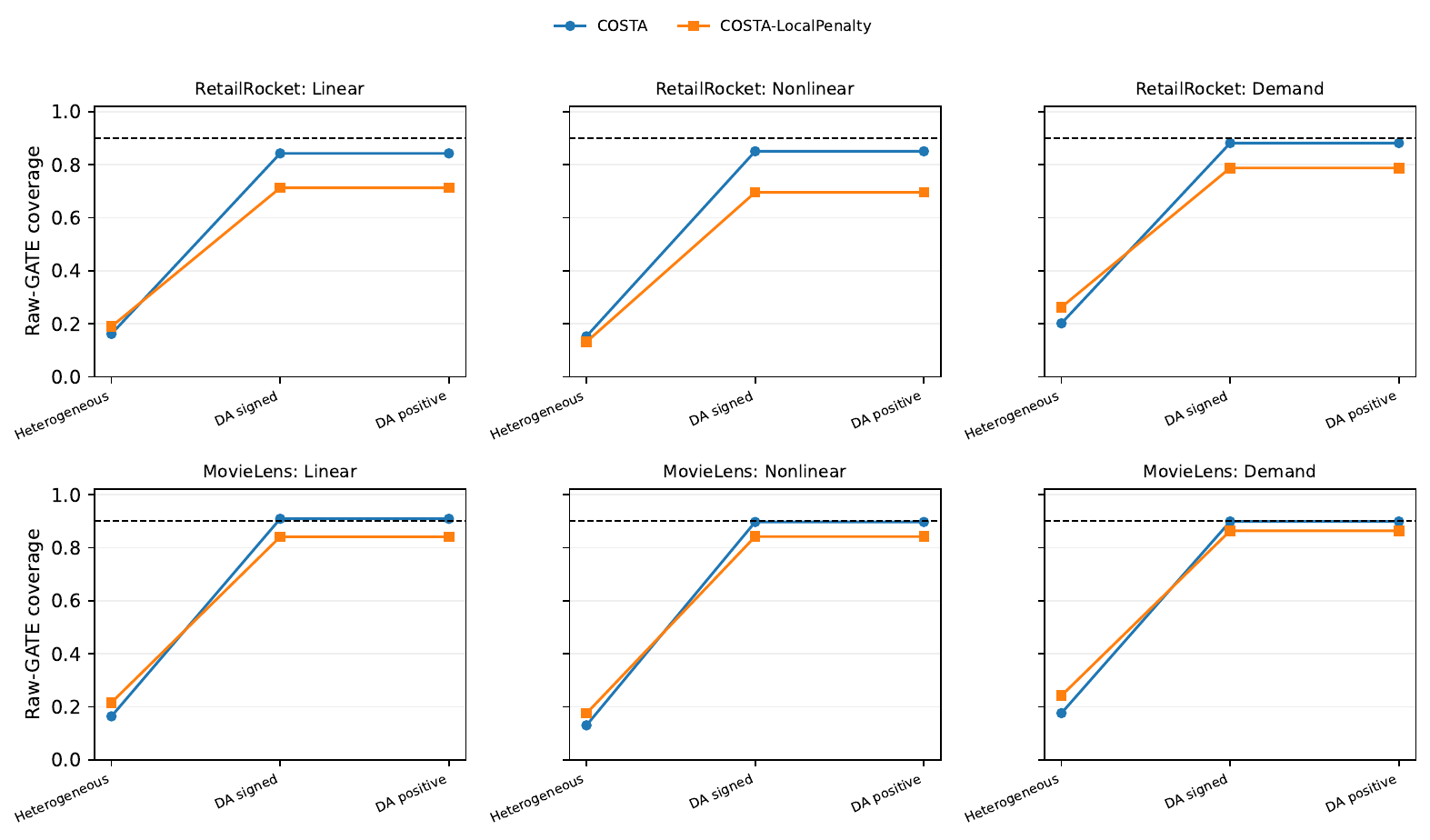}
\caption{Raw-GATE coverage at the nominal $90\%$ level. Differences from Figure~\ref{fig:inference_design_coverage} isolate the causal-centering requirement: a calibrated variance estimator need not remove design bias relative to the sustained-treatment GATE.}
\label{fig:inference_gate_coverage}
\end{figure}

\begin{figure}[!htbp]
\centering
\includegraphics[width=0.98\linewidth]{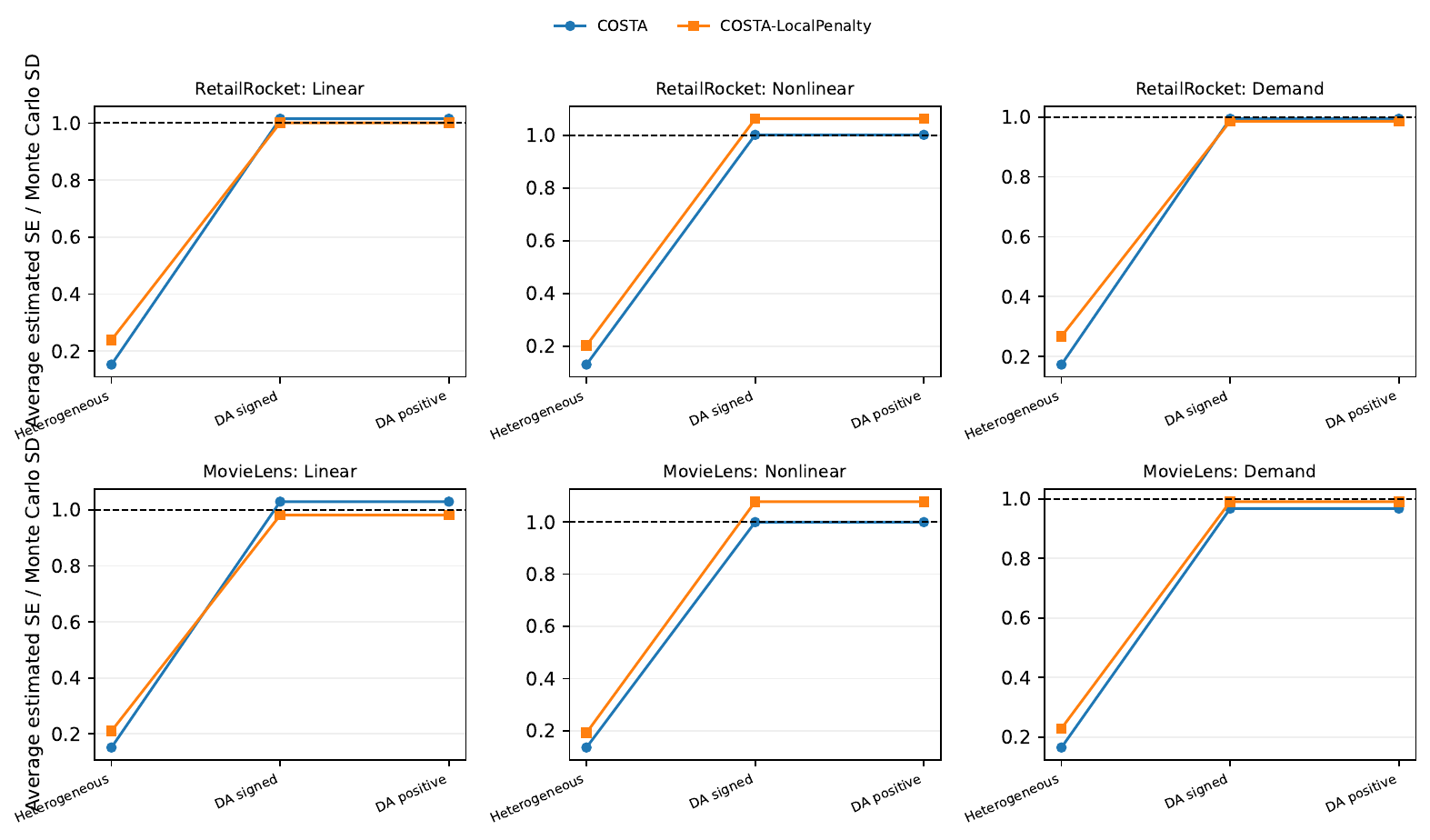}
\caption{Ratio of the average estimated standard error to the Monte Carlo standard deviation. The horizontal reference is one. The design-aware procedures recover the nonlocal covariance tail omitted by the local product-graph estimators.}
\label{fig:inference_se_ratio}
\end{figure}
}{}

\subsection{Why double balance reduces variance in multiple-unit switchbacks}
\label{app:double_balance}

The precision benefit of balancing both unit exposure and temporal-block load can be seen in a simple additive model. Suppose, only for this calculation, that outcomes obey
\begin{equation}
    Y_{ib}(z)=\alpha_i+\lambda_b+\tau z_{ib}+\varepsilon_{ib},
    \label{eq:double_balance_model}
\end{equation}
with no cross-unit interference. For the common-$p$ HT score, the baseline part of the estimator is
\begin{equation}
    \frac{1}{NBp(1-p)}\sum_{i=1}^N\sum_{b=1}^B (Z_{ib}-p)(\alpha_i+\lambda_b)
    =\frac{1}{NBp(1-p)}\left\{\sum_i\alpha_i(T_i-pB)+\sum_b\lambda_b(L_b-pN)\right\},
    \label{eq:double_balance_decomp}
\end{equation}
where $T_i=\sum_bZ_{ib}$ is the number of treated blocks for unit $i$ and $L_b=\sum_iZ_{ib}$ is the treated load in temporal block $b$. If the design is doubly balanced, $T_i=pB$ for all units and $L_b=pN$ for all temporal blocks whenever these quantities are integer, the deterministic unit and block heterogeneity components vanish exactly. Under independent assignment, the same two components contribute variance proportional to $B\sum_i\alpha_i^2$ and $N\sum_b\lambda_b^2$ up to covariance terms. This calculation explains why multiple-unit balanced switchbacks can be efficient even without cross-unit interference: they remove low-frequency unit and time heterogeneity from the randomization score. COSTA uses this idea as a soft regularizer, while the main covariance objective additionally targets network and temporal exposure bias.

\subsection{Implementation details for benchmark designs}

The RBSD baseline in the main text is reported as \textsc{RBSD}. In implementation, exact item balance is used when feasible; for general treatment share $p$, item row sums are constrained to $\lfloor pB\rfloor$ or $\lceil pB\rceil$ and block totals are kept close to $pN$. This is the implementation detail behind the internal label \texttt{rbsd\_approx}; the paper uses the cleaner method name \textsc{RBSD}. The Static-OCD benchmark follows the sign-Gaussian implementation of optimized covariance design \citep{chen2023ocd} and is therefore reported in the general-$p$ ablation only when $p=0.5$.

\subsection{MovieLens setting ablations}

\begin{figure}[H]
    \centering
    \includegraphics[width=0.98\linewidth]{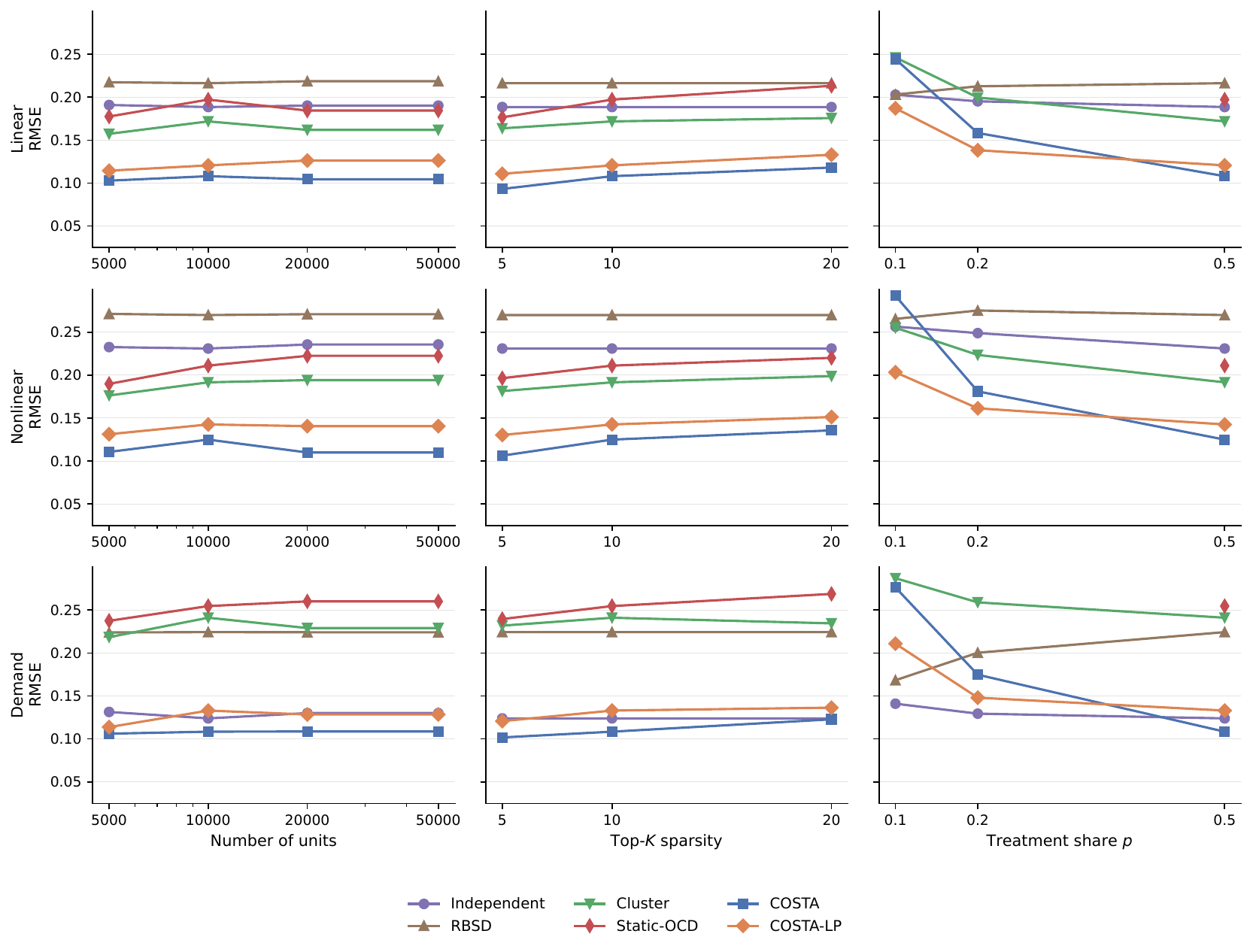}
    \caption{MovieLens setting ablations by potential outcome model. Rows correspond to the linear, nonlinear, and demand models. Columns vary item count $N\in\{5{,}000,10{,}000,20{,}000,50{,}000\}$, top-$K\in\{5,10,20\}$ graph sparsity, and treatment probability $p\in\{0.1,0.2,0.5\}$. The estimator is weighted H{\'a}jek; the design objective is unweighted and the estimand is traffic-weighted.}
    \label{fig:movielens_setting_ablation}
\end{figure}

Figure~\ref{fig:movielens_setting_ablation} is the MovieLens counterpart of Figure~\ref{fig:setting_ablation}. It uses the same item-count, top-$K$, and treatment-share grids and reports the same three potential outcome models. The broad pattern is consistent with RetailRocket in the item-count and graph-sparsity grids. In the low-treatment-share grid the advantage of COSTA can attenuate, especially under the more misspecified demand model, which is consistent with the asymmetric feasible covariance region for Bernoulli margins.

Tables~\ref{tab:retail_item_count_rmse}--\ref{tab:movielens_p_rmse} provide compact RMSE decompositions for the setting-ablation grids in both datasets. These tables are included to make the appendix symmetric across RetailRocket and MovieLens and to avoid averaging across potential outcome models.

\begin{table}[H]
    \centering
    \caption{RetailRocket item-count ablation RMSE. Columns are $N\in\{5{,}000,10{,}000,20{,}000,50{,}000\}$; the setting fixes $B=4$, $K=10$, $p=0.5$, and weighted H{\'a}jek estimation.}
    \label{tab:retail_item_count_rmse}
    \input{tables/tab_retail_item_count_rmse.tex}
\end{table}

\begin{table}[H]
    \centering
    \caption{MovieLens item-count ablation RMSE. Columns are $N\in\{5{,}000,10{,}000,20{,}000,50{,}000\}$; the setting fixes $B=4$, $K=10$, $p=0.5$, and weighted H{\'a}jek estimation.}
    \label{tab:movielens_item_count_rmse}
    \input{tables/tab_movielens_item_count_rmse.tex}
\end{table}

\begin{table}[H]
    \centering
    \caption{RetailRocket top-$K$ sparsity ablation RMSE. Columns are $K\in\{5,10,20\}$; the setting fixes $N=10{,}000$, $B=4$, $p=0.5$, and weighted H{\'a}jek estimation.}
    \label{tab:retail_topk_rmse}
    \input{tables/tab_retail_topk_rmse.tex}
\end{table}

\begin{table}[H]
    \centering
    \caption{MovieLens top-$K$ sparsity ablation RMSE. Columns are $K\in\{5,10,20\}$; the setting fixes $N=10{,}000$, $B=4$, $p=0.5$, and weighted H{\'a}jek estimation.}
    \label{tab:movielens_topk_rmse}
    \input{tables/tab_movielens_topk_rmse.tex}
\end{table}

\begin{table}[H]
    \centering
    \caption{RetailRocket treatment-share ablation RMSE. Columns are $p\in\{0.1,0.2,0.5\}$; the setting fixes $N=10{,}000$, $B=4$, $K=10$, and weighted H{\'a}jek estimation.}
    \label{tab:retail_p_rmse}
    \input{tables/tab_retail_p_rmse.tex}
\end{table}

\begin{table}[H]
    \centering
    \caption{MovieLens treatment-share ablation RMSE. Columns are $p\in\{0.1,0.2,0.5\}$; the setting fixes $N=10{,}000$, $B=4$, $K=10$, and weighted H{\'a}jek estimation.}
    \label{tab:movielens_p_rmse}
    \input{tables/tab_movielens_p_rmse.tex}
\end{table}

\subsection{Temporal-block horizon ablations}

Figure~\ref{fig:block_horizon} reports the RetailRocket temporal-block-horizon ablation for $B\in\{4,8,16\}$. Figure~\ref{fig:movielens_block_horizon} gives the same ablation on MovieLens. Longer horizons provide more within-unit temporal contrast and more opportunities for balancing, but the relative advantage of covariance-optimized designs remains visible in each potential outcome model.

\begin{figure}[H]
    \centering
    \includegraphics[width=0.98\linewidth]{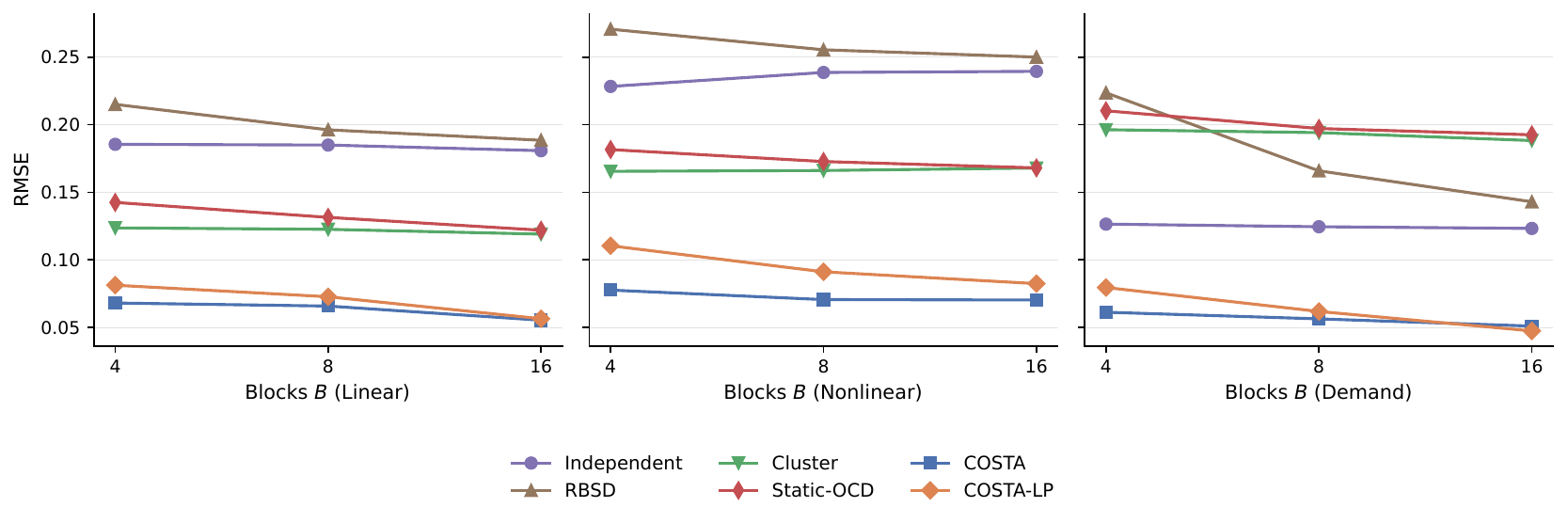}
    \caption{RetailRocket temporal-block-horizon ablation by potential outcome model. The grid is $B\in\{4,8,16\}$, with $N=10{,}000$, top-$K=10$, $p=0.5$, and weighted H{\'a}jek estimation.}
    \label{fig:block_horizon}
\end{figure}

\begin{figure}[H]
    \centering
    \includegraphics[width=0.98\linewidth]{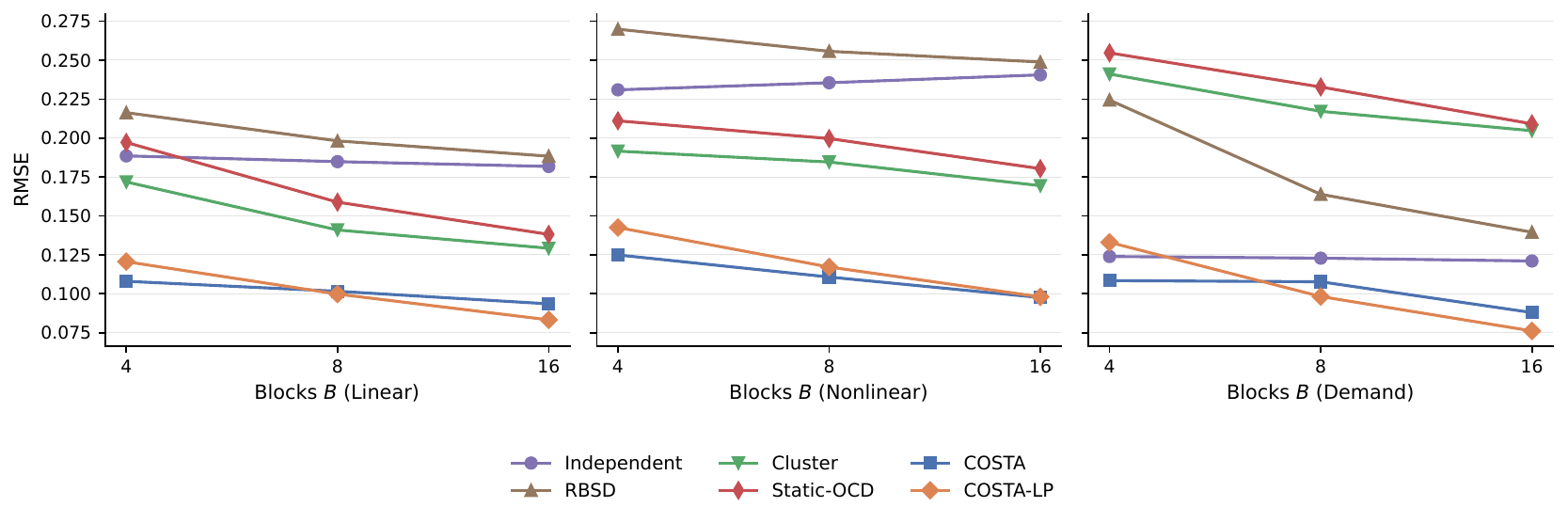}
    \caption{MovieLens temporal-block-horizon ablation by potential outcome model. The grid is $B\in\{4,8,16\}$, with $N=10{,}000$, top-$K=10$, $p=0.5$, and weighted H{\'a}jek estimation.}
    \label{fig:movielens_block_horizon}
\end{figure}

\begin{table}[H]
    \centering
    \caption{RetailRocket temporal-block-horizon ablation RMSE. Columns are $B\in\{4,8,16\}$; the setting fixes $N=10{,}000$, $K=10$, $p=0.5$, and weighted H{\'a}jek estimation.}
    \label{tab:retail_block_rmse}
    \input{tables/tab_retail_block_rmse.tex}
\end{table}

\begin{table}[H]
    \centering
    \caption{MovieLens temporal-block-horizon ablation RMSE. Columns are $B\in\{4,8,16\}$; the setting fixes $N=10{,}000$, $K=10$, $p=0.5$, and weighted H{\'a}jek estimation.}
    \label{tab:movielens_block_rmse}
    \input{tables/tab_movielens_block_rmse.tex}
\end{table}

\subsection{MovieLens bias--standard deviation and Kronecker-dimension checks}

Figure~\ref{fig:movielens_bias_std} is the MovieLens counterpart of Figure~\ref{fig:bias_std}. Figure~\ref{fig:movielens_kro_dim} is the MovieLens counterpart of Figure~\ref{fig:kro_dim}. The same qualitative conclusions hold: COSTA is close to the low-bias, low-variance region, and increasing the Kronecker network dimension from $16$ to $64$ has limited marginal value in this grid.

\begin{figure}[H]
    \centering
    \includegraphics[width=0.98\linewidth]{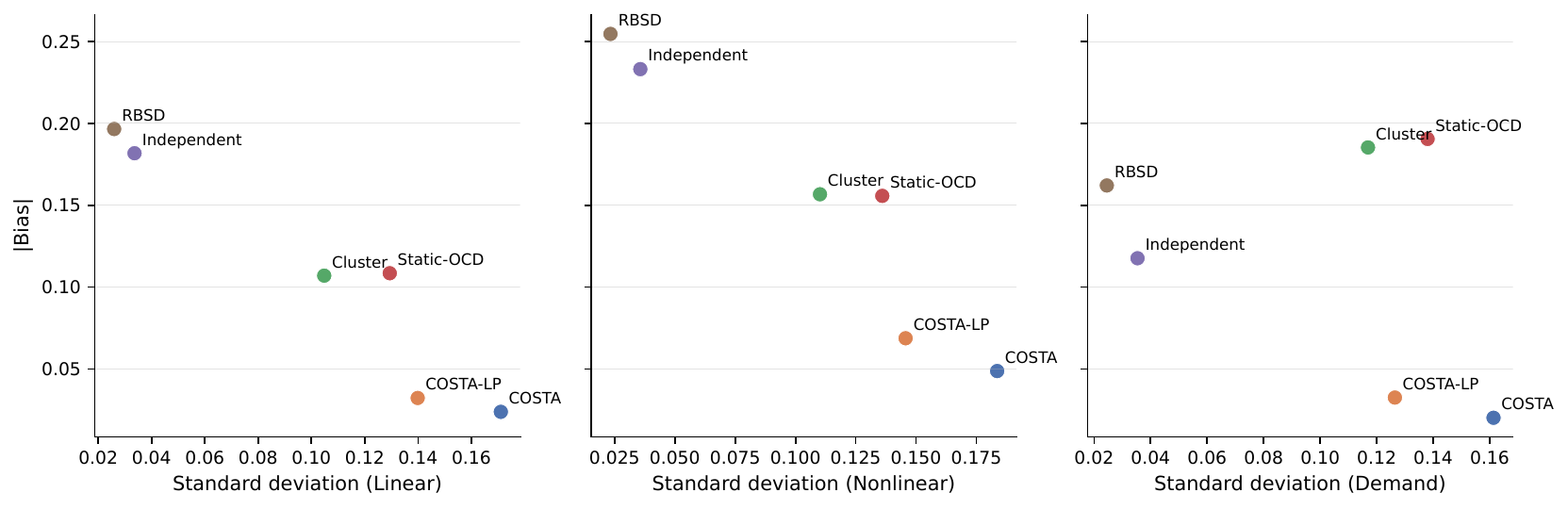}
    \caption{Bias--standard deviation tradeoff in the MovieLens default setting. Each panel is a separate potential outcome model; the estimator is weighted HT and the setting is $N=10{,}000$, $B=8$, top-$K=10$, and $p=0.5$.}
    \label{fig:movielens_bias_std}
\end{figure}

\begin{figure}[H]
    \centering
    \includegraphics[width=0.98\linewidth]{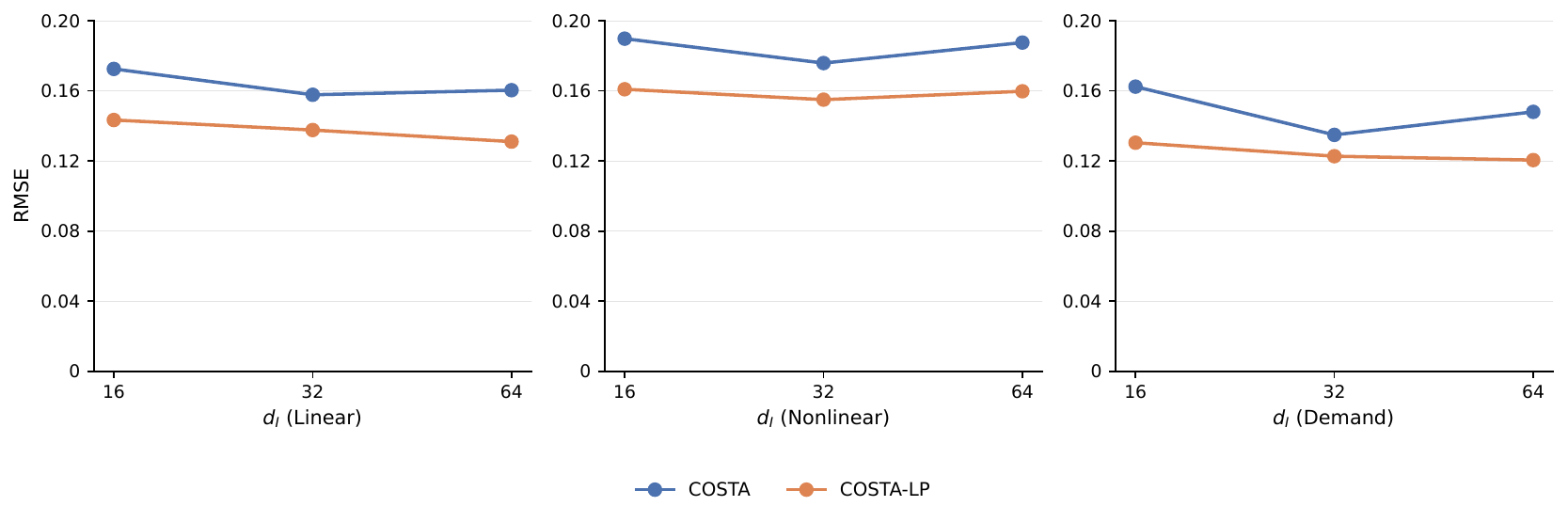}
    \caption{COSTA Kronecker network-dimension ablation on MovieLens. Each panel is a separate potential outcome model; $d_I\in\{16,32,64\}$, $d_T=B-1$, $N=10{,}000$, $B=8$, top-$K=10$, and $p=0.5$. The estimator is weighted HT.}
    \label{fig:movielens_kro_dim}
\end{figure}

\subsection{Cluster/Kronecker ablations on MovieLens}

Table~\ref{tab:cluster_kro_movielens} mirrors the RetailRocket cluster/Kronecker table in the main text. The table separates cluster-level coarsening from the Kronecker reduction of the covariance design space.

\begin{table}[H]
    \centering
    \caption{Cluster and Kronecker design-space ablation on MovieLens. Rows report separate potential outcome models in the default setting with weighted H{\'a}jek.}
    \label{tab:cluster_kro_movielens}
    \input{tables/tab_cluster_kro_movielens.tex}
\end{table}

\subsection{Soft locality-penalty strength}

Figure~\ref{fig:movielens_lambda} gives the MovieLens counterpart of Figure~\ref{fig:lambda_ablation}. Increasing $\lambda$ shrinks the logged \texttt{far\_product\_mean\_abs\_R} diagnostic while changing RMSE because the penalty also limits exposure-edge covariance. As defined after \eqref{eq:local_penalty_exact}, this is an absolute product-cell diagnostic rather than the full squared locality penalty or the canonical-correlation coefficient in \eqref{eq:gaussian_theta_certificate}; the sweep maps a finite-sample point-estimation trade-off and does not certify any formal IC level.

\begin{figure}[H]
    \centering
    \includegraphics[width=0.98\linewidth]{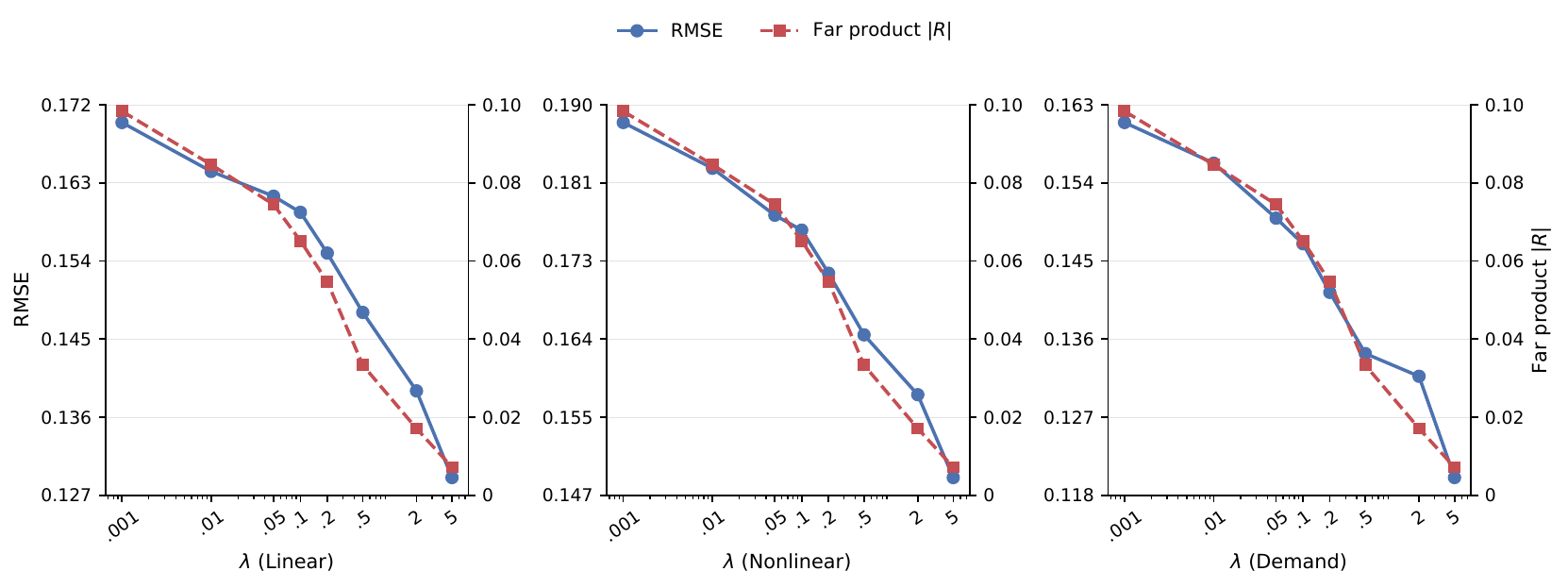}
    \caption{COSTA-LocalPenalty locality-penalty ablation on MovieLens. Each panel is a separate potential outcome model in the $B=8$, top-$K=10$, $p=0.5$ design grid. Circles report RMSE and dashed squares report \texttt{far\_product\_mean\_abs\_R}, the absolute product-cell diagnostic following \eqref{eq:local_penalty_exact}; it is neither the full locality penalty nor the formal block canonical-correlation coefficient. The $\lambda$ axis uses logarithmic spacing over the simulated values.}
    \label{fig:movielens_lambda}
\end{figure}

\subsection{Balance and switching ablation on MovieLens}

Figure~\ref{fig:movielens_component_ablation} is the MovieLens point-estimation counterpart of Figure~\ref{fig:component_ablation}. It isolates the finite-sample point-estimation contribution of the active balance and switching controls; inference is evaluated separately with the heterogeneous-mean and design-aware procedures in Subsection~\ref{subsec:inference_experiments}.

\begin{figure}[H]
    \centering
    \includegraphics[width=0.98\linewidth]{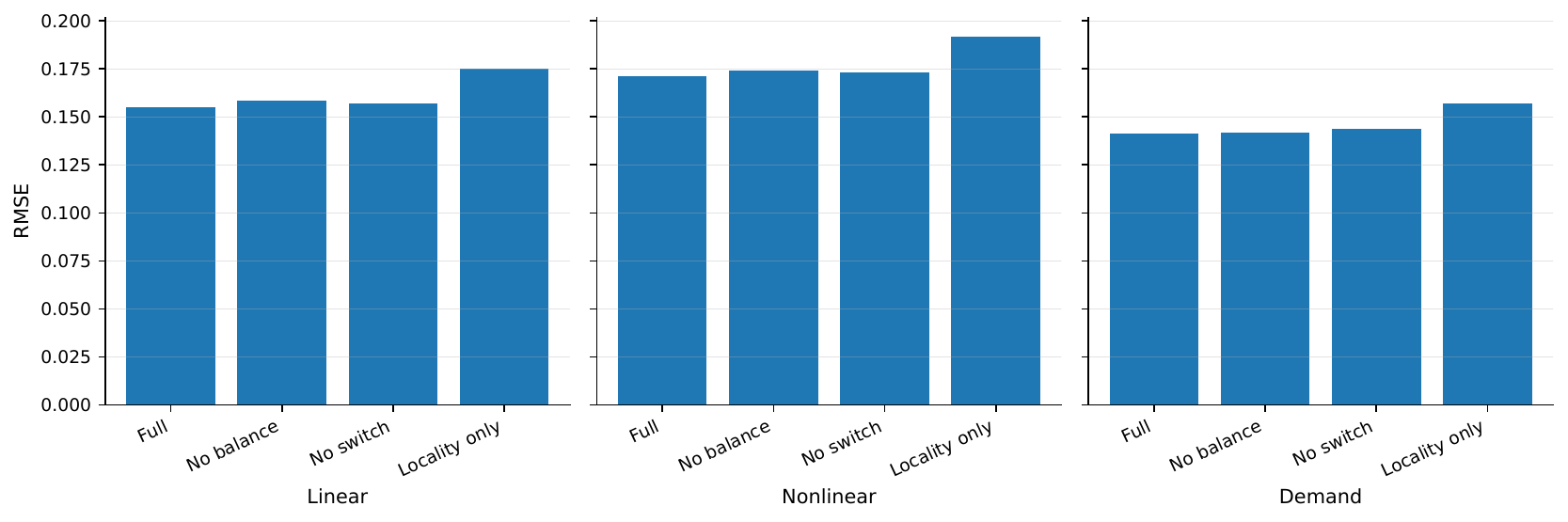}
    \caption{COSTA-LocalPenalty balance and switching penalty ablation on MovieLens. Each panel is a separate potential outcome model; the setting is $N=10{,}000$, $B=8$, top-$K=10$, $p=0.5$, weighted HT, and the design grid previously used for inference diagnostics. Only RMSE is reported here.}
    \label{fig:movielens_component_ablation}
\end{figure}

\subsection{Estimator comparison on MovieLens}

\begin{table}[htbp]
    \centering
    \caption{MovieLens estimator comparison in the default $B=8$ setting. HT and H{\'a}jek are traffic-weighted; DIM is unweighted.}
    \label{tab:estimator_movielens}
    \small
    \input{tables/tab_estimator_movielens.tex}
\end{table}

\subsection{Studentized inference diagnostics}

\IfFileExists{tables/tab_inference_studentized.tex}{Table~\ref{tab:inference_studentized} reports studentized statistics centered at the independently calibrated design expectation for all three potential-outcome models. These diagnostics assess the model-assisted variance workflow in the semi-synthetic setting; they do not turn the sufficient asymptotic conditions of Section~\ref{sec:inference} into an automatic certificate for the fixed-rank implementation.}{The empirical study reports point-estimation evidence only when the supplementary inference artifacts are absent.}

\section{Implementation and Dependence-Diagnostic Tables}
\label{app:implementation_tables}

The tables in this section collect the detailed theory--code map and the diagnostic summaries used by the main text. They are placed in the appendix because they support interpretation and implementation rather than the central argument.

\subsection{Exact implementation details for COSTA-LocalPenalty}
\label{subsec:implementation_details_appendix}

The implementation samples $\mathcal S_I=\{(I_\ell,J_\ell)\}_{\ell=1}^{m_I}$ uniformly with replacement from $\mathcal U_{\ne}$, $\mathcal S_T=\{(B_\ell,R_\ell)\}_{\ell=1}^{m_T}$ uniformly from $\mathcal T_g$, and $\mathcal S_P=\{(I'_\ell,J'_\ell,B'_\ell,R'_\ell)\}_{\ell=1}^{m_P}$ from $\operatorname{Unif}(\mathcal U_{\ne})\otimes\operatorname{Unif}(\mathcal T_g)$. The implementation uses $g=2$, $m_T=B^2$, and candidate count
\[
 m_0=\max\{1000,\min[50000,\max\{1000,10|\mathcal E_I|\}]\},
\]
with self-pairs rejected; $m_I$ and $m_P$ are the retained counts from independent batches of $m_0$ candidates. The unit pairs are uniformly sampled distinct ordered pairs and need not be graph nonedges.

The main-text six-weight display is a formal design objective. The implementation uses normalized network and lag bias weights equal to one, an edge-load/precision weight $0.05$, a sampled block-load balance weight $0.05$, no separate full unit-load quadratic, and switching weight $10$ with target switching rate $0.10$. These settings and $\lambda_{\rm loc}=1$ are fixed before simulated outcomes are generated and reused across datasets and outcome models; the component-removal and locality-sweep experiments are evaluation ablations rather than tuning procedures.

Let $\mathcal E_I$ be the directed network-edge list with nonnegative weights $a_{ij}$ normalized to sum to one. The implemented sparse components are
\begin{align}
 b_{\comp}^{\rm code}
 &=\sum_{(i,j)\in\mathcal E_I}a_{ij}\{\Gp(f_i^\top f_j)-1\},\nonumber\\
 b_{\lag}^{\rm code}
 &=\frac{1}{B-1}\sum_{b=2}^B\{\Gp(t_b^\top t_{b-1})-1\},\nonumber\\
 V_{\rm sur}^{\rm code}
 &=\sum_{(i,j)\in\mathcal E_I}a_{ij}\{\Gp(f_i^\top f_j)+1\}
 +\frac{1}{B-1}\sum_{b=2}^B\{\Gp(t_b^\top t_{b-1})+1\}.
 \label{eq:implemented_sparse_components}
\end{align}
Because both averages are normalized,
\begin{equation}
V_{\rm sur}^{\rm code}=4+b_{\comp}^{\rm code}+b_{\lag}^{\rm code}.
\label{eq:implemented_edge_load_identity}
\end{equation}
Thus this historically named quantity is an exact edge-load/contrast regularizer, not a plug-in approximation to the formal covariance envelope or to the outcome-estimator variance.

The balance proxy reuses the item batch:
\begin{equation}
 J_{\rm bal}=\left(\frac1{m_I}\sum_{\ell=1}^{m_I}R^I_\ell\right)^2
 +\frac1{m_I}\sum_{\ell=1}^{m_I}[R^I_\ell]_+.
 \label{eq:implemented_balance_proxy}
\end{equation}
For fixed factors, each average in \eqref{eq:local_penalty_exact} is unbiased for its corresponding uniform-pair expectation. The batches are nevertheless held fixed while the optimizer adapts to them, so the optimized value is a fixed-Monte-Carlo objective, not an unbiased post-selection estimate of a population dependence bound. The recorded diagnostic
\[
 \texttt{far\_product\_mean\_abs\_R}=m_P^{-1}\sum_{\ell=1}^{m_P}|R^P_\ell|
\]
uses absolute rather than squared covariance and omits the item-only and time-only terms. It therefore differs from $J_{\rm loc}$ and does not estimate or upper-bound $\theta_N^{\rm GC}$.

\begin{table}[H]
\centering
\footnotesize
\caption{Theory--implementation concordance. Formal and code-level precision terms are deliberately separated; neither pairwise regularization nor a finite code objective supplies a formal inference guarantee.}
\label{tab:theory_code_concordance}
\renewcommand{\arraystretch}{1.18}
\setlength{\tabcolsep}{3.2pt}
\begin{tabular}{@{}>{\raggedright\arraybackslash}p{0.15\linewidth}>{\raggedright\arraybackslash}p{0.22\linewidth}>{\raggedright\arraybackslash}p{0.25\linewidth}>{\raggedright\arraybackslash}p{0.29\linewidth}@{}}
\toprule
Object & Formal theory & Implemented approximation & Status and empirical role \\
\midrule
Estimand weights & Target-weighted $w_a$ & Design uses $w_a=1$; headline estimator is traffic-weighted & Coherent robustness comparison; a target-weighted design is available but not the default \\
Bias criterion & Full exposure-weighted cut terms & Normalized network-edge and lag averages in \eqref{eq:implemented_sparse_components} & Theorem-aligned under the linear working model; primary finite-sample MSE driver \\
Precision control & Formal covariance-envelope quadratic in Corollary~\ref{cor:weighted_variance_envelope} & Average normalized exposure-pair treatment-load variance $V_{\rm sur}^{\rm code}=4+b_{\comp}^{\rm code}+b_{\lag}^{\rm code}$ & Exact code-level edge-contrast regularizer; not a theorem-equivalent approximation to the formal envelope or estimator variance \\
Balance & Formal block- and unit-load quadratic forms & $J_{\rm bal}$ based on sampled same-block distinct-unit covariances, plus a separate adjacent-time switching guardrail & The balance term is a block-load proxy; no code term equals the full unit-load quadratic \\
Locality & Separated-set coefficient $\theta_N^{\rm GC}$ or a uniform analytic bound & Sampled pairwise $J_{\rm loc}$ in \eqref{eq:local_penalty_exact} & Pairwise regularizer only; not a certificate or calibrated estimator of $\theta_N^{\rm GC}$ \\
Locality controls & No theorem parameter corresponds to these tuning inputs & \texttt{lambda\_local}$=\lambda_{\rm loc}$; \texttt{n\_far\_pairs} determines $m_0$; \texttt{locality\_time\_gap}$=g$ & Optimization controls for the fixed Monte Carlo surrogate \\
Recorded locality diagnostic & No formal certificate & \texttt{far\_product\_mean\_abs\_R}$=m_P^{-1}\sum_\ell|R^P_\ell|$ & Descriptive product-cell diagnostic; not $J_{\rm loc}$ and not $\theta_N^{\rm GC}$ \\
Inference & IC-0, IC-1, and IC-2 conditions & The supplementary workflow evaluates heterogeneous-mean local HAC and design-aware tail correction & Calibration is tested in theory-aligned settings; asymptotic validity follows from the displayed structural and rate conditions \\
\bottomrule
\end{tabular}
\end{table}

\begin{table}[H]
\centering
\footnotesize
\caption{How outcome architecture maps to the fixed contribution radius. The radius statement concerns measurability for the dependence certificate, not exactness of the cut-bias identity.}
\label{tab:model_radius_mapping}
\renewcommand{\arraystretch}{1.18}
\setlength{\tabcolsep}{3.2pt}
\begin{tabular}{@{}>{\raggedright\arraybackslash}p{0.27\linewidth}>{\raggedright\arraybackslash}p{0.18\linewidth}>{\raggedright\arraybackslash}p{0.47\linewidth}@{}}
\toprule
Outcome architecture & Sufficient product-graph radius & Interpretation \\
\midrule
Exact one-hop network plus one-lag linear model & $r_Y=1$ & Local contribution and exact linear cut identity both apply \\
Nonlinear response of the same one-hop and one-lag exposures & $r_Y=1$ & Locality remains finite, but the exact linear bias identity is replaced by the misspecification decomposition \\
$k$-hop spillover and $h$-lag carryover & $r_Y\le k+h$ & A conservative path-length bound when the inferential graph contains the relevant unit and time edges \\
Global normalization, reallocation, or ranking & No fixed $r_Y$ in general & The local certificate may fail; aggregate conditional-$L^2$ approximation is a separate sufficient route \\
\bottomrule
\end{tabular}
\end{table}

\begin{table}[H]
\centering
\footnotesize
\caption{Proof objects and empirical diagnostics for the local graph-$\psi$ route. Verdicts refer only to this sufficient route.}
\label{tab:certificate_diagnostics}
\renewcommand{\arraystretch}{1.18}
\setlength{\tabcolsep}{3.2pt}
\begin{tabular}{>{\raggedright\arraybackslash}p{0.24\linewidth}>{\raggedright\arraybackslash}p{0.36\linewidth}>{\raggedright\arraybackslash}p{0.30\linewidth}}
\toprule
Object & Logical role & Permissible verdict \\
\midrule
Exact supremum $\theta_N^{\rm GC}(s)$ & Proof-level supremum over all separated nonempty set pairs & \texttt{CERTIFIED} only after its full decay/topology rates are verified \\
Far-row mass $\alpha_N/\kappa_N$ or exact finite range & Uniform analytic upper bound or structural primitive & \texttt{CERTIFIED} when the resulting KMS rates hold \\
Sampled separated-block canonical correlations & Lower-bound stress tests for the exact supremum & A large value can make the route \texttt{ROUTE REFUTED}; small sampled values imply only \texttt{NOT CERTIFIED} \\
Pairwise covariance, effective rank, eigenvalue share & Obstruction and architecture diagnostics & May support \texttt{ROUTE REFUTED} or \texttt{NOT CERTIFIED}, never \texttt{CERTIFIED} by themselves \\
\bottomrule
\end{tabular}
\end{table}

\section{Design Diagnostics and Implementation Checks}
\label{app:design_diagnostics}

Before using an optimized assignment distribution, we recommend reporting Monte Carlo diagnostics for the design target, operational guardrails, and inference conditions:
\begin{enumerate}[leftmargin=2em]
    \item treatment share $N^{-1}\sum_i Z_{ib}$ by block;
    \item treatment exposure $B^{-1}\sum_b Z_{ib}$ by unit;
    \item average network-spillover covariance $\sum_{i,j,b}A_{ij}\Rnorm_{(i,b),(j,b)}/\sum_{i,j,b}A_{ij}$;
    \item average lag covariance $\bar{\Rnorm}_{\lag}$ and realized switching rate;
    \item objective value under empirical covariance estimated from sampled assignments;
    \item inference-locality diagnostics: average shell growth $\delta_N^{\partial}(s;k)$, the mixed shell quantity $c_N(s,m;k)$, effective rank $r_{\rm eff}(\Omega_N)$, far-correlation quantiles, and empirical canonical correlations between separated graph blocks;
    \item Kronecker approximation diagnostics, such as the normalized residual
    \[
        \frac{\|R_{\mathrm{target}}-\Pi_{\mathrm{Kron}}(R_{\mathrm{target}})\|_F}
        {\|R_{\mathrm{target}}\|_F}
    \]
    for the exposure-weighted target covariance proxy;
    \item comparisons with independent unit--time Bernoulli, hard-balanced switchback, graph-cluster randomization, and static unit-level OCD with no temporal adaptation; complete randomization can be added as an optional operational diagnostic.
\end{enumerate}

Certificate verdicts must follow Table~\ref{tab:certificate_diagnostics}. An analytic uniform upper bound or an exact finite-range primitive can support \texttt{CERTIFIED} once the full KMS rates are checked. A sampled separated-block canonical correlation is only a lower bound on the exact supremum: a large sampled value can make the proposed route \texttt{ROUTE REFUTED}, whereas small sampled values leave the design \texttt{NOT CERTIFIED}. Pairwise far-covariance summaries and effective rank are architecture diagnostics and cannot establish \texttt{CERTIFIED}.

\section{Notation Reference}
\label{app:notation_reference}
\label{app:notation}

This appendix collects the notation used throughout the paper. The main text introduces each symbol at its first substantive use; the tables here are only a reference for readers who want to check notation later.

\begin{table}[H]
\centering
\small
\begin{tabular}{ll}
\toprule
Symbol & Meaning \\
\midrule
$N$ & number of eligible units \\
$B$ & number of temporal blocks \\
$\Zcal$ & treatment-assignment cell index set, $[N]\times[B]$ \\
$\Ocal$ & evaluated post-burn-in outcome cells, $[N]\times\{2,\ldots,B\}$ \\
$L$ & number of assignment cells, $NB$ \\
$M$ & number of evaluated outcome cells, $N(B-1)$ \\
$a=(i,b),\ c=(j,r)$ & generic unit--block cells \\
$Z_{ib}$ & treatment of unit $i$ in block $b$ \\
$\mathcal U_b$ & user/session/request opportunities observed in block $b$ \\
$S_{uib}$ & pre-treatment exposure/opportunity weight for user $u$ and unit $i$ in block $b$ \\
$m_{ib}$ & unit-block traffic weight, $\sum_{u\in\mathcal U_b}S_{uib}$ \\
$Q_{uib}$ & user-level response contribution before aggregation \\
$p,q,v$ & treatment probability, control probability $1-p$, and $pq$ \\
$\AI=(A_{ij})$ & network spillover graph or weighted adjacency matrix \\
$H^{\comp}$ & contemporaneous network-spillover exposure matrix \\
$H^{\lag}$ & one-lag own-unit temporal exposure matrix \\
$H$ & generic nonnegative interference matrix \\
$\mu_a$ & structural baseline potential-outcome component \\
$\tilde\alpha_a$ & fixed baseline adjustment used by an estimator \\
$w_a, W_N$ & target weight for cell $a$ and total weight $\sum_{a\in\Ocal}w_a$ \\
$\Sigma$ & raw treatment covariance, $\Cov(Z)$ \\
$\Rnorm$ & normalized treatment covariance, $\Sigma/(pq)$ \\
$\psi_a$ & common-$p$ HT score, $Z_a/p-(1-Z_a)/(1-p)$ \\
$\deg_H$ & exposure-degree vector, $H^\top\one$ \\
$\dlat$ & latent dimension in the full Gaussian-copula parameterization \\
$\dI,\dT$ & unit and temporal-block latent dimensions in the Kronecker parameterization \\
\bottomrule
\end{tabular}
\caption{Core design, outcome, and covariance notation. The control probability is $q=1-p$. Latent dimensions are denoted by $\dlat,\dI,\dT$ and are never denoted by $q$.}
\label{tab:notation}
\end{table}

\begin{table}[H]
\centering
\begingroup
\small
\renewcommand{\arraystretch}{1.32}
\setlength{\extrarowheight}{1.2pt}
\setlength{\tabcolsep}{5.5pt}
\begin{tabular}{@{}>{\raggedright\arraybackslash}p{0.29\linewidth}>{\raggedright\arraybackslash}p{0.63\linewidth}@{}}
\toprule
Symbol & Meaning in the inference theory \\
\midrule
$M_N,L_N$ & numbers of evaluated outcome cells and assignment cells in the triangular array \\
$\mathbb P_N,\mathbb E_N,\Cov_N,\Var_N$ & conditional randomization law and operators given $\mathcal F_N$ \\
$E_N^{\rm good}$ & high-probability pre-treatment environment event on which conditional assumptions hold \\
$\omega_{Na}$ & stabilized target weight $M_Nw_{Na}/W_N$, with $M_N^{-1}\sum_a\omega_{Na}=1$ \\
$Q_{Na},\bar q_{Na},\xi_{Na}$ & weighted HT contribution, its heterogeneous design mean, and centered contribution \\
$T_N,\sigma_N^2,V_N$ & centered sum, its variance, and estimator variance $\sigma_N^2/M_N^2$ \\
$\Sigma_N$ & long-run variance $\sigma_N^2/M_N=M_NV_N$ \\
$\theta_N(s)$ & graph-$\psi$ dependence coefficient at graph separation $s$ \\
$\delta_N^{\partial},\Delta_N,c_N$ & average shell-growth and mixed topology quantities entering the CLT/HAC rates \\
$\rho_{\Omega_N}(U,V)$ & largest Gaussian canonical correlation between two assignment subvectors \\
$\alpha_N(t),\kappa_N$ & absolute far-row correlation mass and a spectral lower bound for $\Omega_N$ \\
$r_{\rm eff}(\Omega_N)$ & effective rank $\{\operatorname{tr}(\Omega_N)\}^2/\|\Omega_N\|_F^2$ \\
$r_Y,r_Z,h_*$ & finite outcome radius, finite assignment-dependence radius, and $h_*=2r_Y+r_Z$ in the dependency-graph specialization only \\
$\mathcal B_{\mathrm{bias},N}^{\star}$ & effect-scale sensitivity envelope for design bias and model misspecification \\
$\widehat\Sigma_N^{\rm or},\widehat\Sigma_N^{\rm MA},\widehat\Sigma_N^{\rm UC}$ & oracle-centered, model-assisted-centered, and PSD uncentered long-run variance estimators \\
$\mathfrak B_{\infty,N}^{\rm abs},\mathfrak S_{\infty,N}(b_N)$ & disconnected-graph cross-component covariance-target and pair-product stochastic remainders \\
$\mathsf P_N$ & external-pilot feature/design matrix for the heterogeneous contribution mean \\
$\mathcal E_N,\mathcal F_N$ & pre-pilot environment field and pilot-augmented field containing the pilot-derived design \\
$\mathbb P_N^D,\mathbb P_N^J$ & current-randomization law given $\mathcal F_N$ and joint pilot--experiment law given $\mathcal E_N$ \\
$V_N^D,V_{\delta,N}^J$ & current-randomization variance and projected pilot variance under their explicitly marked laws \\
$\widehat V_N^{\rm MA,+},\widehat V_{\rm DA}^{+}$ & nonnegative estimator-scale variances obtained by outer positive-part truncation \\
$\widehat\Sigma_{\rm DA}^{\rm sgn},\widehat\Sigma_{\rm DA,raw}^{+}$ & signed and positive-part model-assisted covariance-tail corrections \\
\bottomrule
\end{tabular}
\endgroup
\caption{Additional notation used in the inference section. The bias envelope is on the same effect scale as $\tau_{w,N}$ and can therefore be added directly to a confidence-interval half-width.}
\label{tab:inference_notation}
\end{table}

\begin{table}[H]
\centering
\small
\begin{tabular}{@{}>{\raggedright\arraybackslash}p{0.25\linewidth}>{\raggedright\arraybackslash}p{0.23\linewidth}>{\raggedright\arraybackslash}p{0.44\linewidth}@{}}
\toprule
Code/config name & Mathematical object & Role \\
\midrule
\texttt{lambda\_bias\_comp} & $\lambda_{\mathrm{bias},\comp}$ & network-bias penalty \\
\texttt{lambda\_bias\_lag} & $\lambda_{\mathrm{bias},\lag}$ & lag-bias penalty \\
\texttt{lambda\_var} (historical key) & $\lambda_{\rm edge}$ & code-level edge-load/contrast regularizer weight; no one-to-one mapping to formal $\lambda_V$ \\
\texttt{lambda\_bal} & approximation to $\lambda_{\mathrm{blk}},\lambda_{\mathrm{item}}$ & sampled soft-balance proxy \\
\texttt{switch\_penalty} & $\lambda_{\mathrm{sw}}$ & switching guardrail \\
\texttt{lambda\_local} & $\lambda_{\mathrm{loc}}$ & COSTA-LocalPenalty locality penalty \\
\bottomrule
\end{tabular}
\caption{Mapping between implementation names and mathematical tuning parameters. The formal covariance-envelope weight $\lambda_V$ has no one-to-one key in the implemented sparse objective.}
\label{tab:code_math_mapping}
\end{table}

\section{Proofs for the Inference Theory}
\label{app:inference_proofs}

This appendix records the arguments for Section~\ref{sec:inference}. Except in the explicitly marked proof of Proposition~\ref{prop:external_pilot_uncertainty}, probability statements are understood under the regular conditional law given $\mathcal F_N$. Thus the graph, target weights, structural potential-outcome schedule, design parameters, and any disturbances treated as fixed under randomization are nonrandom inside each ordinary conditional proof. Proposition~\ref{prop:external_pilot_uncertainty} instead uses the joint law $\mathbb P_N^J(\cdot)=\Pr(\cdot\mid\mathcal E_N)$. When the conditioning fields are random, displayed bounds and limits are required to hold under their explicitly stated outer or joint probability law.

\begin{proof}[Proof of Lemma~\ref{lem:conditional_unconditional_transfer}]
For every $t\in\mathbb R$, the tower property gives
\[
\left|\Pr(S_N\le t)-\Phi(t)\right|
=\left|\E\{\Pr(S_N\le t\mid\mathcal F_N)-\Phi(t)\}\right|
\le\E\Delta_N(\mathcal F_N).
\]
Because $0\le\Delta_N(\mathcal F_N)\le1$ and it converges to zero in probability, bounded convergence in probability implies
$\E\Delta_N(\mathcal F_N)\to0$. Taking the supremum over $t$ proves \eqref{eq:conditional_unconditional_transfer}.
\end{proof}

\subsection{Weak-dependence CLT and transparent special cases}\label{subsec:clt_proofs}

\begin{proof}[Proof of Theorem~\ref{thm:cell_clt}]
By construction, $\mathbb E_N\xi_{Na}=0$ and $T_N=\sum_a\xi_{Na}$. On every good-environment realization, Assumption~\ref{ass:contribution_metric_compatibility} makes $\mathcal G_N^{\Ocal}$ a connected graph indexed exactly by the contribution cells and makes $d_N$ its shortest-path metric. Definition~\ref{def:graph_psi_dependence} is a scalar specialization of the graph-$\psi$ dependence condition in \citet{kojevnikov2021network}: its conditional covariance functional is bounded by a constant times
$|A||B|(\|f\|_\infty+\operatorname{Lip}f)(\|g\|_\infty+\operatorname{Lip}g)$. Assumption~\ref{ass:weights_moments} supplies a uniform conditional $L^\nu$ moment for some $\nu>4$, and \eqref{eq:weak_dep_clt_rate_1}--\eqref{eq:weak_dep_clt_rate_2} are the corresponding variance-normalized topology--dependence restrictions.

Let
\[
\Delta_N^{\rm CLT}:=
\sup_{t\in\mathbb R}
\left|\mathbb P_N\left(\frac{T_N}{\sigma_N}\le t\right)-\Phi(t)\right|.
\]
To show $\Delta_N^{\rm CLT}\to_p0$, take an arbitrary subsequence. The good-environment probability and all $\mathcal F_N$-measurable rate conditions converge in probability, so there is a further subsequence along which the good events occur eventually and the displayed rate expressions converge almost surely. For almost every such environment realization, the deterministic KMS theorem applies to the conditional triangular array and gives $\Delta_N^{\rm CLT}\to0$. The subsequence principle therefore yields $\Delta_N^{\rm CLT}\to_p0$ on the original sequence, which is \eqref{eq:clt_centered}. The identity follows from
$\widehat\tau^{HT}_{w,\tilde\alpha}-\mathbb E_N\widehat\tau^{HT}_{w,\tilde\alpha}=T_N/M_N$ and $V_N=\sigma_N^2/M_N^2$.

For causal centering, decompose
\[
    \frac{\widehat\tau^{HT}_{w,\tilde\alpha}-\tau_{w,N}}{\sqrt{V_N}}
    =\frac{T_N}{\sigma_N}
    +\frac{\mathbb E_N\widehat\tau^{HT}_{w,\tilde\alpha}-\tau_{w,N}}{\sqrt{V_N}}.
\]
The second term converges to zero under the outer environment law by \eqref{eq:bias_negligible}. Conditional Slutsky, followed by the same subsequence argument, proves \eqref{eq:clt_tau}.
\end{proof}

\begin{proof}[Proof of Corollary~\ref{cor:poly_geometric_clt}]
The polynomial ball bound implies, uniformly in $s,m$ and for any $\alpha>1$,
\[
    \Delta_N(s,m;k\alpha)^{1/\alpha}
    \le C(1+m)^{dk},
\]
and
\[
    \left\{\delta_N^\partial
    \left(s;\frac{\alpha}{\alpha-1}\right)\right\}^{1-1/\alpha}
    \le C(1+s)^d.
\]
Hence
\begin{equation}
    c_N(s,m;k)\le C(1+m)^{dk}(1+s)^d.
    \label{eq:proof_c_polynomial_bound}
\end{equation}
Because every positive power of $\theta_N(s)$ is bounded by a geometrically decreasing sequence, the sum over $s$ in \eqref{eq:weak_dep_clt_rate_1} is $O\{(1+m_N)^{dk}\}$. Under $\sigma_N^2\asymp M_N$, the prefactor is $O(M_N^{-k/2})$. With $m_N=\lceil C_m\log M_N\rceil$, the two expressions for $k=1,2$ are therefore
\[
    O\{M_N^{-1/2}(\log M_N)^d\}
    \quad\text{and}\quad
    O\{M_N^{-1}(\log M_N)^{2d}\},
\]
which vanish. The second weak-dependence rate is bounded by
\[
    C M_N^{3/2}\rho^{m_N(1-1/\nu)}.
\]
It converges to zero whenever
$C_m>3/[2(1-1/\nu)|\log\rho|]$. This proves the claim.
\end{proof}

\begin{proof}[Proof of Corollary~\ref{cor:finite_range_clt}]
Under Assumption~\ref{ass:local_outcomes}, conditional on $\mathcal F_N$, $\xi_{Na}$ is measurable with respect to assignments in $\mathcal B_N^{\Zcal}(a;r_Y)$. If $d_N(a,c)>2r_Y+r_Z$, the two assignment neighborhoods are more than $r_Z$ apart, so the corresponding contributions are conditionally independent. Connecting pairs within $h_*=2r_Y+r_Z$ therefore gives a dependency graph whose maximum degree $D_N$ is at most $\Delta_N(h_*)-1$. If each contribution also uses an independent cellwise error $\varepsilon_{Na}$, disjoint vertex sets carry disjoint independent error collections, so the same conditional dependency graph remains valid.

Apply Theorem~1 of \citet{baldirinott1989}. Let $L_{Na}^{(k)}$ be the number of connected vertex subsets of cardinality at most $k$ that contain $a$, and define
\[
    A_{k,N}:=\sigma_N^{-k}\sum_{a\in\Ocal_N}
    L_{Na}^{(k)}\E|\xi_{Na}|^k,\qquad k=3,4.
\]
For fixed $k$, a graph of maximum degree $D_N$ satisfies
$L_{Na}^{(k)}\le C_k(D_N+1)^{k-1}$. Uniform fourth moments and $\sigma_N^2\asymp M_N$ therefore give
\[
    A_{3,N}=O\{(D_N+1)^2M_N^{-1/2}\},
    \qquad
    A_{4,N}=O\{(D_N+1)^3M_N^{-1}\}.
\]
The assumed rate $D_N=o(M_N^{1/4})$ makes both quantities vanish. The Baldi--Rinott Kolmogorov bound, which is of order
$A_{3,N}^{1/2}+A_{4,N}^{1/2}$, then yields \eqref{eq:clt_centered}.
\end{proof}

\begin{proof}[Proof of Proposition~\ref{prop:local_approximation_transfer}]
Let $R_N^{\rm app}=T_N-T_N^{(\ell_N)}$. Conditional centering gives
$\sigma_N=\|T_N\|_{2,N}$ and $\sigma_{N,\ell}=\|T_N^{(\ell_N)}\|_{2,N}$. The conditional $L^2$ reverse triangle inequality and \eqref{eq:aggregate_local_approximation} imply
$|\sigma_{N,\ell}/\sigma_N-1|\le\|R_N^{\rm app}\|_{2,N}/\sigma_N\xrightarrow{p}0$.
For every $\epsilon>0$, conditional Markov's inequality gives
$\mathbb P_N(|R_N^{\rm app}|/\sigma_N>\epsilon)
\le\epsilon^{-2}\|R_N^{\rm app}\|_{2,N}^2/\sigma_N^2\xrightarrow{p}0$;
hence $R_N^{\rm app}/\sigma_N\xrightarrow{p\mid\mathcal F}0$. Finally,
\[
    \frac{T_N}{\sigma_N}
    =\frac{T_N^{(\ell_N)}}{\sigma_{N,\ell}}
      \frac{\sigma_{N,\ell}}{\sigma_N}
      +\frac{R_N^{\rm app}}{\sigma_N},
\]
and conditional Slutsky proves \eqref{eq:clt_centered}.
\end{proof}

\subsection{Gaussian certificate and rank barrier}\label{subsec:gaussian_certificate_proofs}

\begin{proof}[Proof of Proposition~\ref{prop:gaussian_canonical_certificate}]
Fix separated nonempty $A,B\subseteq\Ocal_N$ and work conditional on $\mathcal F_N$. Under Assumption~\ref{ass:local_outcomes}, the vector $\xi_{N,A}$ is a measurable function of $G_{N,A^{+r_Y}}$, because assignments are coordinatewise threshold functions of $G_N$ and every other argument of the contribution is $\mathcal F_N$-measurable. Similarly, $\xi_{N,B}$ is measurable with respect to $G_{N,B^{+r_Y}}$.

For bounded Lipschitz $f$ and $g$, set
$F=f(\xi_{N,A})-\E f(\xi_{N,A})$ and
$H=g(\xi_{N,B})-\E g(\xi_{N,B})$. The multidimensional Gebelein inequality gives
\[
    |\E(FH)|
    \le \rho_{\Omega_N}(A^{+r_Y},B^{+r_Y})
    \{\E F^2\}^{1/2}\{\E H^2\}^{1/2}.
\]
The statement remains valid for singular Gaussian subvectors by restricting to their covariance ranges, or equivalently by an arbitrarily small nonsingular Gaussian perturbation followed by an $L^2$ limit. Since
$\Var\{f(\xi_{N,A})\}\le\|f\|_\infty^2$ and likewise for $g$,
\[
    |\Cov\{f(\xi_{N,A}),g(\xi_{N,B})\}|
    \le \rho_{\Omega_N}(A^{+r_Y},B^{+r_Y})
       \|f\|_\infty\|g\|_\infty.
\]
This is stronger than \eqref{eq:graph_psi_definition}, because $|A||B|\ge1$ and
$\|f\|_\infty\le\|f\|_\infty+\operatorname{Lip}(f)$. Taking the supremum over separated sets proves \eqref{eq:gaussian_theta_certificate}. Theorem~\ref{thm:cell_clt} then applies whenever its rate conditions hold.
\end{proof}

\begin{proof}[Proof of Corollary~\ref{cor:gaussian_independent_noise}]
Fix separated sets $A,B$ and bounded Lipschitz functions $f,g$. Write
$F=f(Q_{N,A})-\E f(Q_{N,A})$ and
$H=g(Q_{N,B})-\E g(Q_{N,B})$. Conditional on the two relevant Gaussian subvectors, the cellwise noise blocks indexed by $A$ and $B$ are independent of one another and of $G_N$. Hence, with
$\bar F=\E(F\mid G_{N,A^{+r_Y}})$ and
$\bar H=\E(H\mid G_{N,B^{+r_Y}})$,
\[
    \E(FH\mid G_{N,A^{+r_Y}},G_{N,B^{+r_Y}})=\bar F\bar H,
    \qquad
    \Cov(F,H)=\Cov(\bar F,\bar H).
\]
Conditional expectation is an $L^2$ contraction, so
$\Var(\bar F)\le\Var(F)$ and
$\Var(\bar H)\le\Var(H)$. Applying the Gaussian Gebelein inequality to $\bar F$ and $\bar H$ gives the same canonical-correlation bound as in the proof of Proposition~\ref{prop:gaussian_canonical_certificate}. The remaining graph-$\psi$ comparison is unchanged.
\end{proof}

\begin{proof}[Proof of Proposition~\ref{prop:spectral_rowsum_certificate}]
Fix $s>2r_Y$ and separated nonempty sets $A,B\subseteq\Ocal_N$ with
$d_N(A,B)\ge s$. Put $U=A^{+r_Y}$, $V=B^{+r_Y}$, and
$t=s-2r_Y\ge1$. Geodesic closure and the full-graph triangle inequality give $d_N^{\Zcal}(U,V)\ge t$, so $U$ and $V$ are disjoint. The spectral-floor condition passes to principal submatrices:
$\lambda_{\min}(\Omega_{N,UU})\wedge
\lambda_{\min}(\Omega_{N,VV})\ge\kappa_N$. Hence
\[
\begin{split}
    \rho_{\Omega_N}(U,V)
    &=\|\Omega_{N,UU}^{-1/2}\Omega_{N,UV}
       \Omega_{N,VV}^{-1/2}\|_{\rm op}\\
    &\le \kappa_N^{-1}\|\Omega_{N,UV}\|_{\rm op}\\
    &\le \kappa_N^{-1}
       \{\|\Omega_{N,UV}\|_1
          \|\Omega_{N,UV}\|_\infty\}^{1/2}\\
    &\le \kappa_N^{-1}\alpha_N(t).
\end{split}
\]
The last inequality follows because every row and column of the cross-covariance block is a sub-sum of correlations at graph distance at least $t$. Canonical correlation is at most one, so taking the supremum over $A,B$ proves \eqref{eq:spectral_rowsum_bound}. The geometric-decay conclusion follows immediately from Proposition~\ref{prop:gaussian_canonical_certificate} and Corollary~\ref{cor:poly_geometric_clt}.
\end{proof}

\begin{proof}[Proof of Proposition~\ref{prop:nugget_bias_floor}]
For $a\ne c$, the identity component has zero off-diagonal entry, so
\[
    \Omega_{N,ac}=(1-\lambda_N)\Omega^{\rm str}_{N,ac}
    \le1-\lambda_N.
\]
The Gaussian-copula map $\mathcal G_p$ is increasing, which proves
\eqref{eq:nugget_edge_correlation_cap}. The weighted positive-exposure bias is nonpositive by Proposition~\ref{prop:expected_cut_geometry}; hence its absolute value is
\[
    \frac{\vartheta_N}{W_N}\sum_{a,c}w_{Na}H_{N,ac}
    (1-\Rnorm_{N,ac}).
\]
Keeping only off-diagonal terms and using the correlation cap proves
\eqref{eq:nugget_bias_floor}.

For the boundary expansion, put $t_p=\Phi^{-1}(p)$ and $u=1-\rho$. By \eqref{eq:Gp_derivative},
\[
\begin{split}
    1-\mathcal G_p(1-\lambda)
    &=\frac1{p(1-p)}\int_{1-\lambda}^1
      \varphi_2(t_p,t_p;\rho)\,d\rho\\
    &=\frac1{p(1-p)}\int_0^\lambda
      \frac{\exp\{-t_p^2/(2-u)\}}
      {2\pi\sqrt{2u-u^2}}\,du.
\end{split}
\]
Uniformly for $p$ in any compact overlap set,
\[
    \frac{\exp\{-t_p^2/(2-u)\}}
      {2\pi\sqrt{2u-u^2}}
    =\frac{\varphi(t_p)}{2\sqrt{\pi u}}\{1+o(1)\}
    \qquad (u\downarrow0).
\]
The ratio is uniformly bounded near zero on that set, so integration gives
\[
    1-\mathcal G_p(1-\lambda)
    =\frac{\varphi(t_p)}{p(1-p)\sqrt\pi}\sqrt\lambda
      \{1+o(1)\},
\]
uniformly over the overlap set. If normalized off-diagonal exposure mass and $\vartheta_N$ are bounded below, the bias floor is therefore of order at least $\sqrt{\lambda_N}$. Under a root-$M_N$ standard-error scale this can be $o_p(M_N^{-1/2})$ only if $\lambda_N=o_p(M_N^{-1})$ under the outer environment law.
\end{proof}

\begin{proof}[Proof of Corollary~\ref{cor:nugget_supported_certificate_frontier}]
Proposition~\ref{prop:nugget_bias_floor} gives $\lambda_N=o_p(M_N^{-1})$ under raw-GATE IC-1 and the stated root-$M_N$ exposure conditions. Because $\kappa_N\asymp\lambda_N$, condition \eqref{eq:nugget_supported_rowsum_rate} implies
$\alpha_N(s_N-2r_Y)=o_p(\lambda_N)$. Combining the two relations proves \eqref{eq:nugget_supported_joint_frontier}. If the structured covariance has its own nonvanishing spectral floor, $\kappa_N$ need not be of order $\lambda_N$, which is why the conclusion is explicitly limited to nugget-supported certification.
\end{proof}

\begin{proof}[Proof of Proposition~\ref{prop:global_factor_certifiability}(i)]
Because $\Omega_N$ is a correlation matrix, $\tr(\Omega_N)=L_N$ and
\[
    \|\Omega_N\|_F^2=\frac{L_N^2}{r_{\rm eff}(\Omega_N)}.
\]
There are at most $L_ND_N^{\Zcal}(h)$ ordered pairs $(a,c)$ with $d_N(a,c)\le h$, and every squared correlation is at most one. Therefore
\[
\begin{split}
    \sum_{d_N^{\Zcal}(a,c)>h}\Omega_{N,ac}^2
    &=\|\Omega_N\|_F^2-
      \sum_{d_N^{\Zcal}(a,c)\le h}\Omega_{N,ac}^2\\
    &\ge \frac{L_N^2}{r_{\rm eff}(\Omega_N)}-L_ND_N^{\Zcal}(h).
\end{split}
\]
Let $r=\operatorname{rank}(\Omega_N)$ and let
$\lambda_1,\ldots,\lambda_r$ be its nonzero eigenvalues. Cauchy--Schwarz gives
\[
    \left(\sum_{j=1}^r\lambda_j\right)^2
    \le r\sum_{j=1}^r\lambda_j^2,
\]
and hence $r_{\rm eff}(\Omega_N)\le r$. This proves both inequalities in \eqref{eq:rank_barrier_bound}.

If $\epsilon_N^{\rm far}=\max_{d_N^{\Zcal}(a,c)>h_N}|\Omega_{N,ac}|$, the left-hand side is at most $L_N^2(\epsilon_N^{\rm far})^2$. Consequently,
\[
    \frac{1}{r_{\rm eff}(\Omega_N)}
    \le \frac{D_N^{\Zcal}(h_N)}{L_N}+(\epsilon_N^{\rm far})^2.
\]
The right-hand side converges to zero under the stated locality conditions, proving that effective rank diverges. Finally,
$\operatorname{rank}(\Omega_T\otimes\Omega_I)
=\operatorname{rank}(\Omega_T)\operatorname{rank}(\Omega_I)
\le d_{T,N}d_{I,N}$.
\end{proof}

\begin{proof}[Proof of Proposition~\ref{prop:global_factor_certifiability}(ii)]
Write $G_N=U_N\zeta_N$ with $\zeta_N\sim N(0,I_{r_N})$. Full column rank of $U_{N,A}$ implies that
\[
    \zeta_N=(U_{N,A}^\top U_{N,A})^{-1}U_{N,A}^\top G_{N,A}
\]
is measurable with respect to $G_{N,A}$. The same is true using $B$. Hence any nondegenerate scalar $v^\top\zeta_N$ is simultaneously a square-integrable function of $G_{N,A}$ and of $G_{N,B}$. Taking these identical functions in the definition of Gaussian maximal correlation yields correlation one. Since canonical correlation is at most one, \eqref{eq:shared_factor_saturation} follows.
\end{proof}

\begin{proof}[Proof of Proposition~\ref{prop:global_factor_certifiability}(iii)]
Fix a unit vector $v\in\mathbb R^{r_N}$ and write
$G_A=U_{N,A}^\top U_{N,A}$ and $G_B=U_{N,B}^\top U_{N,B}$. Choose
\[
    x=U_{N,A}G_A^{-1}v,
    \qquad
    y=U_{N,B}G_B^{-1}v.
\]
Then $U_{N,A}^\top x=v$ and $U_{N,B}^\top y=v$. Under
\eqref{eq:weak_common_factor_model}, the cross covariance of
$x^\top G_{N,A}$ and $y^\top G_{N,B}$ is $a_N$, while
\[
\begin{split}
    x^\top\Omega_{N,AA}x
    &=a_N+(1-a_N)v^\top G_A^{-1}v
      \le a_N+\frac{1-a_N}{\lambda_{\min}(G_A)},\\
    y^\top\Omega_{N,BB}y
    &\le a_N+\frac{1-a_N}{\lambda_{\min}(G_B)}.
\end{split}
\]
Substituting this admissible pair into the supremum defining canonical correlation proves
\eqref{eq:weak_common_factor_lower_bound}. If both minimum Gram eigenvalues are proportional to their block sizes and $a_N\min(|A|,|B|)\to\infty$, each denominator is $a_N^{1/2}\{1+o(1)\}$, so the lower bound tends to one.
\end{proof}

\begin{proof}[Proof of Corollary~\ref{cor:tight_frame_block_saturation}]
Choose $L_N$, $s_N$, and two positive integers $m_{A,N},m_{B,N}$ divisible by $r$ such that
\[
L_N=m_{A,N}+m_{B,N}+s_N-1,\qquad
m_{A,N}\asymp m_{B,N}\asymp L_N,
\qquad s_N\to\infty,\quad s_N=o(L_N),
\]
and, after passing to a harmless subsequence, let $s_N-1$ also be divisible by $r$. Construct two connected communities on vertex sets $A_N$ and $B_N$ of sizes $m_{A,N}$ and $m_{B,N}$, choose vertices $a_N^\circ\in A_N$ and $b_N^\circ\in B_N$, and join them by a simple path with $s_N$ edges and $s_N-1$ internal vertices. There are no other cross-community edges. The resulting assignment graph is connected and
$d_N^{\Zcal}(A_N,B_N)=s_N$. Take $\Ocal_N=\Zcal_N$.

Assign the coordinate vectors $e_1,\ldots,e_r$ equally often within each of $A_N$, $B_N$, and the internal bridge. The resulting matrix $U_N$ has unit-norm rows and
\[
U_N^\top U_N=\frac{L_N}{r}I_r,\qquad
U_{N,A_N}^\top U_{N,A_N}=\frac{|A_N|}{r}I_r,
\qquad
U_{N,B_N}^\top U_{N,B_N}=\frac{|B_N|}{r}I_r.
\]
Set $a_N=L_N^{-1/4}$ and
$\Omega_N=a_NU_NU_N^\top+(1-a_N)I_{L_N}$. Unit row norms give
$\max_{a\ne c}|\Omega_{N,ac}|\le a_N\to0$. Since $U_NU_N^\top$ has $r$ nonzero eigenvalues equal to $L_N/r$,
\[
r_{\rm eff}(\Omega_N)
=\frac{L_N^2}{r\{1-a_N+a_NL_N/r\}^2+(L_N-r)(1-a_N)^2}
\sim rL_N^{1/2}.
\]
Proposition~\ref{prop:global_factor_certifiability}(iii) gives
\[
\rho_{\Omega_N}(A_N,B_N)
\ge
\frac{a_N}{a_N+(1-a_N)r/\min\{|A_N|,|B_N|\}}
\longrightarrow1.
\]
Canonical correlation is nondecreasing when coordinates are added. Because $s_N>2r_Y$ eventually,
\[
\theta_N^{\rm GC}(s_N)
\ge \rho_{\Omega_N}(A_N^{+r_Y},B_N^{+r_Y})
\ge \rho_{\Omega_N}(A_N,B_N)\longrightarrow1.
\]
This proves the latent-certificate obstruction. No reverse maximal-correlation inequality through the threshold map is used.
\end{proof}

\begin{proof}[Proof of Corollary~\ref{cor:local_factor_ic0}]
Conditional on $\mathcal F_N$, the local-factor construction in \eqref{eq:local_factor_sampler} has covariance $\Omega_N=\mathsf L_N\mathsf L_N^\top\succeq0$ and unit diagonal. If two assignment sets are more than $r_F$ apart, their shock-index unions are disjoint and the corresponding Gaussian subvectors are independent. Thresholding preserves independence. Expanding two $r_Y$-local outcome neighborhoods adds at most $2r_Y$ to the separation radius. Corollary~\ref{cor:finite_range_clt} with $h_*=2r_Y+r_F$ and \eqref{eq:local_factor_ic0_rate} then gives \eqref{eq:clt_centered}.
\end{proof}

\subsection{Disagreement lower bounds and causal-centering incompatibility}\label{subsec:causal_centering_proofs}

\begin{proof}[Proof of Theorem~\ref{thm:path_trilemma}]
The centered CLT follows from Theorem~\ref{thm:cell_clt}. For the path and $s=0$, the shell $\partial\mathcal B_N(a;0)$ contains only $a$, and
$\mathcal B_N(a;m_N)\setminus\mathcal B_N(a;-1)=\mathcal B_N(a;m_N)$. The power-mean inequality and the explicit path-ball cardinalities imply, for a universal $c>0$ and every $m_N\ge1$,
\begin{equation}
    c_N(0,m_N;1)
    \ge \frac1{M_N}\sum_{a=1}^{M_N}|\mathcal B_N(a;m_N)|
    \ge c\min\{m_N+1,M_N\}.
    \label{eq:path_c_lower_bound_proof}
\end{equation}
Because $\theta_N(0)=1$, the $s=0$, $k=1$ summand in
\eqref{eq:weak_dep_clt_rate_1} is bounded below by a constant multiple of
$M_N\min\{m_N+1,M_N\}/\sigma_N^3$. Under
$\sigma_N^2\asymp M_N$ on the good-environment events, convergence of the rate in outer probability forces
$\min\{m_N+1,M_N\}=o(M_N^{1/2})$, and hence
$m_N=o(M_N^{1/2})$. This proves \eqref{eq:path_radius_upper_rootM}.

By \eqref{eq:far_assignment_decorrelation}, Proposition~\ref{prop:path_bias_lower_bound} applies with
$\epsilon_{m_N}=\varepsilon_N^Z=o_p(1)$. Since $m_N=o(M_N)$, after reducing a constant $c_0>0$ if necessary,
\[
|\mathbb E_N\widehat\tau_N^{HT}-\tau_N|\ge \frac{c_0}{m_N}.
\]
Meanwhile $\sqrt{V_N}=\sigma_N/M_N\asymp M_N^{-1/2}$. Therefore
\[
\frac{|\mathbb E_N\widehat\tau_N^{HT}-\tau_N|}{\sqrt{V_N}}
\ge c_1\frac{M_N^{1/2}}{m_N}\xrightarrow{p}\infty,
\]
which proves \eqref{eq:path_trilemma_divergence}. No implication from contribution graph-$\psi$ dependence to assignment decorrelation is used: the latter enters only through the explicit condition \eqref{eq:far_assignment_decorrelation}.
\end{proof}

\begin{proof}[Proof of Corollary~\ref{cor:gaussian_path_trilemma}]
Since $\sigma_N^2\asymp M_N$,
\[
\frac{M_N^2}{\sigma_N}\asymp M_N^{3/2}\longrightarrow\infty.
\]
The second weak-dependence rate therefore implies
$\theta_N^{\rm GC}(m_N)\xrightarrow{p}0$. For every admissible $i$, the singleton sets
$A=\{i\}$ and $B=\{i+m_N\}$ are separated by $m_N$; eventually
$m_N>2r_Y$. Since their expanded Gaussian neighborhoods contain the original coordinates,
\[
|\Omega_{N,i,i+m_N}|
\le \rho_{\Omega_N}(A^{+r_Y},B^{+r_Y})
\le \theta_N^{\rm GC}(m_N).
\]
The normalized Bernoulli assignment correlation is
$\Rnorm_{N,i,i+m_N}=\mathcal G_{p_N}(\Omega_{N,i,i+m_N})$.
The map $(p,\rho)\mapsto\mathcal G_p(\rho)$ is uniformly continuous on the compact overlap set and satisfies $\mathcal G_p(0)=0$. Consequently
\[
\sup_{1\le i\le M_N-m_N}|\Rnorm_{N,i,i+m_N}|\xrightarrow{p}0,
\]
which verifies \eqref{eq:far_assignment_decorrelation}. Theorem~\ref{thm:path_trilemma} gives the result.
\end{proof}

\begin{proof}[Proof of Corollary~\ref{cor:variance_scale_window}]
Because $m_N=o(M_N)$, the path lower bound \eqref{eq:path_c_lower_bound_proof} is of order $m_N$. The $s=0$, $k=1$ term of \eqref{eq:weak_dep_clt_rate_1} therefore implies
$M_Nm_N/\sigma_N^3\to0$, which is \eqref{eq:clt_radius_upper_window}. Vanishing far correlation and Proposition~\ref{prop:path_bias_lower_bound} give a raw bias of order at least $m_N^{-1}$. Negligibility relative to
$\sqrt{V_N}=\sigma_N/M_N$ therefore requires
$m_N^{-1}=o(\sigma_N/M_N)$, which is
\eqref{eq:bias_radius_lower_window}. The two strict inequalities can share an order window only if
\[
    \frac{M_N}{\sigma_N}
    =o\!\left(\frac{\sigma_N^3}{M_N}\right),
\]
equivalently $\sigma_N^2/M_N\to\infty$.
\end{proof}

\begin{proof}[Proof of Corollary~\ref{cor:fixed_range_bias_floor}]
Take $m=r_F+1$. Independence at that distance implies
$\Rnorm_{N,i,i+m}=0$ for every admissible $i$. Proposition~\ref{prop:path_bias_lower_bound} gives
\[
    |\E\widehat\tau_N^{HT}-\tau_N|
    \ge \frac{\vartheta_0h_0}{M_N}
      \frac{M_N-m}{m}.
\]
Letting $N\to\infty$ with fixed $m$ proves
\eqref{eq:fixed_range_path_bias_floor}.
\end{proof}

\begin{proof}[Proof of Proposition~\ref{prop:external_debiasing}]
Under the exact two-channel model, Proposition~\ref{prop:misspecification_bias} with zero remainder gives
$\E\widehat\tau^{HT}_{w,\tilde\alpha}-\tau_{w,N}=\delta_N^\top c_N$. Hence
\[
\begin{split}
    \frac{\widehat\tau_N^{\rm BC}-\tau_{w,N}}{\sqrt{V_N}}
    ={}&
    \frac{\widehat\tau^{HT}_{w,\tilde\alpha}
      -\E\widehat\tau^{HT}_{w,\tilde\alpha}}{\sqrt{V_N}}
    -\frac{c_N^\top(\widehat\delta_N-\delta_N)}{\sqrt{V_N}}.
\end{split}
\]
The first term converges to $N(0,1)$ by IC-0 and the second is $o_p(1)$ by
\eqref{eq:debiased_nuisance_rate}. Slutsky's theorem proves
\eqref{eq:debiased_causal_clt}.
\end{proof}

\begin{proof}[Proof of Corollary~\ref{cor:external_debiasing_pilot_size}]
Cauchy--Schwarz and \eqref{eq:external_debiasing_parameter_rate} give
\[
|c_N^\top(\widehat\delta_N-\delta_N)|
\le\|c_N\|_2\|\widehat\delta_N-\delta_N\|_2
=O_p\!\left(\|c_N\|_2
\sqrt{\frac{d_{\delta,N}}{n_{\delta,N}}}\right).
\]
Condition \eqref{eq:external_debiasing_pilot_size} makes this
$o_p(\sqrt{V_N})$, which is \eqref{eq:debiased_nuisance_rate}. The root-$M_N$ statement follows by substituting
$V_N\asymp M_N^{-1}$, fixed $d_{\delta,N}$, and
$\|c_N\|_2\asymp1$.
\end{proof}

\begin{proof}[Proof of Proposition~\ref{prop:external_pilot_uncertainty}]
All operators in this proof are under $\mathbb P_N^J(\cdot)=\Pr(\cdot\mid\mathcal E_N)$. The exact bias identity, conditional on the realized pilot-derived design, gives
\[
\widehat\tau_N^{\rm BC}-\tau_{w,N}=U_N-D_N.
\]
Therefore
\[
\frac{\widehat\tau_N^{\rm BC}-\tau_{w,N}}
{\sqrt{V_N^D+V_{\delta,N}^J}}
=\sqrt{\lambda_N^J}\frac{U_N}{\sqrt{V_N^D}}
-\sqrt{1-\lambda_N^J}\frac{D_N}{\sqrt{V_{\delta,N}^J}}.
\]
By \eqref{eq:external_pilot_joint_limits}, the right-hand side converges under the joint law to
$\sqrt\lambda Z_1-\sqrt{1-\lambda}Z_2$, where
$(Z_1,Z_2)\sim N_2(0,I_2)$. This linear combination is standard normal for every
$\lambda\in[0,1]$, proving \eqref{eq:external_pilot_variance_propagation}. Joint-law ratio consistency and Slutsky's theorem give the feasible version.

For the primitive route stated after the proposition, stable conditional convergence of $U_N/\sqrt{V_N^D}$ relative to $\mathcal F_N$ implies joint convergence with every pilot-measurable variable having its own weak limit. Applying this to $D_N/\sqrt{V_{\delta,N}^J}$ yields the independent bivariate Gaussian limit even though the current design is pilot-adaptive. Also, $\mathbb E_N^D(U_N\mid\mathcal F_N)=0$, so $\mathbb E_N^J(U_ND_N\mid\mathcal E_N)=0$ exactly; this uncorrelatedness alone would not suffice without the stable joint-limit argument.
\end{proof}

\subsection{Misspecification-aware bias}

\begin{proof}[Proof of Proposition~\ref{prop:misspecification_bias}]
The fixed adjustment contributes no expectation because
\[
    \E\{\psi_{Na}(Z_{Na})\tilde\alpha_{Na}\}=0.
\]
A disturbance fixed under randomization drops out for the same reason, while a random measurement error drops out by \eqref{eq:noise_ht_orthogonality}. The common-$p_N$ HT identities give
\[
    \E\{\psi_{Na}(Z_{Na})Z_{Na}\}=1,
    \qquad
    \E\{\psi_{Na}(Z_{Na})Z_{Nc}\}=\Rnorm_{N,ac}.
\]
Applying these identities to \eqref{eq:working_model_with_remainder} yields
\[
\begin{split}
    \E\widehat\tau^{HT}_{w,\tilde\alpha}
    =\frac{1}{W_N}\sum_aw_{Na}\Bigg[
      \beta_{Na}
      +\gamma_N\sum_cH^{\comp}_{N,ac}\Rnorm_{N,ac}
      +\eta_N\sum_cH^{\lag}_{N,ac}\Rnorm_{N,ac}
      +\E\{\psi_{Na}r_{Na}(Z_N)\}\Bigg].
\end{split}
\]
The weighted global contrast is
\[
\begin{split}
    \tau_{w,N}=\frac{1}{W_N}\sum_aw_{Na}\Bigg[
      \beta_{Na}
      +\gamma_N\sum_cH^{\comp}_{N,ac}
      +\eta_N\sum_cH^{\lag}_{N,ac}
      +r_{Na}(\one)-r_{Na}(\mathbf 0)\Bigg].
\end{split}
\]
Subtracting proves \eqref{eq:weighted_bias_inference}--\eqref{eq:remainder_bias_definition}; the triangle inequality proves \eqref{eq:bias_bound_inference}.

If $\sup_z|r_{Na}(z)|\le\bar r_{Na}$, then
$\E|\psi_{Na}|=2$, so
$|\E(\psi_{Na}r_{Na})|\le2\bar r_{Na}$. Also
$|r_{Na}(\one)-r_{Na}(\mathbf 0)|\le2\bar r_{Na}$. Summing the resulting bound proves \eqref{eq:crude_remainder_bias_bound}.
\end{proof}

The incompatibility concerns causal centering of raw HT, not variance estimation. Once IC-1 is restored---through shrinking interference, debiasing, a different estimand, or a different asymptotic route---the remaining and logically separate IC-2 problem is to estimate $V_N$. We now turn to that step.

\subsection{Network HAC and feasible centering}\label{subsec:hac_proof_appendix}

\begin{proof}[Proof of Proposition~\ref{prop:oracle_network_hac_connected}]
The target $\Sigma_N=\Var_N(T_N/\sqrt{M_N})$ is the scalar conditional long-run variance in the network HAC theory of \citet{kojevnikov2021network}. Assumption~\ref{ass:contribution_metric_compatibility} makes the induced contribution graph $\mathcal G_N^{\Ocal}$ connected and makes $d_N$ its shortest-path metric. Their oracle quadratic form is exactly \eqref{eq:oracle_network_hac}. Assumption~\ref{ass:weights_moments} supplies the uniform $L^\nu$ bound, while \eqref{eq:hac_kernel_bias_rate} and \eqref{eq:hac_stochastic_rate} are a scalar sufficient specialization of the kernel-approximation and network-complexity restrictions. Applying their theorem pointwise on the good-environment events and using the same subsequence argument as in Theorem~\ref{thm:cell_clt} gives
$\widehat\Sigma_N^{\rm or}-\Sigma_N\xrightarrow{p\mid\mathcal F}0$. Since $\Sigma_N$ is bounded away from zero on those events, division gives conditional ratio consistency, and
$\widehat V_N^{\rm or}/V_N=\widehat\Sigma_N^{\rm or}/\Sigma_N$.
\end{proof}

\begin{proof}[Proof of Proposition~\ref{prop:oracle_network_hac_disconnected}]
Because $\mathbb E_N\xi_{Na}=0$ and $K(\infty)=0$,
\begin{equation}
\begin{split}
\mathbb E_N\widehat\Sigma_N^{\rm or,disc}-\Sigma_N
={}&\frac1{M_N}\sum_{a,c:\,\widetilde d_N(a,c)<\infty}
\{K(\widetilde d_N(a,c)/b_N)-1\}
\Cov_N(\xi_{Na},\xi_{Nc})\\
&-\frac1{M_N}\sum_{a,c:\,\widetilde d_N(a,c)=\infty}
\Cov_N(\xi_{Na},\xi_{Nc}).
\end{split}
\label{eq:disconnected_hac_expectation_decomposition}
\end{equation}
The usual truncation/interpolation argument for graph-$\psi$ arrays and the conditional uniform $L^\nu$ bound yield, for finite separation $s$,
\[
|\Cov_N(\xi_{Na},\xi_{Nc})|
\le C\theta_N(s)^{1-2/\nu}.
\]
Grouping the first sum by finite shells bounds it by a constant multiple of
\eqref{eq:hac_disconnected_finite_bias_rate}; the absolute value of the second sum is bounded by
$\mathfrak B_{\infty,N}^{\rm abs}$. Thus
$\mathbb E_N\widehat\Sigma_N^{\rm or,disc}-\Sigma_N\xrightarrow{p}0$ under the outer environment law.

For conditional stochastic error,
\begin{equation}
\Var_N(\widehat\Sigma_N^{\rm or,disc})
=\frac1{M_N^2}
\sum_{(a,c),(k,l)\in\mathcal I_N(b_N)}
\widetilde{\mathsf K}_N(a,c)\widetilde{\mathsf K}_N(k,l)
\Cov_N(\xi_{Na}\xi_{Nc},\xi_{Nk}\xi_{Nl}).
\label{eq:disconnected_hac_variance_decomposition}
\end{equation}
Split the sum according to whether the two included pairs have finite or infinite set distance. The same fourth-moment truncation and finite-shell counting argument as in connected network HAC bounds the finite-distance part by
\begin{equation}
\frac C{M_N}\sum_{s\ge0}
\widetilde c_N(s,b_N;2)\theta_N(s)^{1-4/\nu},
\label{eq:disconnected_hac_finite_variance_bound}
\end{equation}
which converges to zero in outer probability by \eqref{eq:hac_disconnected_finite_stochastic_rate}. The absolute value of the infinite-distance part is at most
$\mathfrak S_{\infty,N}(b_N)$ and therefore also vanishes in outer probability. Conditional Chebyshev's inequality, the conditional expectation result, and the good-environment lower bound for $\Sigma_N$ prove \eqref{eq:oracle_hac_disconnected_consistency}.
\end{proof}

\begin{proof}[Proof of Corollary~\ref{cor:component_independent_hac}]
If $\widetilde d_N(a,c)=\infty$, then $a$ and $c$ belong to distinct connected components, so conditional component independence gives
$\Cov_N(\xi_{Na},\xi_{Nc})=0$. Hence
$\mathfrak B_{\infty,N}^{\rm abs}=0$. Every included pair $(a,c)\in\mathcal I_N(b_N)$ lies within a single component. If two included pairs have infinite set distance, they lie in distinct components, and their products are independent; hence the corresponding covariance is zero and
$\mathfrak S_{\infty,N}(b_N)=0$.
\end{proof}

\begin{proof}[Proof of Proposition~\ref{prop:model_assisted_hac}]
Because $\widehat\xi_N^{\rm MA}=\xi_N-r_N$,
\begin{equation}
\begin{split}
    \widehat\Sigma_N^{\rm MA}-\widehat\Sigma_N^{\rm or}
    &=\frac{1}{M_N}
      \{(\xi_N-r_N)^\top\mathsf K_N(\xi_N-r_N)
      -\xi_N^\top\mathsf K_N\xi_N\}\\
    &=\frac{1}{M_N}
      \{r_N^\top\mathsf K_Nr_N-2r_N^\top\mathsf K_N\xi_N\}.
\end{split}
\label{eq:feasible_oracle_hac_difference}
\end{equation}
The first quadratic form in \eqref{eq:centering_error_conditions} is $\mathcal F_N$-measurable under the convention in \eqref{eq:conditional_probability_operators}; hence its outer $o_p$ rate is also conditional-in-probability convergence in the sense of \eqref{eq:conditional_probability_convergence_notation}. Together with the second condition, \eqref{eq:feasible_oracle_hac_difference} is
$o_{p\mid\mathcal F}(\sigma_N^2/M_N)=o_{p\mid\mathcal F}(\Sigma_N)$. Combining this with oracle consistency proves
$\widehat\Sigma_N^{\rm MA}/\Sigma_N\xrightarrow{p\mid\mathcal F}1$. Because $\Sigma_N$ is bounded away from zero on the good-environment events, this ratio convergence implies
$\mathbb P_N\{\widehat\Sigma_N^{\rm MA}>0\}\xrightarrow{p}1$; hence taking the outer positive part leaves the limit unchanged and gives
$\widehat V_N^{\rm MA,+}/V_N\xrightarrow{p\mid\mathcal F}1$.
\end{proof}

\begin{proof}[Proof of Corollary~\ref{cor:primitive_centering_rate}]
Let $r_N=\widehat q_N^0-\bar q_N$. Positive semidefiniteness and the operator-norm bound imply
\[
    0\le r_N^\top\mathsf K_Nr_N
    \le\Lambda_N\|r_N\|_2^2=o_p(\sigma_N^2).
\]
Oracle consistency gives
$\xi_N^\top\mathsf K_N\xi_N=O_{p\mid\mathcal F}(\sigma_N^2)$ because
$M_N^{-1}\xi_N^\top\mathsf K_N\xi_N
=\widehat\Sigma_N^{\rm or}=O_{p\mid\mathcal F}(\Sigma_N)$ and
$M_N\Sigma_N=\sigma_N^2$. Cauchy--Schwarz in the
$\mathsf K_N$ seminorm therefore yields
\[
    |r_N^\top\mathsf K_N\xi_N|
    \le(r_N^\top\mathsf K_Nr_N)^{1/2}
       (\xi_N^\top\mathsf K_N\xi_N)^{1/2}
    =o_{p\mid\mathcal F}(\sigma_N^2).
\]
Thus the outer quadratic-form rate and the conditional cross-term rate in \eqref{eq:centering_error_conditions} hold. If
$\sigma_N^2\asymp M_N$ and $\Lambda_N\lesssim D_N(b_N)$, rearranging
\eqref{eq:primitive_centering_l2_rate} gives
\eqref{eq:average_centering_mse_rate}.
\end{proof}

\begin{proof}[Proof of Corollary~\ref{cor:indefinite_kernel_centering}]
The operator-norm inequality gives
\[
|r_N^\top\mathsf K_Nr_N|\le\Lambda_N\|r_N\|_2^2,
\qquad
|r_N^\top\mathsf K_N\xi_N|
\le\Lambda_N\|r_N\|_2\|\xi_N\|_2.
\]
Together with $\|\xi_N\|_2=O_{p\mid\mathcal F}(\sqrt{M_N})$ and
\eqref{eq:indefinite_kernel_centering_rates}, these bounds verify
\eqref{eq:centering_error_conditions}. The algebraic decomposition \eqref{eq:feasible_oracle_hac_difference} and the selected certified oracle-HAC proposition then prove conditional ratio consistency.

Under the external-pilot model,
\[
\|r_N\|_2^2
\le \lambda_{\max}(\mathsf P_N^\top\mathsf P_N)
\|\widehat\vartheta_N-\vartheta_{0,N}\|_2^2
=O_p\!\left(\frac{M_Nd_{q,N}}{n_{0,N}}\right).
\]
When $\sigma_N^2\asymp M_N$, the two conditions in
\eqref{eq:indefinite_kernel_centering_rates} reduce to
$\Lambda_Nd_{q,N}/n_{0,N}\to_p0$ and
$\Lambda_N\sqrt{d_{q,N}/n_{0,N}}\to_p0$. Because a symmetric kernel with diagonal entries one has $\|\mathsf K_N\|_{\rm op}\ge1$, a valid bound satisfies $\Lambda_N\ge1$; hence
\eqref{eq:indefinite_external_pilot_rate} implies both requirements.
\end{proof}

\begin{proof}[Proof of Corollary~\ref{cor:external_pilot_centering}]
Proposition~\ref{prop:oracle_network_hac_connected} supplies conditional oracle ratio consistency. It therefore remains only to verify the primitive centering rate. Let $r_N=\mathsf P_N(\widehat\vartheta_N-\vartheta_{0,N})$. By
\eqref{eq:external_pilot_mean_model}--\eqref{eq:external_pilot_parameter_rate},
\[
\|r_N\|_2^2
\le\lambda_{\max}(\mathsf P_N^\top\mathsf P_N)
\|\widehat\vartheta_N-\vartheta_{0,N}\|_2^2
=O_p\!\left(\frac{M_Nd_{q,N}}{n_{0,N}}\right).
\]
Condition \eqref{eq:external_pilot_hac_rate} therefore implies
\eqref{eq:primitive_centering_l2_rate}, and Corollary~\ref{cor:primitive_centering_rate} proves conditional ratio consistency. Under
$\sigma_N^2\asymp M_N$ and $\Lambda_N\lesssim D_N(b_N)$,
\eqref{eq:external_pilot_rootM_rate} is exactly a sufficient rearrangement of
\eqref{eq:external_pilot_hac_rate}. Combining ratio consistency with IC-1 and Corollary~\ref{cor:studentized_model_assisted} yields IC-2.
\end{proof}

\begin{proof}[Proof of Corollary~\ref{cor:studentized_model_assisted}]
Proposition~\ref{prop:model_assisted_hac} gives
$\widehat V_N^{\rm MA,+}/V_N\xrightarrow{p\mid\mathcal F}1$. Combining this ratio convergence with the centered limit in Theorem~\ref{thm:cell_clt} and conditional Slutsky proves
\eqref{eq:studentized_centered_clt}. If
$|\mathbb E_N\widehat\tau^{HT}_{w,\tilde\alpha}-\tau_{w,N}|/\sqrt{V_N}\xrightarrow{p}0$, replacing the design expectation by $\tau_{w,N}$ changes the numerator by $o(\sqrt{V_N})$ and hence leaves the limit unchanged.
\end{proof}

\begin{proof}[Proof of Proposition~\ref{prop:uncentered_conservative}]
Write $Q_N=\bar q_N+\xi_N$. Since $\E\xi_N=0$,
\[
\begin{split}
    \E\widehat\Sigma_N^{\rm UC}
    &=\frac{1}{M_N}\E(Q_N^\top\mathsf K_NQ_N)\\
    &=\frac{1}{M_N}\tr\{\mathsf K_N\Cov(\xi_N)\}
      +\frac{1}{M_N}\bar q_N^\top\mathsf K_N\bar q_N\\
    &=\Sigma_{K,N}
      +\frac{1}{M_N}\bar q_N^\top\mathsf K_N\bar q_N.
\end{split}
\]
The second term is nonnegative because $\mathsf K_N\succeq0$. The one-sided approximation assumption gives
$\E\widehat\Sigma_N^{\rm UC}\ge\Sigma_N-o(\Sigma_N)$. Adding the assumed concentration error proves \eqref{eq:uncentered_conservative_result}.
\end{proof}

\begin{proof}[Proof of Proposition~\ref{prop:bias_aware_coverage}]
Let $b_N=\E\widehat\tau^{HT}_{w,\tilde\alpha}-\tau_{w,N}$ and
$z=z_{1-\alpha_0/2}$. On the intersection of $\mathcal E_N$ and the variance-calibration event, noncoverage implies
\[
    \frac{|\widehat\tau^{HT}_{w,\tilde\alpha}
    -\E\widehat\tau^{HT}_{w,\tilde\alpha}|}{\sqrt{V_N}}
    >z\sqrt{1-\varepsilon_N}-\frac{s_N}{\sqrt{V_N}}.
\]
The right-hand threshold converges to $z$. The centered CLT and continuity of the standard normal distribution therefore make the limiting probability of this event at most $\alpha_0$. The complement of the variance-calibration event has probability $o(1)$, while
$\limsup_N\Pr(\mathcal E_N^c)\le\delta$. A union bound proves \eqref{eq:bias_aware_coverage_bound}.
\end{proof}

\begin{proof}[Proof of Proposition~\ref{prop:direction_aware_coverage}]
Write
$b_N=\E\widehat\tau^{HT}_{w,\tilde\alpha}-\tau_{w,N}$ and
$U_N=\widehat\tau^{HT}_{w,\tilde\alpha}-
      \E\widehat\tau^{HT}_{w,\tilde\alpha}$.
On $\mathcal E_N$, condition \eqref{eq:directional_bias_event} implies
\[
\tau_{w,N}\in
[\E\widehat\tau^{HT}_{w,\tilde\alpha}-s_N,
 \E\widehat\tau^{HT}_{w,\tilde\alpha}+\widehat B_N+s_N].
\]
On the additional variance-calibration event, noncoverage by
\eqref{eq:direction_aware_ci} can occur only if
\[
\frac{U_N}{\sqrt{V_N}}
>z_{1-\alpha_L}\sqrt{1-\varepsilon_N}-\frac{s_N}{\sqrt{V_N}}
\]
or
\[
\frac{U_N}{\sqrt{V_N}}
<-z_{1-\alpha_U}\sqrt{1-\varepsilon_N}+\frac{s_N}{\sqrt{V_N}}.
\]
The centered CLT and continuity of the normal distribution make the limiting probabilities of the two events at most $\alpha_L$ and $\alpha_U$, respectively. Their sum is $\alpha_0$. The complement of the variance-calibration event is $o(1)$ and
$\limsup_N\Pr(\mathcal E_N^c)\le\delta$. A union bound proves
\eqref{eq:direction_aware_coverage_bound}.
\end{proof}

\subsection{Design-aware tail correction}

\begin{proof}[Proof of Proposition~\ref{prop:design_aware_hac}]
The signed estimator satisfies
\[
\begin{split}
    \widehat\Sigma_{\rm DA}^{\rm sgn}-\Sigma_N
    &=(\widehat\Sigma_{\rm loc}-\Sigma_{K,N})
      +(\widehat{\mathcal T}_{\rm pl}-\mathcal T_N)
      =o_p(\Sigma_N),
\end{split}
\]
proving ratio consistency.

The map $x\mapsto[x]_+$ is one-Lipschitz, so
$[\widehat{\mathcal T}_{\rm pl}]_+
=[\mathcal T_N]_++o_p(\Sigma_N)$. Therefore
\[
\begin{split}
    \widehat\Sigma_{\rm DA,raw}^{+}-\Sigma_N
    &=[\mathcal T_N]_+-\mathcal T_N+o_p(\Sigma_N)\\
    &=\mathcal T_N^-+o_p(\Sigma_N).
\end{split}
\]
This proves one-sided conservativeness and shows that ratio consistency is equivalent to
$\mathcal T_N^-=o(\Sigma_N)$. In addition,
$\widehat\Sigma_{\rm DA,raw}^{+}\ge\Sigma_N-o_p(\Sigma_N)$ and $\Sigma_N>0$ imply
$\Pr\{\widehat\Sigma_{\rm DA,raw}^{+}>0\}\to1$. The outer positive part in \eqref{eq:design_aware_hac} therefore guarantees finite-sample nonnegativity without changing either the one-sided conclusion or the ratio-consistency characterization.
\end{proof}

To justify the plug-in sufficient conditions, let
$S_N=\sum_aQ_{Na}$ and $E_N=\sum_ae_{Na}$. The elementary inequality
\[
    |\Var(S_N+E_N)-\Var(S_N)|
    \le \Var(E_N)+2\{\Var(S_N)\Var(E_N)\}^{1/2}
\]
shows that \eqref{eq:plugin_full_l2_condition} makes the full variance error
$o_p(\sigma_N^2)$. For a PSD $\mathsf K_N$, apply Cauchy--Schwarz in the seminorm
$\|x\|_{\mathsf K_N}=(x^\top\mathsf K_Nx)^{1/2}$ to the centered true and plug-in errors. Condition \eqref{eq:plugin_local_l2_condition}, together with an $O(\sigma_N^2)$ true local quadratic variance, makes both the pure error term and cross term $o_p(\sigma_N^2)$. Dividing by $M_N$ gives \eqref{eq:tail_consistency}. Standard sample-mean and sample-variance calculations give the stated Monte Carlo condition.

\subsection{H{\'a}jek and difference-in-means limits}

\begin{proof}[Proof of Proposition~\ref{prop:hajek_delta}]
Fix $\lambda\in\mathbb R^4$. If $\lambda^\top\Gamma_R\lambda>0$, Theorem~\ref{thm:cell_clt} applied to the projected scalar array gives
\[
    \frac{\lambda^\top S_{R,N}}
    {(\lambda^\top\Gamma_N\lambda)^{1/2}}
    \Rightarrow N(0,1),
\]
and hence $\lambda^\top S_{R,N}\Rightarrow
N(0,\lambda^\top\Gamma_R\lambda)$. If
$\lambda^\top\Gamma_R\lambda=0$, then
$\Var(\lambda^\top S_{R,N})=\lambda^\top\Gamma_N\lambda\to0$, so Chebyshev's inequality gives $\lambda^\top S_{R,N}\to_p0$. The Cram\'er--Wold device proves \eqref{eq:hajek_multivariate_clt} without assuming that conclusion in advance.

Because $\bar R_N-\rho_N=o_p(1)$ and $\rho_{N,3}\to\rho_3>0$, $\rho_{N,4}\to\rho_4>0$, we have
$\Pr(\mathcal D_N^H)\to1$. Hence the arbitrary fixed convention on
$(\mathcal D_N^H)^c$ is asymptotically irrelevant. The function
$h(x)=x_1/x_3-x_2/x_4$ is continuously differentiable on a neighborhood of $\rho$. The multivariate CLT also gives $S_{R,N}=O_p(1)$ and therefore
$\bar R_N-\rho_N=O_p(M_N^{-1/2})$ along the nondegenerate coordinates. A first-order Taylor expansion on $\mathcal D_N^H$ gives
\[
    \sqrt{M_N}\{h(\bar R_N)-h(\rho_N)\}
    =\nabla h(\rho_N)^\top S_{R,N}+o_p(1).
\]
Because $\nabla h(\rho_N)\to\nabla h(\rho)$, the multivariate delta method proves \eqref{eq:hajek_delta_limit}. The causal-centering statement follows by Slutsky's theorem.

For the strong-balance claim, write
\[
    A_N=M_N^{-1}\sum_a\omega_{Na}Z_{Na}X_{Na}^{\rm obs},\quad
    C_N=M_N^{-1}\sum_a\omega_{Na}(1-Z_{Na})X_{Na}^{\rm obs},
\]
and let $D_N,E_N$ be the corresponding normalized treatment and control denominators. Then
$\widehat\tau^H=A_N/D_N-C_N/E_N$, while
$\widehat\tau^{HT}=A_N/p_N-C_N/(1-p_N)$. Condition \eqref{eq:strong_balance_ratio} gives
$D_N-p_N=o_p(M_N^{-1/2})$; because $D_N+E_N=1$, it also gives
$E_N-(1-p_N)=o_p(M_N^{-1/2})$. A Taylor expansion of both ratios, together with overlap and $A_N,C_N=O_p(1)$, yields
$\widehat\tau^H-\widehat\tau^{HT}=o_p(M_N^{-1/2})$.
\end{proof}

\section{Technical Derivations}
\label{app:technical_derivations}
\label{app:derivations}

Throughout the appendix, $q=1-p$, $v=pq$, and $\Rnorm=\Cov(Z)/v$. Rows of $H$ outside $\Ocal$ are zero, and $\beta_c=0$ for $c\notin\Ocal$.

\subsection{Basic HT identities}

Recall
\[
    \psi_a=\frac{Z_a}{p}-\frac{1-Z_a}{1-p}=\frac{Z_a-p}{v}.
\]
Then
\[
    \E[\psi_a]=\frac{\E[Z_a]-p}{v}=0.
\]
Since $Z_a^2=Z_a$,
\[
    \E[\psi_aZ_a]
    =\frac{\E[(Z_a-p)Z_a]}{v}
    =\frac{\E[Z_a]-p\E[Z_a]}{v}
    =\frac{p(1-p)}{v}=1.
\]
For any cell $c$,
\[
    \E[\psi_aZ_c]
    =\frac{\E[(Z_a-p)Z_c]}{v}
    =\frac{\E[Z_aZ_c]-p\E[Z_c]}{v}
    =\frac{\Cov(Z_a,Z_c)}{v}
    =\Rnorm_{ac}.
\]

\subsection{Proof of Theorem~\ref{thm:main}: bias}\label{subsec:main_bias_proof}

Let the structural baseline be $\mu_a$ and let the estimator subtract a fixed adjustment $\tilde\alpha_a$. Under the generic baseline version of \eqref{eq:single_H_model},
\[
    Y_a^{\mathrm{obs}}-\tilde\alpha_a
    =(\mu_a-\tilde\alpha_a)+\beta_aZ_a+\theta\sum_cH_{ac}Z_c+\varepsilon_a.
\]
If $\varepsilon_a$ is fixed, then $\E(\psi_a\varepsilon_a)=\varepsilon_a\E\psi_a=0$; if it is random, condition \eqref{eq:noise_ht_orthogonality} gives the same conclusion. The fixed adjustment and structural baseline also drop out because $\E[\psi_a]=0$:
\begin{align*}
    \E[\widehat\tau_{\tilde\alpha}]
    &=\frac{1}{M}\sum_{a\in\Ocal}
    \E\left[\psi_a\left((\mu_a-\tilde\alpha_a)+\beta_aZ_a+\theta\sum_cH_{ac}Z_c\right)\right]\\
    &=\frac{1}{M}\sum_{a\in\Ocal}
    \left(\beta_a\E[\psi_aZ_a]+\theta\sum_cH_{ac}\E[\psi_aZ_c]\right)\\
    &=\frac{1}{M}\sum_{a\in\Ocal}
    \left(\beta_a+\theta\sum_cH_{ac}\Rnorm_{ac}\right).
\end{align*}
The global treatment-control estimand is
\[
    \tau=\frac{1}{M}\sum_{a\in\Ocal}
    \left(\beta_a+\theta\sum_cH_{ac}\right).
\]
Therefore
\[
    \E[\widehat\tau_{\tilde\alpha}]-\tau
    =\frac{\theta}{M}\sum_{a\in\Ocal}\sum_cH_{ac}(\Rnorm_{ac}-1)
    =\frac{\theta}{M}\langle H,\Rnorm-\one\one^\top\rangle.
\]
For the two-component model,
\[
    \E[\widehat\tau_{\tilde\alpha}]-\tau
    =\frac{\gamma}{M}\langle H^{\comp},\Rnorm-\one\one^\top\rangle
    +\frac{\eta}{M}\langle H^{\lag},\Rnorm-\one\one^\top\rangle.
\]

\subsection{Derivation of the exact assignment-polynomial representation}\label{subsec:main_polynomial_proof}

Let the structural baseline be $\mu_a$ and let the estimator subtract a fixed adjustment $\tilde\alpha_a$. If the disturbance is fixed under the design law, define $r_a^{\tilde\alpha}=\mu_a+\varepsilon_a-\tilde\alpha_a$ and absorb it into the assignment-generated component. If measurement noise is random, define $r_a^{\tilde\alpha}=\mu_a-\tilde\alpha_a$ and first set the random noise to zero. In either case,
\begin{align*}
    M(\widehat\tau_{\tilde\alpha})^{(0)}
    &=\sum_{a\in\Ocal}\frac{Z_a-p}{v}
    \left\{r_a^{\tilde\alpha}+\beta_aZ_a+\theta\sum_cH_{ac}Z_c\right\}\\
    &=\frac{1}{v}\sum_{a\in\Ocal}r_a^{\tilde\alpha}(Z_a-p)
    +\sum_{a\in\Ocal}\frac{\beta_a}{v}(Z_a-p)Z_a\\
    &\quad +\frac{\theta}{v}\sum_{a,c}H_{ac}Z_aZ_c
    -\frac{\theta p}{v}\sum_{a,c}H_{ac}Z_c.
\end{align*}
The term $-p\sum_a r_a^{\tilde\alpha}/v$ is constant and drops out of the variance. Because $(Z_a-p)Z_a=qZ_a$,
\[
    \sum_{a\in\Ocal}\frac{\beta_a}{v}(Z_a-p)Z_a
    =\sum_{a\in\Ocal}\frac{\beta_a}{p}Z_a.
\]
Let $\deg_H=H^\top\one$. Since rows of $H$ outside $\Ocal$ are zero,
\[
    \sum_{a,c}H_{ac}Z_c=\deg_H^\top Z.
\]
Hence, up to constants,
\[
    M\widehat\tau_{\tilde\alpha}^{(0)}
    =\left(\frac{r^{\tilde\alpha}}{v}+\frac{\beta}{p}-\frac{\theta \deg_H}{q}\right)^\top Z
    +\frac{\theta}{v}Z^\top HZ.
\]
Define
\[
    g_p^{\tilde\alpha}=\frac{r^{\tilde\alpha}}{v}+\frac{\beta}{p}-\frac{\theta \deg_H}{q}.
\]
Then
\[
    \Var\{\widehat\tau_{\tilde\alpha}^{(0)}\}
    =\frac{1}{M^2}\Var\left((g_p^{\tilde\alpha})^\top Z+\frac{\theta}{v}Z^\top HZ\right).
\]
Expanding gives
\begin{align*}
    \Var\{\widehat\tau_{\tilde\alpha}^{(0)}\}
    =\frac{1}{M^2}\Bigg[&(g_p^{\tilde\alpha})^\top\Sigma g_p^{\tilde\alpha}
    +\frac{2\theta}{v}\Cov\{(g_p^{\tilde\alpha})^\top Z,Z^\top HZ\}\\
    &+\frac{\theta^2}{v^2}\Var(Z^\top HZ)\Bigg].
\end{align*}
The second and third terms involve third- and fourth-order moments, motivating the covariance-level upper bound.

If random mean-zero measurement noise is independent of $Z$ and mutually independent across cells,
\[
    \widehat\tau_{\tilde\alpha}=\widehat\tau_{\tilde\alpha}^{(0)}+\frac{1}{M}\sum_{a\in\Ocal}\psi_a\varepsilon_a.
\]
Independence and the zero error means imply that the cross covariance is zero. Since $\E[\psi_a^2]=1/v$,
\[
    \Var\left(\frac{1}{M}\sum_{a\in\Ocal}\psi_a\varepsilon_a\right)
    =\frac{1}{M^2v}\sum_{a\in\Ocal}\sigma_a^2.
\]
If measurement noise is network-time local rather than independent, it can be included in the centered contribution $\xi_{Na}$ used for the CLT, with the graph-dependence coefficient enlarged to account for the noise dependence. In the finite-range specialization this amounts to enlarging the relevant locality radius. The simple additive noise term above is for independent noise only.

\subsection{Weighted identities, disagreement geometry, and robust bias}\label{subsec:weighted_derivations}

\begin{proof}[Proof of Corollary~\ref{cor:weighted_ht_identity}]
The bias calculation in Theorem~\ref{thm:main} can be applied after multiplying the $a$th contribution by the fixed weight $w_a$. The baseline adjustment and any fixed disturbance have zero design expectation after multiplication by $\psi_a$. For random measurement noise, condition \eqref{eq:noise_ht_orthogonality} gives the same conclusion. Therefore
\[
\begin{split}
    \E\widehat\tau_{w,\tilde\alpha}^{HT}
    =\frac1W\sum_aw_a\left\{
      \beta_a+\theta\sum_cH_{ac}\Rnorm_{ac}\right\},
\end{split}
\]
whereas
\[
    \tau_w=\frac1W\sum_aw_a\left\{
      \beta_a+\theta\sum_cH_{ac}\right\}.
\]
Subtracting proves \eqref{eq:weighted_exact_bias}.

For the algebraic representation, use the same convention for
$r^{\tilde\alpha}$ as in Theorem~\ref{thm:main}. Up to a nonrandom constant,
\begin{equation}
W(\widehat\tau_{w,\tilde\alpha}^{HT})^{(0)}
=\frac1v(w\odot r^{\tilde\alpha})^\top Z
+\frac1p(w\odot\beta)^\top Z
-\frac{\theta}{q}(H^\top w)^\top Z
+\frac{\theta}{v}Z^\top D_wHZ .
\label{eq:weighted_exact_algebra}
\end{equation}

Finally, suppose the random errors are independent of $Z$, have mean zero, are mutually uncorrelated, and satisfy
$\Var(\varepsilon_a)=\sigma_{\varepsilon,a}^2$. The noise component is
$W^{-1}\sum_aw_a\psi_a\varepsilon_a$. Its covariance with every square-integrable assignment-measurable component is zero. Cross-cell noise covariances are zero, and $\E\psi_a^2=1/v$, so
\begin{equation}
\Var\left(W^{-1}\sum_aw_a\psi_a\varepsilon_a\right)
=\frac{1}{W^2v}\sum_aw_a^2\sigma_{\varepsilon,a}^2.
\label{eq:weighted_noise_variance}
\end{equation}
\end{proof}

\begin{proof}[Proof of Corollary~\ref{cor:signed_exposure_operator}]
The common-marginal HT identities give
$\E\{\psi_a(Z_a)Z_c\}=\Rnorm_{ac}$ for every $a,c$. Applying them row by row to \eqref{eq:signed_exposure_model}, subtracting the all-treated versus all-control contrast, and weighting by $w_a/W$ proves
\eqref{eq:signed_exposure_bias}. For the algebraic identity, expand
\[
    \frac1v\sum_{a,c}w_a\Gamma_{ac}(Z_a-p)Z_c
    =\frac1vZ^\top D_w\Gamma Z
      -\frac1q(\Gamma^\top w)^\top Z.
\]
Combining this with the weighted baseline and direct-effect terms in the proof of Corollary~\ref{cor:weighted_ht_identity} yields
\eqref{eq:signed_exposure_algebra}.
\end{proof}

\begin{proof}[Proof of Proposition~\ref{prop:expected_cut_geometry}]
Because the two Bernoulli variables have the same marginal $p$,
\[
\begin{split}
    \Pr(Z_a\ne Z_c)
    &=\E(Z_a-Z_c)^2
      =2p-2\E(Z_aZ_c)\\
    &=2p-2\{p^2+v\Rnorm_{ac}\}
      =2v(1-\Rnorm_{ac}).
\end{split}
\]
This proves \eqref{eq:disagreement_covariance_identity}. For each realized assignment vector $z$, the map
$d_z(a,c)=\ind\{z_a\ne z_c\}$ is a cut semimetric: in particular,
$\ind\{z_a\ne z_c\}\le
 \ind\{z_a\ne z_b\}+\ind\{z_b\ne z_c\}$.
Expectation and multiplication by the positive constant $(2v)^{-1}$ preserve all semimetric inequalities. Moreover,
$\{D_Z(a,c):a<c\}=\E\{d_z(a,c):a<c\}$ is a probability-weighted convex combination of realized cut vectors, so it lies in the cut polytope. Substitution of
$\Rnorm_{ac}-1=-d_Z(a,c)$ into \eqref{eq:weighted_exact_bias} gives
\eqref{eq:bias_expected_cut}. Nonpositivity follows because all summands are nonnegative when $\theta>0$, $w_a\ge0$, and $H_{ac}\ge0$.
\end{proof}

\begin{proof}[Proof of Corollary~\ref{cor:zero_bias_collapse}]
By \eqref{eq:undirected_cut_bias}, the magnitude of the positive-exposure bias is a sum of nonnegative terms $\theta h_ed_Z(e)/W$. Because every $h_e$ is strictly positive, the sum vanishes if and only if $d_Z(e)=0$ on every exposure edge, which is equivalent to \eqref{eq:zero_bias_edge_agreement}. Along a finite connected path of exposure edges, almost-sure equality is transitive; hence the assignment is almost surely constant on each connected component. On a connected graph there is a Bernoulli random variable $S$ such that $Z=S\one$, and the common marginal condition gives $\Pr(S=1)=p$.

Conditional on the fixed schedule and disturbances, every statistic measurable with respect to this assignment has at most two support points. Under overlap, the two probabilities stay bounded away from zero. Any nondegenerate location--scale normalization therefore remains a two-point law with two nonvanishing atoms and cannot converge weakly to the continuous standard normal distribution.
\end{proof}

\begin{proof}[Proof of Proposition~\ref{prop:path_bias_lower_bound}]
Apply the semimetric triangle inequality successively along the path from $i$ to $i+m$:
\[
    d_Z(i,i+m)\le\sum_{k=i}^{i+m-1}d_Z(k,k+1).
\]
Summing over $i=1,\ldots,n-m$, each adjacent edge appears in at most $m$ of the resulting intervals. This proves
\eqref{eq:path_disagreement_inequality}. If the distance-$m$ correlations are at most $\epsilon_m$, then
$d_Z(i,i+m)=1-\Rnorm_{i,i+m}\ge1-\epsilon_m$; substitution proves
\eqref{eq:path_average_disagreement_bound}. Finally, the positive-exposure bias identity and $H_{k,k+1}\ge h_0$ give
\[
\begin{split}
    |\E\widehat\tau^{HT}-\tau|
    &\ge \frac{\theta_0h_0}{n}
      \sum_{k=1}^{n-1}d_Z(k,k+1)\\
    &\ge \frac{\theta_0h_0}{n}
      \frac{n-m}{m}(1-\epsilon_m),
\end{split}
\]
which is \eqref{eq:path_bias_rate_bound}.
\end{proof}

\begin{proof}[Proof of Proposition~\ref{prop:fractional_multicut_bias}]
Set $x_e=d_Z(e)$ on the exposure graph. For every connected terminal pair $(u,v)\in\mathcal P$ and every exposure-graph path $P:u\leadsto v$, repeated application of the triangle inequality gives
\[
    \sum_{e\in P}x_e\ge d_Z(u,v)\ge\delta_{uv}.
\]
Thus $x$ is feasible for \eqref{eq:fractional_multicut_program}. Its objective equals the exposure-boundary term in \eqref{eq:undirected_cut_bias}. Since the value at any feasible point is at least the minimum of the program, multiplying by $\theta/W$ proves \eqref{eq:fractional_multicut_bias_bound}.
\end{proof}

\begin{proof}[Proof of Proposition~\ref{prop:fractional_multicut_dual}]
The primal program \eqref{eq:fractional_multicut_program} has one nonnegative variable $x_e$ per exposure edge and one inequality for each terminal-pair/simple-path combination. The graph is finite, so the collection of simple paths is finite. Introducing a nonnegative dual variable $f_{uv,P}$ for each path constraint gives the Lagrangian
\[
\sum_{e}h_ex_e+
\sum_{(u,v),P}f_{uv,P}
\left(\delta_{uv}-\sum_{e\in P}x_e\right).
\]
The infimum over $x_e\ge0$ is finite exactly when the total dual load on edge $e$ is at most $h_e$, yielding \eqref{eq:fractional_multicut_dual}. The primal is feasible and bounded below, so finite-dimensional linear-programming strong duality proves equality.

For the edge-disjoint certificate, set $f_{uv,P_{uv}}=h_0$ on the selected path for each terminal pair and set all other dual variables to zero. Edge disjointness implies that every edge carries load at most $h_0\le h_e$. The resulting dual objective is $h_0\sum_{(u,v)}\delta_{uv}$, proving \eqref{eq:edge_disjoint_path_certificate}.
\end{proof}

\begin{proof}[Proof of Corollary~\ref{cor:grid_bias_floor}]
The $n^{d-1}$ coordinate lines joining $(1,j)$ to $(n,j)$ are pairwise edge-disjoint. Apply \eqref{eq:edge_disjoint_path_certificate} with one terminal pair per line, common demand $\delta$, and edge capacity lower bound $h_0$. Then
$\operatorname{MC}\ge h_0\delta n^{d-1}$. Proposition~\ref{prop:fractional_multicut_bias} gives \eqref{eq:grid_bias_floor}. If $\Ocal=\Zcal=[n]^d$ and $W\asymp L=n^d$, division by $W$ yields order $n^{-1}=L^{-1/d}$.
\end{proof}

\begin{proof}[Proof of Corollary~\ref{cor:expander_bias_floor}]
For a realized assignment, let $S=\{a:Z_a=1\}$. Exact balance gives $|S|=pL$, and the weighted expansion definition implies
\[
\sum_{e\in\partial_H S}h_e
\ge\phi_hL\min\{p,1-p\}.
\]
Taking expectation and using
$\E\sum_{e\in\partial_H S}h_e=
\sum_eh_e\Pr(Z\text{ disagrees on }e)$, the expected-cut identity gives
\[
|\E\widehat\tau^{HT}_{w,\tilde\alpha}-\tau_w|
=\frac{\theta}{2p(1-p)W}
\E\sum_{e\in\partial_H S}h_e,
\]
which proves \eqref{eq:expander_bias_floor}.
\end{proof}

\begin{proof}[Proof of Proposition~\ref{prop:switching_carryover_frontier}]
Equation~\eqref{eq:switch_rate} gives
$\Pr(Z_a\ne Z_c)=2v(1-\Rnorm_{ac})$ on every lag edge. Therefore
\[
\begin{split}
    \bar s_{\lag}^{(w)}
    &=\frac{2v}{E_{\lag}^{(w)}}
      \sum_{a,c}w_aH_{ac}^{\lag}(1-\Rnorm_{ac})
      =-\frac{2v}{E_{\lag}^{(w)}}b_{\lag}^{(w)}(\Rnorm).
\end{split}
\]
Rearranging proves both identities in \eqref{eq:switching_carryover_frontier}; the lower-bound statements follow immediately.
\end{proof}

\begin{proof}[Proof of Proposition~\ref{prop:robust_bias_objective}]
Write $\delta=\Lambda^{1/2}u$. The ellipsoidal constraint is $\|u\|_2\le1$, so Cauchy--Schwarz gives
\[
    \sup_{\delta\in\mathcal U_\Lambda}
    |\delta^\top\bar b|
    =\sup_{\|u\|_2\le1}|u^\top\Lambda^{1/2}\bar b|
    =\|\Lambda^{1/2}\bar b\|_2.
\]
Squaring proves \eqref{eq:ellipsoidal_robust_bias}. Under coordinatewise bounds, the triangle inequality gives
$|\delta^\top\bar b|\le
 \bar\delta_1|\bar b_1|+\bar\delta_2|\bar b_2|$.
Choose each $\delta_j$ at its bound with sign matching $\bar b_j$ to attain equality, which proves \eqref{eq:box_robust_bias}.
\end{proof}

\subsection{Proof of Theorem~\ref{thm:variance_envelope}}\label{subsec:variance_envelope_proof}

Let $s\ge \deg_H$ componentwise and suppose
\[
    |g_{p,c}^{\tilde\alpha}|\le\omega_gs_c.
\]
Then, for $Z\ge0$,
\[
    |(g_p^{\tilde\alpha})^\top Z|\le \omega_gs^\top Z.
\]
Moreover, because $H\ge0$ and $Z_a\in\{0,1\}$,
\[
    Z^\top HZ=\sum_{a,c}H_{ac}Z_aZ_c
    \le\sum_{a,c}H_{ac}Z_c=\deg_H^\top Z\le s^\top Z.
\]
Thus
\[
    \left|\frac{\theta}{v}Z^\top HZ\right|
    \le\frac{|\theta|}{v}s^\top Z.
\]
Using $(x+y)^2\le2x^2+2y^2$,
\[
    \left((g_p^{\tilde\alpha})^\top Z+\frac{\theta}{v}Z^\top HZ\right)^2
    \le
    2\left(\omega_g^2+\frac{\theta^2}{v^2}\right)(s^\top Z)^2.
\]
Since $v=pq$, the second term is $\theta^2/(p^2q^2)$. Also $\Var(W)\le\E[W^2]$ and
\[
    \E[ZZ^\top]=\Cov(Z)+\E[Z]\E[Z]^\top=v\Rnorm+p^2\one\one^\top.
\]
Therefore
\[
    \Var\left((g_p^{\tilde\alpha})^\top Z+\frac{\theta}{v}Z^\top HZ\right)
    \le
    2\left(\omega_g^2+\frac{\theta^2}{p^2q^2}\right)
    s^\top(v\Rnorm+p^2\one\one^\top)s.
\]
Combining with the exact variance display and the independent-noise term gives \eqref{eq:variance_bound_main}.

\subsection{\texorpdfstring{General-$p$ Gaussian copula map}{General-p Gaussian copula map}}

Let $(G_a,G_c)$ be standard bivariate normal with correlation $\rho$. With $\tp=\Phi^{-1}(p)$ and $Z_a=\ind\{G_a\le\tp\}$,
\[
    \E[Z_aZ_c]=\Pr(G_a\le\tp,G_c\le\tp)=\Phi_2(\tp,\tp;\rho).
\]
Therefore
\[
    \Rnorm_{ac}=\frac{\Phi_2(\tp,\tp;\rho)-p^2}{p(1-p)}=\Gp(\rho).
\]
By Plackett's identity,
\[
    \frac{\partial}{\partial\rho}\Phi_2(\tp,\tp;\rho)=\varphi_2(\tp,\tp;\rho),
\]
so
\begin{equation}
    \Gp'(\rho)=\frac{\varphi_2(\tp,\tp;\rho)}{p(1-p)}.
    \label{eq:Gp_derivative}
\end{equation}
At equal thresholds,
\[
    \varphi_2(\tp,\tp;\rho)
    =\frac{1}{2\pi\sqrt{1-\rho^2}}
    \exp\left\{-\frac{\tp^2}{1+\rho}\right\}.
\]
For $p=1/2$, $\tp=0$ and
\[
    \Phi_2(0,0;\rho)=\frac14+\frac{1}{2\pi}\arcsin(\rho),
\]
which yields
\[
    \mathcal{G}_{1/2}(\rho)=\frac{2}{\pi}\arcsin(\rho).
\]

\subsection{Bounds on feasible normalized covariance}

For any two Bernoulli-$p$ random variables $Z_a,Z_c$,
\[
    \max\{0,2p-1\}\le \E[Z_aZ_c]\le p.
\]
Therefore
\[
    \Rnorm_{ac}=\frac{\E[Z_aZ_c]-p^2}{p(1-p)}\le\frac{p-p^2}{p(1-p)}=1.
\]
If $p\le1/2$, then $\E[Z_aZ_c]\ge0$, so
\[
    \Rnorm_{ac}\ge-\frac{p^2}{p(1-p)}=-\frac{p}{1-p}.
\]
If $p>1/2$, then $\E[Z_aZ_c]\ge2p-1$, so
\[
    \Rnorm_{ac}\ge\frac{2p-1-p^2}{p(1-p)}=-\frac{1-p}{p}.
\]
Thus
\[
    -\min\left\{\frac{p}{1-p},\frac{1-p}{p}\right\}\le\Rnorm_{ac}\le1.
\]

\section{Additional Discussion}
\label{app:additional_discussion}

\subsection{Primitive observations and unit--block outcomes}
\label{app:primitive_observations}

The theory only requires a potential outcome $Y_{ib}(Z)$ for every evaluated unit--block cell, but many platform applications generate these outcomes by aggregating finer observations. Let $\mathcal U_b$ denote the set of user, session, request, location, or opportunity units observed in block $b$, and let $S_{uib}\ge0$ be a pre-treatment exposure or opportunity weight for opportunity $u$ and experimental unit $i$ in block $b$. The unit--block traffic weight is
\[
    m_{ib}=\sum_{u\in\mathcal U_b}S_{uib}.
\]
If $Q_{uib}(Z)$ is the primitive response contribution, such as click, conversion, revenue, sales, or service completion, the unit--block outcome can be written as
\begin{equation}
    Y_{ib}(Z)=\frac{1}{m_{ib}\vee1}\sum_{u\in\mathcal U_b}S_{uib}Q_{uib}(Z).
    \label{eq:item_block_outcome}
\end{equation}
This layer is important in the recommendation experiments because user demand induces both traffic weights and an empirical exposure graph; users need not be treatment-assignment units.

A generic microfoundation is
\begin{equation}
    Q_{uib}(Z)=r^0_{uib}+\Delta_{uib}Z_{ib}
    +\gamma\sum_{j\ne i}\kappa_{ij}S_{ujb}Z_{jb}
    +\eta_{uib}Z_{i,b-1}+\xi_{uib},
    \label{eq:user_micro_model}
\end{equation}
where $\kappa_{ij}$ measures similarity, substitutability, exposure overlap, geographic proximity, or any other pre-treatment relation. Aggregating \eqref{eq:user_micro_model} through \eqref{eq:item_block_outcome} yields a reduced-form unit--block model with a direct effect, contemporaneous network spillover, and lagged own-unit carryover.

\subsection{\texorpdfstring{General-$p$ HT score and fixed baseline adjustment}{General-p HT score and fixed baseline adjustment}}

This discussion records the HT score and fixed adjustment used in the main text. For common marginal treatment probability $p$, define the general-$p$ HT score
\begin{equation}
    \psi_a(Z_a)=\frac{Z_a}{p}-\frac{1-Z_a}{1-p}=\frac{Z_a-p}{p(1-p)}=\frac{Z_a-p}{v}.
    \label{eq:psi_def}
\end{equation}
Let $\tilde\alpha_a$ be any fixed baseline adjustment, such as a historical forecast or pre-experiment prediction. The baseline-adjusted HT estimator is
\begin{equation}
    \widehat\tau_{\tilde\alpha}=\frac{1}{M}\sum_{a\in\Ocal}\psi_a(Z_a)\{Y_a^{\mathrm{obs}}-\tilde\alpha_a\}.
    \label{eq:estimator}
\end{equation}
The adjustment $\tilde\alpha_a$ need not equal the structural baseline in the potential outcome model. The structural baseline $\mu_a$ is a conceptual component of the potential-outcome schedule: it is the part of the outcome that would remain under the working model before direct treatment and exposure terms are added. The fixed adjustment $\tilde\alpha_a$ is instead an observable pre-treatment prediction or forecast that the analyst subtracts from the observed outcome. Because $\tilde\alpha_a$ is fixed before assignment and $\E[\psi_a]=0$, it cannot create bias in the HT identity; its role is to reduce variance when it approximates $\mu_a$.

The traffic-weighted empirical version is
\[
    \widehat\tau_{w,\tilde\alpha}=\frac{1}{W_N}\sum_{a\in\Ocal}w_a\psi_a(Z_a)\{Y_a^{\mathrm{obs}}-\tilde\alpha_a\},
\]
where $w_a$ is fixed before treatment assignment. In simulations we report both the unweighted and traffic-weighted estimators when unit traffic is highly skewed.

The key HT identities are
\begin{equation}
    \E[\psi_a]=0,
    \qquad
    \E[\psi_aZ_a]=1,
    \qquad
    \E[\psi_aZ_c]=\Rnorm_{ac}.
    \label{eq:basic_identities}
\end{equation}

\subsection{Network-time jackknife}

A practical alternative to \eqref{eq:model_assisted_hac} is a network-time buffered jackknife. Partition the network-time graph into blocks $\mathcal C_1,\ldots,\mathcal C_J$ of bounded diameter or use graph clusters crossed with temporal blocks. For each $j$, remove or downweight $\mathcal C_j$ together with a buffer neighborhood and recompute the estimator, yielding $\widehat\tau^{(-j)}$. The jackknife variance
\[
    \widehat V_{PJ}=\frac{J-1}{J}\sum_{j=1}^J
    \left(\widehat\tau^{(-j)}-\frac{1}{J}\sum_{\ell=1}^J\widehat\tau^{(-\ell)}\right)^2
\]
can be viewed as a recomputation-based variance diagnostic. This connects naturally to recent Neyman-jackknife ideas for design-based inference under interference \citep{parkwager2026jackknife}. The network-time HAC estimator remains the main analytic object because its truncation and stochastic errors can be matched directly to the graph weak-dependence rates. The buffered jackknife is an implementation diagnostic; a formal result would additionally require deletion-block growth, buffer, and remainder conditions, so the displayed formula is not asserted to be automatically consistent.

\end{document}

%% file: tables/tab_default_retail.tex
\begin{tabular}{l ccc c ccc c ccc}
\toprule
 & \multicolumn{3}{c}{Linear} && \multicolumn{3}{c}{Nonlinear} && \multicolumn{3}{c}{Demand} \\
\cmidrule{2-4}\cmidrule{6-8}\cmidrule{10-12}
Design & Bias & Std. & RMSE && Bias & Std. & RMSE && Bias & Std. & RMSE \\
\midrule
Independent & -0.183 & 0.025 & 0.185 && -0.237 & 0.024 & 0.239 && -0.122 & 0.026 & 0.125 \\
RBSD & -0.195 & 0.020 & 0.196 && -0.255 & 0.020 & 0.255 && -0.165 & 0.018 & 0.166 \\
Cluster & -0.115 & 0.043 & 0.123 && -0.159 & 0.047 & 0.166 && -0.187 & 0.051 & 0.194 \\
Static-OCD & -0.115 & 0.064 & 0.131 && -0.160 & 0.066 & 0.173 && -0.186 & 0.066 & 0.197 \\
COSTA & -0.036 & 0.055 & 0.066 && -0.048 & 0.052 & 0.071 && -0.015 & 0.054 & 0.056 \\
COSTA-LP & -0.046 & 0.057 & 0.073 && -0.072 & 0.056 & 0.091 && -0.027 & 0.056 & 0.062 \\
\bottomrule
\end{tabular}

%% file: tables/tab_default_movielens.tex
\begin{tabular}{l ccc c ccc c ccc}
\toprule
 & \multicolumn{3}{c}{Linear} && \multicolumn{3}{c}{Nonlinear} && \multicolumn{3}{c}{Demand} \\
\cmidrule{2-4}\cmidrule{6-8}\cmidrule{10-12}
Design & Bias & Std. & RMSE && Bias & Std. & RMSE && Bias & Std. & RMSE \\
\midrule
Independent & -0.182 & 0.034 & 0.185 && -0.233 & 0.035 & 0.235 && -0.118 & 0.035 & 0.123 \\
RBSD & -0.196 & 0.026 & 0.198 && -0.255 & 0.023 & 0.256 && -0.162 & 0.024 & 0.164 \\
Cluster & -0.107 & 0.091 & 0.141 && -0.161 & 0.090 & 0.185 && -0.191 & 0.103 & 0.217 \\
Static-OCD & -0.112 & 0.113 & 0.159 && -0.160 & 0.119 & 0.200 && -0.198 & 0.123 & 0.233 \\
COSTA & -0.030 & 0.097 & 0.102 && -0.055 & 0.096 & 0.111 && -0.025 & 0.105 & 0.108 \\
COSTA-LP & -0.039 & 0.092 & 0.100 && -0.076 & 0.090 & 0.117 && -0.036 & 0.092 & 0.098 \\
\bottomrule
\end{tabular}

%% file: tables/tab_estimator_retail.tex
\begin{tabular}{lllrrr}
\toprule
Model & Estimator & Design & Bias & Std. & RMSE \\
\midrule
Linear & HT & COSTA & -0.023 & 0.109 & 0.112 \\
Linear & HT & COSTA-LP & -0.039 & 0.079 & 0.088 \\
Linear & Hájek & COSTA & -0.036 & 0.055 & 0.066 \\
Linear & Hájek & COSTA-LP & -0.046 & 0.057 & 0.073 \\
Linear & DIM & COSTA & -0.065 & 0.040 & 0.076 \\
Linear & DIM & COSTA-LP & -0.075 & 0.044 & 0.087 \\
\addlinespace[0.25em]
Nonlinear & HT & COSTA & -0.052 & 0.113 & 0.124 \\
Nonlinear & HT & COSTA-LP & -0.075 & 0.084 & 0.112 \\
Nonlinear & Hájek & COSTA & -0.048 & 0.052 & 0.071 \\
Nonlinear & Hájek & COSTA-LP & -0.072 & 0.056 & 0.091 \\
Nonlinear & DIM & COSTA & -0.081 & 0.037 & 0.089 \\
Nonlinear & DIM & COSTA-LP & -0.103 & 0.047 & 0.113 \\
\addlinespace[0.25em]
Demand & HT & COSTA & -0.019 & 0.099 & 0.101 \\
Demand & HT & COSTA-LP & -0.030 & 0.077 & 0.082 \\
Demand & Hájek & COSTA & -0.015 & 0.054 & 0.056 \\
Demand & Hájek & COSTA-LP & -0.027 & 0.056 & 0.062 \\
Demand & DIM & COSTA & -0.043 & 0.037 & 0.057 \\
Demand & DIM & COSTA-LP & -0.055 & 0.044 & 0.070 \\
\bottomrule
\end{tabular}

%% file: tables/tab_cluster_kro_retail.tex
\begin{tabular}{llrrr}
\toprule
Model & Design & Bias & Std. & RMSE \\
\midrule
Linear & Cluster & -0.109 & 0.059 & 0.124 \\
Linear & COSTA & -0.035 & 0.058 & 0.068 \\
Linear & COSTA cluster no-kro & -0.089 & 0.449 & 0.457 \\
Linear & COSTA cluster kro & -0.018 & 0.220 & 0.221 \\
\addlinespace[0.25em]
Nonlinear & Cluster & -0.153 & 0.064 & 0.166 \\
Nonlinear & COSTA & -0.051 & 0.059 & 0.078 \\
Nonlinear & COSTA cluster no-kro & -0.066 & 0.438 & 0.443 \\
Nonlinear & COSTA cluster kro & -0.028 & 0.220 & 0.222 \\
\addlinespace[0.25em]
Demand & Cluster & -0.184 & 0.068 & 0.196 \\
Demand & COSTA & -0.020 & 0.058 & 0.061 \\
Demand & COSTA cluster no-kro & -0.079 & 0.423 & 0.430 \\
Demand & COSTA cluster kro & -0.022 & 0.222 & 0.223 \\
\bottomrule
\end{tabular}

%% file: tables/tab_inference_validation.tex
\resizebox{\linewidth}{!}{%
\begin{tabular}{@{}lllcccccc@{}}
\toprule
Dataset & Outcome model & Design & DC cov. & GATE cov. & RC cov. & Width & Width/MC & SE/SD \\
\midrule
RetailRocket & Linear & COSTA & 0.920 & 0.843 & 0.920 & 0.231 & 1.048 & 1.015 \\
 &  & COSTA-LocalPenalty & 0.904 & 0.713 & 0.904 & 0.144 & 1.023 & 1.001 \\
\addlinespace[2pt]
 & Nonlinear & COSTA & 0.903 & 0.851 & 0.903 & 0.294 & 1.006 & 1.002 \\
 &  & COSTA-LocalPenalty & 0.920 & 0.696 & 0.920 & 0.198 & 1.052 & 1.063 \\
\addlinespace[2pt]
 & Demand & COSTA & 0.899 & 0.882 & 0.899 & 0.191 & 0.992 & 0.994 \\
 &  & COSTA-LocalPenalty & 0.891 & 0.788 & 0.891 & 0.122 & 0.971 & 0.986 \\
\addlinespace[2pt]
MovieLens & Linear & COSTA & 0.922 & 0.909 & 0.922 & 0.394 & 1.058 & 1.030 \\
 &  & COSTA-LocalPenalty & 0.897 & 0.841 & 0.897 & 0.267 & 0.991 & 0.981 \\
\addlinespace[2pt]
 & Nonlinear & COSTA & 0.913 & 0.897 & 0.913 & 0.474 & 1.034 & 0.999 \\
 &  & COSTA-LocalPenalty & 0.932 & 0.842 & 0.932 & 0.356 & 1.100 & 1.078 \\
\addlinespace[2pt]
 & Demand & COSTA & 0.900 & 0.899 & 0.900 & 0.313 & 1.014 & 0.968 \\
 &  & COSTA-LocalPenalty & 0.888 & 0.864 & 0.888 & 0.229 & 0.967 & 0.991 \\
\bottomrule
\end{tabular}%
}

%% file: tables/tab_inference_studentized.tex
\resizebox{\linewidth}{!}{%
\begin{tabular}{@{}lllrrrrrr@{}}
\toprule
Dataset & Outcome model & Design & Mean & SD & $q_{0.05}$ & $q_{.50}$ & $q_{0.95}$ & Trunc. \\
\midrule
RetailRocket & Linear & COSTA & -0.026 & 0.985 & -1.534 & -0.014 & 1.605 & 0.000 \\
 &  & COSTA-LocalPenalty & 0.045 & 0.999 & -1.517 & 0.009 & 1.693 & 0.000 \\
\addlinespace[2pt]
 & Nonlinear & COSTA & -0.008 & 0.998 & -1.604 & -0.016 & 1.668 & 0.000 \\
 &  & COSTA-LocalPenalty & 0.035 & 0.941 & -1.573 & 0.056 & 1.559 & 0.000 \\
\addlinespace[2pt]
 & Demand & COSTA & -0.005 & 1.006 & -1.726 & -0.019 & 1.593 & 0.000 \\
 &  & COSTA-LocalPenalty & -0.013 & 1.014 & -1.699 & 0.005 & 1.695 & 0.000 \\
\addlinespace[2pt]
MovieLens & Linear & COSTA & -0.066 & 0.971 & -1.563 & -0.106 & 1.545 & 0.000 \\
 &  & COSTA-LocalPenalty & -0.056 & 1.018 & -1.644 & -0.078 & 1.672 & 0.000 \\
\addlinespace[2pt]
 & Nonlinear & COSTA & 0.062 & 1.000 & -1.447 & 0.002 & 1.733 & 0.000 \\
 &  & COSTA-LocalPenalty & -0.025 & 0.927 & -1.432 & -0.087 & 1.552 & 0.000 \\
\addlinespace[2pt]
 & Demand & COSTA & 0.083 & 1.033 & -1.436 & 0.034 & 1.806 & 0.000 \\
 &  & COSTA-LocalPenalty & -0.003 & 1.009 & -1.644 & -0.070 & 1.761 & 0.000 \\
\bottomrule
\end{tabular}%
}

%% file: tables/tab_retail_item_count_rmse.tex
\begin{tabular}{llrrrr}
\toprule
Model & Design & 5000 & 10000 & 20000 & 50000 \\
\midrule
Linear & Independent & 0.186 & 0.186 & 0.184 & 0.182 \\
Linear & RBSD & 0.220 & 0.215 & 0.216 & 0.214 \\
Linear & Cluster & 0.126 & 0.124 & 0.123 & 0.121 \\
Linear & Static-OCD & 0.173 & 0.143 & 0.139 & 0.127 \\
Linear & COSTA & 0.088 & 0.068 & 0.054 & 0.051 \\
Linear & COSTA-LP & 0.093 & 0.081 & 0.074 & 0.076 \\
\addlinespace[0.25em]
Nonlinear & Independent & 0.231 & 0.228 & 0.228 & 0.227 \\
Nonlinear & RBSD & 0.274 & 0.271 & 0.271 & 0.273 \\
Nonlinear & Cluster & 0.160 & 0.166 & 0.158 & 0.152 \\
Nonlinear & Static-OCD & 0.193 & 0.182 & 0.160 & 0.158 \\
Nonlinear & COSTA & 0.094 & 0.078 & 0.062 & 0.055 \\
Nonlinear & COSTA-LP & 0.112 & 0.110 & 0.102 & 0.098 \\
\addlinespace[0.25em]
Demand & Independent & 0.130 & 0.126 & 0.125 & 0.122 \\
Demand & RBSD & 0.225 & 0.224 & 0.223 & 0.219 \\
Demand & Cluster & 0.197 & 0.196 & 0.198 & 0.201 \\
Demand & Static-OCD & 0.234 & 0.210 & 0.204 & 0.197 \\
Demand & COSTA & 0.085 & 0.061 & 0.050 & 0.045 \\
Demand & COSTA-LP & 0.089 & 0.080 & 0.076 & 0.068 \\
\bottomrule
\end{tabular}

%% file: tables/tab_movielens_item_count_rmse.tex
\begin{tabular}{llrrrr}
\toprule
Model & Design & 5000 & 10000 & 20000 & 50000 \\
\midrule
Linear & Independent & 0.191 & 0.189 & 0.190 & 0.190 \\
Linear & RBSD & 0.217 & 0.216 & 0.219 & 0.219 \\
Linear & Cluster & 0.157 & 0.172 & 0.162 & 0.162 \\
Linear & Static-OCD & 0.177 & 0.197 & 0.184 & 0.184 \\
Linear & COSTA & 0.103 & 0.108 & 0.104 & 0.104 \\
Linear & COSTA-LP & 0.114 & 0.121 & 0.126 & 0.126 \\
\addlinespace[0.25em]
Nonlinear & Independent & 0.233 & 0.231 & 0.236 & 0.236 \\
Nonlinear & RBSD & 0.271 & 0.270 & 0.271 & 0.271 \\
Nonlinear & Cluster & 0.176 & 0.192 & 0.194 & 0.194 \\
Nonlinear & Static-OCD & 0.190 & 0.211 & 0.222 & 0.222 \\
Nonlinear & COSTA & 0.111 & 0.125 & 0.110 & 0.110 \\
Nonlinear & COSTA-LP & 0.131 & 0.143 & 0.141 & 0.141 \\
\addlinespace[0.25em]
Demand & Independent & 0.131 & 0.124 & 0.130 & 0.130 \\
Demand & RBSD & 0.224 & 0.224 & 0.224 & 0.224 \\
Demand & Cluster & 0.218 & 0.241 & 0.229 & 0.229 \\
Demand & Static-OCD & 0.237 & 0.255 & 0.260 & 0.260 \\
Demand & COSTA & 0.106 & 0.108 & 0.109 & 0.109 \\
Demand & COSTA-LP & 0.114 & 0.133 & 0.128 & 0.128 \\
\bottomrule
\end{tabular}

%% file: tables/tab_retail_topk_rmse.tex
\begin{tabular}{llrrr}
\toprule
Model & Design & 5 & 10 & 20 \\
\midrule
Linear & Independent & 0.185 & 0.186 & 0.186 \\
Linear & RBSD & 0.215 & 0.215 & 0.215 \\
Linear & Cluster & 0.122 & 0.124 & 0.129 \\
Linear & Static-OCD & 0.127 & 0.143 & 0.157 \\
Linear & COSTA & 0.063 & 0.068 & 0.071 \\
Linear & COSTA-LP & 0.073 & 0.081 & 0.084 \\
\addlinespace[0.25em]
Nonlinear & Independent & 0.228 & 0.228 & 0.228 \\
Nonlinear & RBSD & 0.271 & 0.271 & 0.271 \\
Nonlinear & Cluster & 0.160 & 0.166 & 0.168 \\
Nonlinear & Static-OCD & 0.162 & 0.182 & 0.197 \\
Nonlinear & COSTA & 0.069 & 0.078 & 0.083 \\
Nonlinear & COSTA-LP & 0.100 & 0.110 & 0.116 \\
\addlinespace[0.25em]
Demand & Independent & 0.126 & 0.126 & 0.126 \\
Demand & RBSD & 0.224 & 0.224 & 0.224 \\
Demand & Cluster & 0.199 & 0.196 & 0.199 \\
Demand & Static-OCD & 0.199 & 0.210 & 0.222 \\
Demand & COSTA & 0.055 & 0.061 & 0.069 \\
Demand & COSTA-LP & 0.074 & 0.080 & 0.083 \\
\bottomrule
\end{tabular}

%% file: tables/tab_movielens_topk_rmse.tex
\begin{tabular}{llrrr}
\toprule
Model & Design & 5 & 10 & 20 \\
\midrule
Linear & Independent & 0.189 & 0.189 & 0.189 \\
Linear & RBSD & 0.216 & 0.216 & 0.216 \\
Linear & Cluster & 0.164 & 0.172 & 0.176 \\
Linear & Static-OCD & 0.177 & 0.197 & 0.213 \\
Linear & COSTA & 0.093 & 0.108 & 0.118 \\
Linear & COSTA-LP & 0.111 & 0.121 & 0.133 \\
\addlinespace[0.25em]
Nonlinear & Independent & 0.231 & 0.231 & 0.231 \\
Nonlinear & RBSD & 0.270 & 0.270 & 0.270 \\
Nonlinear & Cluster & 0.181 & 0.192 & 0.199 \\
Nonlinear & Static-OCD & 0.196 & 0.211 & 0.220 \\
Nonlinear & COSTA & 0.106 & 0.125 & 0.136 \\
Nonlinear & COSTA-LP & 0.130 & 0.143 & 0.151 \\
\addlinespace[0.25em]
Demand & Independent & 0.124 & 0.124 & 0.124 \\
Demand & RBSD & 0.224 & 0.224 & 0.224 \\
Demand & Cluster & 0.232 & 0.241 & 0.235 \\
Demand & Static-OCD & 0.240 & 0.255 & 0.269 \\
Demand & COSTA & 0.102 & 0.108 & 0.123 \\
Demand & COSTA-LP & 0.121 & 0.133 & 0.136 \\
\bottomrule
\end{tabular}

%% file: tables/tab_retail_p_rmse.tex
\begin{tabular}{llrrr}
\toprule
Model & Design & 0.1 & 0.2 & 0.5 \\
\midrule
Linear & Independent & 0.188 & 0.186 & 0.186 \\
Linear & RBSD & 0.203 & 0.209 & 0.215 \\
Linear & Cluster & 0.148 & 0.136 & 0.124 \\
Linear & Static-OCD & -- & -- & 0.143 \\
Linear & COSTA & 0.146 & 0.097 & 0.068 \\
Linear & COSTA-LP & 0.130 & 0.090 & 0.081 \\
\addlinespace[0.25em]
Nonlinear & Independent & 0.247 & 0.243 & 0.228 \\
Nonlinear & RBSD & 0.265 & 0.275 & 0.271 \\
Nonlinear & Cluster & 0.189 & 0.169 & 0.166 \\
Nonlinear & Static-OCD & -- & -- & 0.182 \\
Nonlinear & COSTA & 0.165 & 0.117 & 0.078 \\
Nonlinear & COSTA-LP & 0.139 & 0.110 & 0.110 \\
\addlinespace[0.25em]
Demand & Independent & 0.137 & 0.131 & 0.126 \\
Demand & RBSD & 0.166 & 0.199 & 0.224 \\
Demand & Cluster & 0.218 & 0.209 & 0.196 \\
Demand & Static-OCD & -- & -- & 0.210 \\
Demand & COSTA & 0.192 & 0.100 & 0.061 \\
Demand & COSTA-LP & 0.141 & 0.083 & 0.080 \\
\bottomrule
\end{tabular}

%% file: tables/tab_movielens_p_rmse.tex
\begin{tabular}{llrrr}
\toprule
Model & Design & 0.1 & 0.2 & 0.5 \\
\midrule
Linear & Independent & 0.203 & 0.195 & 0.189 \\
Linear & RBSD & 0.203 & 0.213 & 0.216 \\
Linear & Cluster & 0.246 & 0.200 & 0.172 \\
Linear & Static-OCD & -- & -- & 0.197 \\
Linear & COSTA & 0.244 & 0.158 & 0.108 \\
Linear & COSTA-LP & 0.187 & 0.138 & 0.121 \\
\addlinespace[0.25em]
Nonlinear & Independent & 0.256 & 0.249 & 0.231 \\
Nonlinear & RBSD & 0.265 & 0.275 & 0.270 \\
Nonlinear & Cluster & 0.255 & 0.223 & 0.192 \\
Nonlinear & Static-OCD & -- & -- & 0.211 \\
Nonlinear & COSTA & 0.292 & 0.181 & 0.125 \\
Nonlinear & COSTA-LP & 0.203 & 0.161 & 0.143 \\
\addlinespace[0.25em]
Demand & Independent & 0.141 & 0.129 & 0.124 \\
Demand & RBSD & 0.168 & 0.200 & 0.224 \\
Demand & Cluster & 0.287 & 0.259 & 0.241 \\
Demand & Static-OCD & -- & -- & 0.255 \\
Demand & COSTA & 0.276 & 0.175 & 0.108 \\
Demand & COSTA-LP & 0.211 & 0.148 & 0.133 \\
\bottomrule
\end{tabular}

%% file: tables/tab_retail_block_rmse.tex
\begin{tabular}{llrrr}
\toprule
Model & Design & 4 & 8 & 16 \\
\midrule
Linear & Independent & 0.186 & 0.185 & 0.181 \\
Linear & RBSD & 0.215 & 0.196 & 0.188 \\
Linear & Cluster & 0.124 & 0.123 & 0.119 \\
Linear & Static-OCD & 0.143 & 0.131 & 0.122 \\
Linear & COSTA & 0.068 & 0.066 & 0.055 \\
Linear & COSTA-LP & 0.081 & 0.073 & 0.056 \\
\addlinespace[0.25em]
Nonlinear & Independent & 0.228 & 0.239 & 0.239 \\
Nonlinear & RBSD & 0.271 & 0.255 & 0.250 \\
Nonlinear & Cluster & 0.166 & 0.166 & 0.168 \\
Nonlinear & Static-OCD & 0.182 & 0.173 & 0.168 \\
Nonlinear & COSTA & 0.078 & 0.071 & 0.070 \\
Nonlinear & COSTA-LP & 0.110 & 0.091 & 0.082 \\
\addlinespace[0.25em]
Demand & Independent & 0.126 & 0.125 & 0.123 \\
Demand & RBSD & 0.224 & 0.166 & 0.143 \\
Demand & Cluster & 0.196 & 0.194 & 0.188 \\
Demand & Static-OCD & 0.210 & 0.197 & 0.193 \\
Demand & COSTA & 0.061 & 0.056 & 0.051 \\
Demand & COSTA-LP & 0.080 & 0.062 & 0.047 \\
\bottomrule
\end{tabular}

%% file: tables/tab_movielens_block_rmse.tex
\begin{tabular}{llrrr}
\toprule
Model & Design & 4 & 8 & 16 \\
\midrule
Linear & Independent & 0.189 & 0.185 & 0.182 \\
Linear & RBSD & 0.216 & 0.198 & 0.188 \\
Linear & Cluster & 0.172 & 0.141 & 0.129 \\
Linear & Static-OCD & 0.197 & 0.159 & 0.138 \\
Linear & COSTA & 0.108 & 0.102 & 0.094 \\
Linear & COSTA-LP & 0.121 & 0.100 & 0.083 \\
\addlinespace[0.25em]
Nonlinear & Independent & 0.231 & 0.235 & 0.241 \\
Nonlinear & RBSD & 0.270 & 0.256 & 0.249 \\
Nonlinear & Cluster & 0.192 & 0.185 & 0.169 \\
Nonlinear & Static-OCD & 0.211 & 0.200 & 0.180 \\
Nonlinear & COSTA & 0.125 & 0.111 & 0.098 \\
Nonlinear & COSTA-LP & 0.143 & 0.117 & 0.098 \\
\addlinespace[0.25em]
Demand & Independent & 0.124 & 0.123 & 0.121 \\
Demand & RBSD & 0.224 & 0.164 & 0.140 \\
Demand & Cluster & 0.241 & 0.217 & 0.205 \\
Demand & Static-OCD & 0.255 & 0.233 & 0.209 \\
Demand & COSTA & 0.108 & 0.108 & 0.088 \\
Demand & COSTA-LP & 0.133 & 0.098 & 0.076 \\
\bottomrule
\end{tabular}

%% file: tables/tab_cluster_kro_movielens.tex
\begin{tabular}{llrrr}
\toprule
Model & Design & Bias & Std. & RMSE \\
\midrule
Linear & Cluster & -0.119 & 0.124 & 0.172 \\
Linear & COSTA & -0.031 & 0.104 & 0.108 \\
Linear & COSTA cluster no-kro & -0.152 & 0.284 & 0.321 \\
Linear & COSTA cluster kro & -0.054 & 0.204 & 0.211 \\
\addlinespace[0.25em]
Nonlinear & Cluster & -0.149 & 0.121 & 0.192 \\
Nonlinear & COSTA & -0.053 & 0.113 & 0.125 \\
Nonlinear & COSTA cluster no-kro & -0.164 & 0.305 & 0.346 \\
Nonlinear & COSTA cluster kro & -0.057 & 0.209 & 0.216 \\
\addlinespace[0.25em]
Demand & Cluster & -0.200 & 0.135 & 0.241 \\
Demand & COSTA & -0.023 & 0.106 & 0.108 \\
Demand & COSTA cluster no-kro & -0.115 & 0.293 & 0.314 \\
Demand & COSTA cluster kro & -0.025 & 0.212 & 0.213 \\
\bottomrule
\end{tabular}

%% file: tables/tab_estimator_movielens.tex
\begin{tabular}{lllrrr}
\toprule
Model & Estimator & Design & Bias & Std. & RMSE \\
\midrule
Linear & HT & COSTA & -0.024 & 0.171 & 0.173 \\
Linear & HT & COSTA-LP & -0.032 & 0.140 & 0.143 \\
Linear & Hájek & COSTA & -0.030 & 0.097 & 0.102 \\
Linear & Hájek & COSTA-LP & -0.039 & 0.092 & 0.100 \\
Linear & DIM & COSTA & -0.086 & 0.049 & 0.099 \\
Linear & DIM & COSTA-LP & -0.101 & 0.048 & 0.111 \\
\addlinespace[0.25em]
Nonlinear & HT & COSTA & -0.049 & 0.184 & 0.190 \\
Nonlinear & HT & COSTA-LP & -0.069 & 0.146 & 0.161 \\
Nonlinear & Hájek & COSTA & -0.055 & 0.096 & 0.111 \\
Nonlinear & Hájek & COSTA-LP & -0.076 & 0.090 & 0.117 \\
Nonlinear & DIM & COSTA & -0.106 & 0.051 & 0.117 \\
Nonlinear & DIM & COSTA-LP & -0.129 & 0.048 & 0.138 \\
\addlinespace[0.25em]
Demand & HT & COSTA & -0.020 & 0.161 & 0.162 \\
Demand & HT & COSTA-LP & -0.033 & 0.126 & 0.130 \\
Demand & Hájek & COSTA & -0.025 & 0.105 & 0.108 \\
Demand & Hájek & COSTA-LP & -0.036 & 0.092 & 0.098 \\
Demand & DIM & COSTA & -0.064 & 0.054 & 0.084 \\
Demand & DIM & COSTA-LP & -0.073 & 0.054 & 0.091 \\
\bottomrule
\end{tabular}